\documentclass{article}
\usepackage{authblk}
\usepackage[utf8]{inputenc}
\usepackage{color,soul}
\usepackage{graphicx}
\usepackage{appendix}
\usepackage{amsmath}
\usepackage{amssymb}
\usepackage{amsthm}
\usepackage{comment}
\usepackage[ruled,vlined,linesnumbered]{algorithm2e}
\usepackage[margin=3cm]{geometry}

\usepackage{setspace}
\usepackage[section]{placeins}

\makeatletter
\AtBeginDocument{%
  \expandafter\renewcommand\expandafter\subsection\expandafter{%
    \expandafter\@fb@secFB\subsection
  }%
}
\makeatother
\makeatletter
\AtBeginDocument{%
  \expandafter\renewcommand\expandafter\subsubsection\expandafter{%
    \expandafter\@fb@secFB\subsubsection
  }%
}
\makeatother

\graphicspath{ {tex/figures/} }

\title{Modeling the subjective perspective of consciousness and its role in the control of behaviours}

\author[1,2,3]{Rudrauf, D.\thanks{Corresponding author: david.rudrauf@unige.ch}}
\author[2]{Sergeant-Perthuis, G.}
\author[2,4]{Belli, O.}
\author[1,2]{Tisserand, Y.}
\author[5,3]{Di Marzo Serugendo, G.}

\affil[1]{FPSE, Section of Psychology, University of Geneva, Geneva, Switzerland}
\affil[2]{Swiss Center for Affective Sciences, University of Geneva, Geneva, Switzerland}
\affil[3]{Computer Science University Center, University of Geneva, Geneva, Switzerland}
\affil[4]{Evolutio, Geneva, Switzerland}
\affil[5]{SDS, University of Geneva, Geneva, Switzerland}
\date{July 3 2021}
\theoremstyle{definition}
\newtheorem{defn}{Definition}[section]
\newtheorem*{defn*}{Definition}

\theoremstyle{plain}
\newtheorem{thm}{Theorem}[section]
\newtheorem{prop}{Proposition}[section]
\newtheorem*{prop*}{Proposition}

\newtheorem*{theo*}{Theorem}
 
\theoremstyle{remark}

\newtheorem{rem}{Remark}[section]
\newtheorem{ex}{Example}[section]

\DeclareMathOperator{\argmin}{\operatorname{\mathbf{argmin}}}

\newcommand{\R}{\mathbb{R}}
\newcommand{\E}{\mathbb{E}}

\DeclareMathOperator{\DKL}{\operatorname{DKL}}
\DeclareMathOperator{\FE}{\operatorname{FE}}

\newcommand{\N}{\mathbb{N}}
\newcommand{\Path}{\mathcal{P}}
\newcommand{\update}{\textit{Pred}}
\newcommand{\scb}{\textit{Scb}}
\begin{document}

\maketitle

\newpage

\section*{Highlights}

\begin{itemize}

\item Novel operationalisation of the role of the subjective perspective of consciousness in the control of behaviour.

\item Novel account of appraisal, drive and theory of mind based on that operationalization.

\item Model accounts for psychophysical relationships between appraisal and distance.

\item Model generates adaptive and maladaptive behaviours as a function of interpretable parameters.

\end{itemize}

\newpage

\tableofcontents

\newpage

\begin{abstract}

Consciousness has been hypothesized to operate as a global workspace, which accesses and integrates multimodal information in a unified manner, supports expectation violation monitoring and reduction, and the motivation, programming and control of action. One important yet open issue concerns how the subjective perspective at the core of consciousness, and subjective properties of manifestation of the environment in such perspective as an embodied experience, plays a role in such process. We operationalised the concept of subjective perspective using the principles of the Projective Consciousness Model (PCM), based on the projective geometrical concept of Field of Consciousness. We show how these principles can account for documented relationships between appraisal and distance as an inverse distance law, yield a generative model of affective and epistemic drives based on purely subjective parameters, such as the apparent size of objects, and can be generalised to implement Theory of Mind, in a manner that is consistent with simulation theory. We used simulations of artificial agents, based on psychological rationale, to demonstrate how different model parameters could generate a variety of emergent adaptive and maladaptive behaviours that are relevant to developmental and clinical psychology: the ability to be resilient in the face of obstacles through imaginary projections, the emergence of social approach and joint attention behaviours, the ability to take advantage of false beliefs attributed to others, the emergence of avoidance behaviours as observed in social anxiety disorders, the presence of restricted interests as observed in autism spectrum disorders. The simulation of agents was applied to a specific robotic context, and agents' behaviours were demonstrated by controlling the corresponding robots.

\end{abstract}

\bigskip

\section{Introduction}

The cyberneticist Valentino Braitenberg, a pioneer in synthetic psychology, used simple toy vehicles as models of animal navigation (the so-called ``Braintenberg vehicles"), to demonstrate how basic sensorimotor rules, can give rise to the emergence of relatively complex behaviours \cite{braitenberg1}. We followed in his footsteps, while acknowledging that simple sensorimotor rules have limited explanatory power. We used a similar approach to investigate the role of more sophisticated cognitive, affective and social processing, present in higher animals, and particularly developed in humans, such as consciousness, in the control of adaptive and maladaptive behaviours.

More specifically, we sought to better understand and model the possible causal role of the subjective perspective at the core of consciousness in functions attributed to consciousness \cite{dehaene2017consciousness}, and explicitly relate core aspects of the phenomenology of subjective experience to the motivation and control of behaviour. We used the principles of the Projective Consciousness Model (PCM) \cite{rudrauf3, rudrauf4, rudrauf5, williford}, and its concept of Field of Consciousness, to design and implement artificial agents with an explicit subjective perspective, and simulate adaptive and maladaptive behaviours governed by such perspective, in a manner that was consistent with consciousness theories. We applied the model to a specific robotic context, using Cozmo (Anki) robots, to demonstrate more concretely the behaviours of the agents. 

Our main contributions in this report are:
\begin{enumerate}
    \item to show how the PCM can offer an operational way to understand and model the role of the subjective perspective of consciousness, in a manner that is consistent with the current understanding of the operation and functions of consciousness; 
    \item to show how the PCM can account, based on projective geometry and embodied criteria, for documented psychophysical laws involved in appraisal, and offer a generative model of projective drives for the motivation of action;
    \item to show on this basis how the PCM can be naturally extended to social-affective processing and Theory of Mind (ToM) for behavioural planning, leveraging projective transformations for perspective taking; 
    \item to present a first proof-of-concept implementation of the PCM allowing us to simulate, in an interpretable manner, a variety of relevant behaviours in artificial agents;
    \item to reproduce a spectrum of documented adaptive and maladaptive behaviours that are relevant to developmental and clinical psychology, such as approach-avoidance and joint attention behaviours, based on our operationalisation of consciousness subjective perspective and classical psychological hypotheses. 
\end{enumerate}

\section{Background}

\subsection{Active inference}

Active inference can be described as an iterative cycle, defined by an inference about the causes of sensory information that influences the choice of action in an agent, which outcome in turn influences inference, and so on. Any process that integrates cycles of perception, prediction, and action, in order to confirm or update prior beliefs and satisfy preferences according to sensory evidence, can be called active inference. As an approach to the control and adjustment of behaviour, active inference is relevant to the general problem of cybernetic governance of autonomous agents \cite{haugeland,varela,bach1}. It offers a general framework to reformulate classical Belief-Desire-Intention (BDI) models \cite{georgeff1998belief}, and appraisal models \cite{broekens2008formal}, in a comprehensive psychological and computational framework.

A prominent implementation of active inference is the so-called Free-Energy Principle (FEP). In that framework, an agent has a model of the causes of its perceptions. Each new sensory input induces a posterior on these causes and the agent approximates this probability distribution on the causes. It then acts according to its belief on those causes, in order to further minimize the surprise induced by the input (see \cite{apps,friston2,friston3,limanowski,seth4}). Importantly, internal models of beliefs used for driving the action of agents may include models of preferences and desires, which can be encoded as expectations, e.g. an agent may expect that its actions will be rewarding following some criterion, and as such embedded in a FE function. FEP has been used in multiple contexts, from accounting for visual search and saccadic eye movements \cite{feldman,friston4,brown,parr3,parr4,parr5,parr6,veissiere}, to the modeling of the emergence of affective, affiliative and communicative behaviours \cite{friston5,constant,constant2,veissiere2020thinking,joffily,rudrauf3}(see also \cite{cunningham1}). Although based on a different initial framework, FEP is compatible with Reinforcement Learning (RL) \cite{friston3}, which is based on reward and punishment signals and value functions \cite{hassabis,bach}. 
 
The question that interests us here is to understand and model key aspects of consciousness as part of a process of active inference, in the general sense of the term. 

\subsection{Theories of consciousness}\label{princ-mach-consc}

The mechanisms governing information processing and action regulation in biological systems such as humans are largely unconscious
\cite{velmans1991human, kihlstrom1996perception,doyle2011architecture,merker2013body,van2012unconscious}. Consciousness is only the tip of the iceberg. Nevertheless, the role of consciousness in that cybernetics remains a fundamental question. 

Theories and models of consciousness developed over the last three decades encompass five broad conceptual frameworks \cite{reggia}: integrated information theories, global workspace theories \cite{baars}, internal self-model theories, higher-level representations, and attention mechanisms; the two first frameworks being the most prominent. 

Integrated Information Theory holds that consciousness is a process of information integration featuring integration and differentiation within a coherent dynamical structure associated with a unique subjective perspective \cite{tononi,tononi2016integrated}. As formulated however, the theory lacks specificity and leads to paradoxical predictions imbuing systems such as power grids, metabolic networks, and simple systems of logical gates with consciousness \cite{merker2021integrated}.

In the Global Workspace Theory \cite{baars, baars1}, consciousness is compared to a theatre with a scene on which an attentional beam is directed to highlight different aspects of information used for decision-making. The concept has further been operationalized in cognitive neuroscience \cite{dehaene2011global}. In recent refinements and extensions, which synthesize a broad spectrum of theories of consciousness, two broad classes of complementary yet potentially independent computational functions of consciousness have been distinguished \cite{dehaene2017consciousness}. Following \cite{dehaene2017consciousness}, at a first level, consciousness is conceived of as a mode of access to, and exploitation of multimodal information within a global workspace. According to the theory, the workspace has limited, imperfect representational capacities, is subject to illusions, and entails a dimensional reduction, which yields a single, contextual conscious representation at a time. The workspace is used for serial and sustained processing, selecting and pooling information from memory and multimodal sensing. Its role is to coordinate and broadcast information for further computation in cognitive and sensorimotor systems, for (verbal and non-verbal) reports, and ultimately for high-level decision making, such as the explicit planning and selection of action to achieve long-term goals. Keeping with \cite{dehaene2017consciousness}, at a second level, consciousness is conceived of as an ensemble of mechanisms subtending self-monitoring, or “meta-cognition” of internal representations in the workspace. It relies on memory and imagination, building upon self-models, models of others and the world, and the computation of expected outcomes. It notably aims at reducing uncertainty, prospectively and retrospectively, about sensory evidence, at detecting and correcting errors of predictions in reference to prior knowledge and expectations, and at facilitating social sharing of relevant signals. Social sharing, which is based on the collective broadcasting of information, is thought to be important for the optimization of decision-making, notably by allowing the conscious system to infer that others may have different points of view \cite{bahrami2010optimally}. As a mechanism of reduction of uncertainty, consciousness is understood as fostering curiosity, including through long-term reflection and imagination. Its overarching function is to further optimize decision and action based on a probabilistic sense of confidence and value. It interacts in a complex manner with exogeneous and endogenous attentional processes.

Along similar lines, Aleksander \cite{aleksander} has proposed five "axioms" that must be met to consider artificial systems as including a model of consciousness : 1) presence, which involves mechanisms representing the world and the organism in it; 2) imagination, which corresponds to internal simulations of state trajectories in the absence of sensory input; 3) attention, which involves mechanisms guiding sensors during perceptual acts, and modulating imaginary processing; 4) planning, which relies on the internal exploration of possible actions through imaginary processing; 5) emotion, which participates in the evaluation of plans.  

According to prominent theories, another central aspect of consciousness (which is compatible with the above theories) is the integration of a representation of the body in space and of its relations to the environment, aimed at fostering homeostasis, survival and well-being, through embodied appraisal and emotion \cite{seth2012interoceptive,damasio1999feeling,man2019homeostasis,blanke2012multisensory}.

\subsection{The prominent yet elusive place of a subjective perspective in consciousness theory}\label{princ-mach-consc}

It remains largely mysterious however, how information is specifically accessed, shaped and exploited through consciousness, and in particular, how unified principles could account for the phenomenology of subjective experience and its role in the functions of consciousness. Understanding how qualitative experience (or subjective character) participates in the process of information integration and behavioral control carried out by consciousness is an important challenges for consciousness theory \cite{seth,seth2,seth3,revonsuo,chella,manzotti}.

From classical theories to more recent ones, much of consciousness research has brought to the fore the phenomenologically salient and central role of a "subjective perspective", as a non-trivial, viewpoint-dependent, unified internal representation of the world in perspective \cite{james1890principles, nagel1974like, trehub2007space, velmans1990consciousness, lehar2003world, merker1, merker2021integrated, godfrey2019evolving, tononi, rudrauf3}.

Although the notion connects to research in visual awareness \cite{hering1942spatial, roelofs1959considerations, howard1966human, starmans2012windows}, such view-point dependent subjective perspective in fact pertains more generally to multimodal or supramodal representations of space in relation to the situated body \cite{rudrauf4}. It would play a pivotal role not only in the framing of perception, but also in the shaping of non-social and social perspective taking (see below) through imaginary projections, for the integration and exploitation of information \cite{trehub1977neuronal, cox1999initial, neelon2004isoazimuthal, sukemiya2008location, limanowski2011we, blanke2012multisensory, berthoz2010spatial, cattaneo2011blind, metzinger2005out, merker1, merker2013body, heller1990perspective, tinti2018my}.

Theories such as Integrated Information Theory \cite{tononi} have attempted to address the question at a basic level, in a manner that reduces the subjective perspective under a unique point of view to purely information theoretic considerations, and in doing so largely fails to capture the phenomenon \cite{dehaene2017consciousness, merker2021integrated}. While acknowledging the importance of subjective experience as characterized by its phenomenology for consciousness theory, others have preferred to leave the problem aside \cite{dehaene2017consciousness, crick1990towards} (see: \cite{merker2021integrated}). It thus remain to be understood how a subjective perspective may be amenable to a computational operationalisation, which would make its features an integral part of the cybernetic functions that consciousness is thought to play, in a manner that subsumes current theories of consciousness \cite{dehaene2017consciousness} and Aleksander's axioms \cite{aleksander}.

\subsection{Perspective taking abilities at the core of human consciousness}\label{intro}

In this section, we further review the concept of perspective taking and its relation to social inference and theory of mind, as, in light of the above, the concept appears central to consciousness and its relations to motivated behaviors. Perspective taking relies on the spatial framing of multimodal sensory and affective information, relating egocentric and allocentric frames with respect to the body, and conditions how embodied agents represent and interact with objects and others (see review in \cite{riva1, rudrauf4, mchugh1}). Interestingly, abstract reasoning has been linked to the representation of space in situated autonomous agents \cite{constantinescu}. Furthermore, humans, and other animals, are social beings that perform theory of mind (ToM), which is itself grounded in perspective taking \cite{Premack1}. These mechanisms are essential for the development of adaptive and maladaptive behaviours. 
 

Visual perspective taking (VPT) is the process through which one can simulate how the world may appear from another point of view \cite{hamilton1}. VPT enables people to infer that a specific object can or cannot be seen from another point of view (VPT $level-1$, or VPT1), e.g. due to obstructions, and that the same object can be seen differently by two persons (VPT $level-2$, or VPT2). Performance in VPT2 has been linked to social mentalising abilities or Theory of Mind (ToM) \cite{hamilton1}. The terminology might be somewhat misleading as VPT does not imply an exact representation of the detailed visual experience. Instead, it can be sketchy and imprecise, and nevertheless play an important functional role. Moreover, it is not necessarily purely visual but in fact should be conceived as a supramodal process of spatial cognition \cite{rudrauf3} (see also above). 


ToM is the ability to make inferences about others' metal states, beliefs and desires, in order to understand and predict their overt behaviours. It encompasses cognitive and affective dimensions, and relies on imaginative processes \cite{kalbe1,baroncohen1,wimmer1}. Mentalising is another term for ToM \cite{frith1}, used in a psychodynamic perspective \cite{gergely1}. The cognitive and psychodynamic literature emphasizes the integral link between perspective taking and ToM or mentalising, in connection to developmental and clinical rationale aiming at understanding and treating psychopathology (for a review on this link see for instance \cite{barnesholmes1}).  

Six levels of perspective taking underlying ToM have been distinguished as components of the normal development of children, which disruption may lead to differentiated symptomatic behaviours \cite{hadwin1}. $Level-1$ and $-2$ respectively correspond to VPT1 and VTP2 described above. $Level-3$ is the ability to understand that knowing requires to reduce uncertainty through direct sensory evidence (to know for sure what a box contains, one has to look inside it). $Level-4$ corresponds to the use of true beliefs about past experiences to correctly predict others' behaviour (one may predict that another knowing that a box contains a desired object will approach it). $Level-5$ is the ability to understand false beliefs and their consequences. A seminal task to probe $level-5$ ToM is the so-called Sally-Anne test \cite{baroncohen2,leslie1}. Typically, two puppets, Sally and Anne are presented to a child. Sally puts a marble in a basket while Anne is watching. Sally goes out for a walk. Anne takes the marble out of the basket and puts it into a box next to the basket. The question to the child is then: where will Sally look for her marble when she comes back? A child with $level-5$ ToM will answer in the basket, whereas a child without it will answer in the box (where it is now). Non-verbal tasks have been used to demonstrate that $level 5$ may be developed much earlier than previously thought (as soon as 15-month-old if not earlier) \cite{onishi1}. $Level-6$ ToM, also called second-order ToM, corresponds to the ability to understand embedded beliefs, i.e. that other people can themselves make inferences about what other people are thinking. It is the ability to perform ToM about others' ToM.  

Simulation theory, i.e. the internal simulation of self and others by projection of self-models, has been proposed to account for essential aspects of social perspective taking, empathy and affective learning \cite{lamm, berthoz2010spatial}. The adaptive use of perspective taking is believed to play a critical role in emotion regulation and social-affective development, according to both cognitive-behavioural \cite{gross,clement} and psychodynamic traditions \cite{gergely,beckes,fonagy}. It also connects to appraisal theory \cite{ellsworth,moors,scherer}, and in particular reappraisal, which fosters resilience \cite{kalisch}. 


An integral link between the capacity for perspective taking at multiple levels and the emergence of behaviours such as joint attention behaviours has been hypothesized. As a behaviour, joint attention or shared attention corresponds to a situation in which two or more individuals look at the same object (or person) with expressed interest. Psychologically, it requires that the individuals establish a shared interest for the item through ToM. It is typically triggered by an individual signaling the interest of the item through orientation, eye-gazing, pointing, vocalizations or verbal utterances. The ability to perform joint attention normally appears during the first year of age \cite{scaife1}. According to a prominent theory, it implies the ability to represent another person's representation of items, in terms of both the spatial relationships between that person and the items, and how they appear tagged for this person with a positive or negative valence, which underlines an interest in, or preference for attending to the items \cite{baroncohen3, ciompi1991affects}. For this reason, we use the term ``social-affective perspective taking" throughout this article. This process could represent a powerful method for view-point dependent, situated systems, to make inference, learn and expand knowledge adaptively about their environment through the experience of others \cite{dawson1}.         


Disturbances of social-affective perspective taking have been hypothesized to play a central role in the emergence of certain maladaptive behaviours and disorders. Thus Autism Spectrum Disorders (ASD) and Social Anxiety Disorders (SAD) (among other disorders) have been related to disturbances of social-affective perspective taking at different levels, which would impact adaptive emotion regulation strategies, approach-avoidance and joint attention behaviours.

ASD is characterized by deficits in social interactions and communication, restricted interests, and repetitive behaviours \cite{dsm5}. One of the key early diagnostic criteria for ASD is a deficit in joint attention behaviours. ASD has been associated with impairment in social-affective perspective taking or more generally ToM, and their interactions with emotion regulation \cite{samson1,hamilton1}. According to the Empathizing-Systemizing (E-S) theory, ASD would result from a developmental delay or disturbance of ToM, causing a ``mindblindness" and reduced empathetic abilities, combined with tendencies for systemizing, i.e. seeking to explain the world with limited, rigid rules \cite{baroncohen4}. Mindblindness would be related to asymmetries in processing the self \textit{versus} others \cite{lombardo1}. Social motivation theory emphasizes a primary deficit in interest for social stimuli, effectively leading to mindblindness \cite{chevallier1}. 

SAD is characterized by a debilitating fear and avoidance of social situations, in which the performance of an individual might be judged in an humiliating way \cite{dsm5}. SAD individuals tend to over-mentalise and attribute negative beliefs about them to others, which altogether would underlie social avoidance behaviours \cite{washburn1,moscovitch1}. At the same time, they maintain high social standards and strong affiliation needs (for review, see \cite{morrison1,steinstein1}). In other words, SAD individuals find themselves in a conflict of approach-avoidance with respect to others, which they fear but care about \cite{clarkwells1}. They appear to struggle with shifting their attention towards positive cues that could disconfirm their negative beliefs \cite{morrison1}, which is reinforced by their avoidance of social exposure \cite{moscovitch1}.

\section{Bridging the gap: what is thus the task at hands?}\label{bridgegap}

\subsection{Epistemological considerations}

Considering the features and functions of consciousness reviewed above \cite{dehaene2017consciousness, aleksander}, and the issue of the binding of these through a subjective perspective \cite{merker2021integrated, rudrauf4}, the task at hand is thus to integrate the following into an operational model. The model should subsume a subjective perspective, representing objects and others, under a unique point of view, in a manner that is consistent with the phenomenology of subjective experience. That subjective perspective should serve as a mechanism of access to, and integration of information, and operate as a global workspace playing a causal role in active inference and the control of behaviour. Following \cite{dehaene2017consciousness, aleksander}, this entails understanding how such subjective perspective could participate: 1) in the appraisal of values and monitoring of uncertainty, as a basis to inform (affective and epistemic) motivational drives, and 2) in programming action through internal simulations encompassing counterfactual non-social and social imaginary perspective taking \cite{baroncohen3, lamm}. 

The problem requires to understand how a subjective versus objective representation of the environment could in and of itself contribute to the process. More specifically, one must consider how the perspectival manifestation of contents (which does not respect Euclidean properties of objects in ambient space, and thus does not correspond to their objective existence), may support information processing and control, and serve as an intermediary level of representation and as a mechanism for the programming of action. Figure \ref{fig:Fig_0} presents a synthetic illustration of the problem.  

\begin{figure}
    \centering
    \includegraphics[width=1\textwidth]{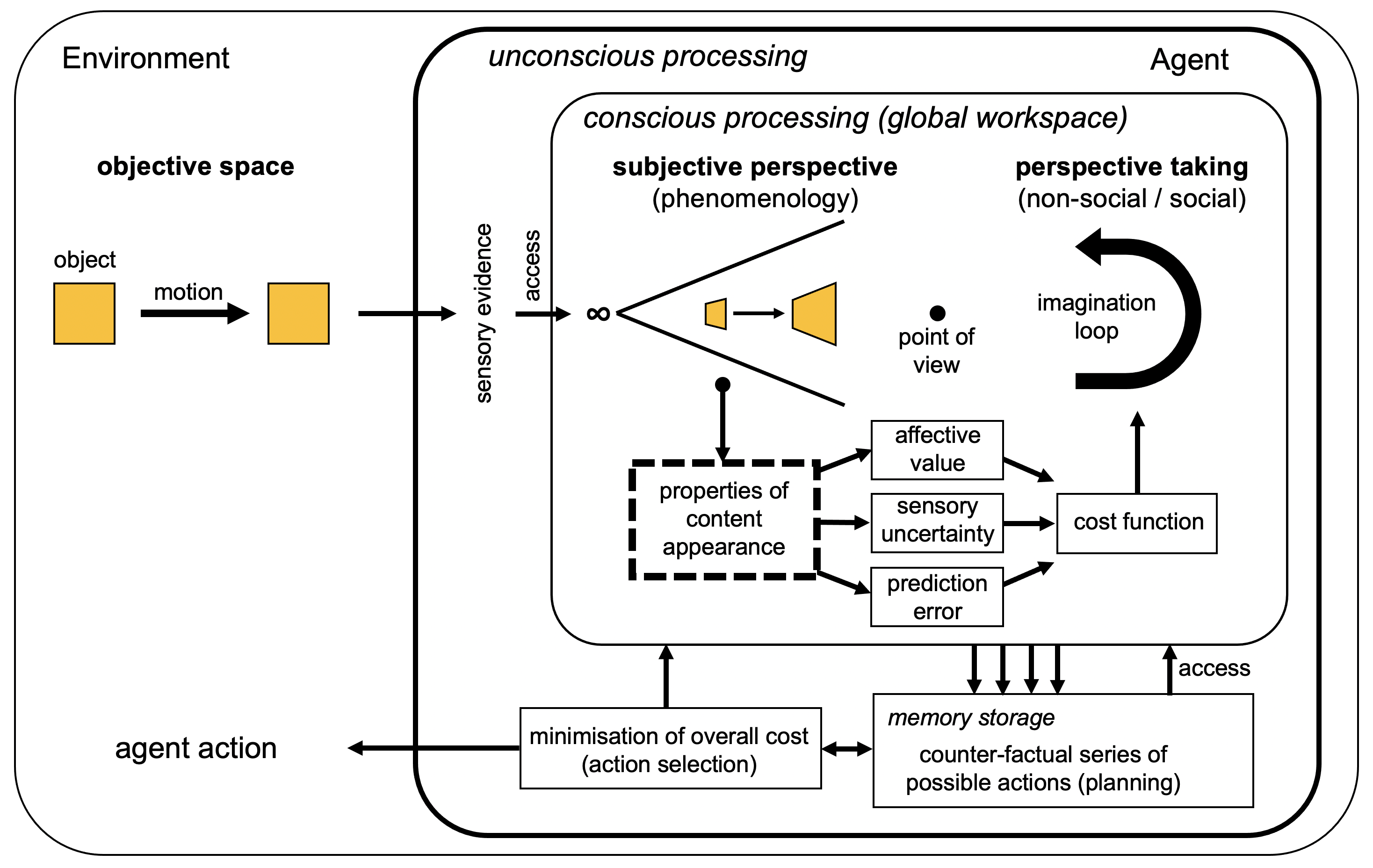}
    \caption{\textbf{Concepts to integrate in a model of consciousness}}
    \medskip
    \small
    \raggedright
    Following \cite{dehaene2017consciousness, aleksander, merker2021integrated, rudrauf4, baroncohen3}, the model should access objective information, e.g. the Euclidean properties of objects in space, within a viewpoint-dependent subjective perspective serving as a global workspace, and compute, from intrinsic (phenomenal) properties of content appearance within that perspective, quantities that capture affective values, sensory uncertainty, and prediction error, in relation to prior beliefs and preferences. It should do so in order to evaluate possible actions through an imagination loop, implementing non-social and social perspective taking. The overall cost of actions should be evaluated on this subjective basis in order to select actual actions, so that the subjective perspective itself can be considered to play a causal role in the process.   
    \label{fig:Fig_0}
\end{figure}

\subsection{The projective consciousness model}

The principles of the Projective Consciousness Model (PCM) are precisely meant to pursue that goal, by modeling the subjective perspective based on a view-point dependent projective geometrical framing of information \cite{rudrauf4,williford,rudrauf5,rudrauf3}. The approach emphasizes the importance of combining geometrical and information theoretic concepts for a fuller operational account of consciousness. The PCM is a general model based on two overarching principles: 1) consciousness is at the core of active inference as a process, 2) it integrates information within a global workspace, i.e. the Field of Consciousness (FoC), which is framed in a viewpoint-dependent manner and governed by 3-dimensional projective geometry. The second principle derives from considering the geometrical features of the phenomenology of the subjective experience of space and its contents, which appear, both in perception and imagination, to be organized according to a 3-dimensional space in perspective from an elusive point of view \cite{rudrauf4}.

Projective geometry encompasses some of the properties of perspective. It extends affine geometry with points at infinity, which are considered as any other points, at which all "parallel" lines meet (as in perspective drawing), thus defining a projective space. It results from (projective) transformations that preserve the incidence structure of points, lines, planes and hyperplanes, but do not preserve angles. The 2-dimensional projective space, or projective plane, is the space of all lines of the 3 dimensional Euclidean space going through the origin. It is a model of how the space of all light beams converging towards the eye of an observer behaves under different perspectives. A 3-dimensional projective space $P_3(\R)$ can be defined from a 4-dimensional vector space in which all collinear vectors are considered equivalent. Transformations of $P_3(\R)$ are governed by the action of the group of projective transformations in 4-dimensions $PGL(4)$. As we shall see, one interest of perspective taking as governed by projective geometry is that it discriminates objects of the euclidean space by ordering them according to a point of view such that the apparent size of objects changes depending on choices of focalisation. 

In the PCM, such a 3-dimensional projective space is the key component that defines the FoC as a globale workspace. It explicitly models a subjective perspective from a unique point of view, as a mechanism of access and processing of information. The FoC is used to access and represent an internal, unconscious world model, and integrates cognitive and affective parameters associated with it, in a view-point dependent manner. It can be transformed to take multiple perspectives on that world model based on projective transformations. Thus one key feature of consciousness according to the PCM is that it accesses and frames unconscious representations by transforming their encoded, viewpoint-independent Euclidean ``coordinates" into viewpoint-dependent projective (homogeneous) ``coordinates".

We recently demonstrated how the principles of the PCM could explain perceptual illusions such as the Moon Illusion and Ames Room, as a result of the calibration of a 3-dimensional projective chart under free energy (FE) minimisation (\cite{rudrauf5}, see also \cite{rudrauf4}). We previously offered preliminary descriptions of how the same principles could integrate affective processing and imaginary projections to support active inference \cite{rudrauf4, rudrauf3}.

The general rationale is that the FoC provides an agent with a situated, perspectival representation of the world, including the agent itself and its environment, in three dimensions, and directly supports the analysis, appraisal and reappraisal of spatial-temporal, epistemtic, affective and social information within that subjective referential, in which orientation, directions, and relations of incidence are essential, in order to motivate and program action. Through the FoC, an agent can relate perception (representations directly informed by sensory evidence), and the imagination (representations informed by prior beliefs and preferences only) as part of a process of active inference, by accessing unconscious representations, encoded in memory, of objects, self and others, and their relations, which are weighted by expectations, and updated based on sensory evidence. The process is governed by changes of projective charts, enabling an agent to take different perspectives on its internal model, which would result from different choices of actions, in order to simulate and optimize action outcomes. 

\subsection{Beyond qualitative considerations}

Until now, these considerations have remained largely qualitative and no computational model integrating these different aspects had explicitly been implemented yet. 

As we shall see, the projective framing of the world by an agent in its FoC, can be a basis for deriving affective and epistemic values that are explicitly subjective and dependent on how the world appears to the agent. These in turn can be associated with a FE to minimise in order to appraise and select the best overall course of action. Recursive, imaginary perspective taking through the FoC can be used to minimise FE, with a variable depth of processing (here simply the number of successive perspectives being anticipated as a result of possible actions), in order to envision better states beyond the immediate surrounding of the agent, and in a manner that can take into account others' perspectives.

\section{Model}\label{mat-and-meth} 

In the following, we present an initial implementation of the model as a proof-of-concept. Our approach embeds a projective model of the subjective perspective of consciousness within a process of optimization based on the concepts of active inference and free energy minimization. However, we consider the problem of active inference within a broader framework than Friston's FEP, and do not follow the FEP formally, in part for convenience and flexibility in modeling. Our model is divided into two steps of inference and actions, but we depart from FEP in particular with respect to the action part, as we consider that agents have precise goals or value-functions defining the satisfaction of affective and epistemic drives. We use a cost function for action that is related to a free energy (FE), and formulated in relation to probability densities, in a manner that allows agents to optimize both preferences and uncertainty reduction (yielding a form of intrinsic curiosity in the agents). In this initial implementation of the model, we used an approach that is related to classical optimal control. In the future, we may extend it to stochastic optimal control, e.g. to use a larger set of algorithms to find critical points of the cost function, and relate the approach more explicitly to RL, but this is beyond the scope of this contribution. Likewise, we used simplifications for certain components of the processing pipeline we implemented, e.g. we do not claim to compete with state of the art models of emotion processing and expression. In the following sections, we describe the main components of the model (see appendix \ref{model-tech-details} for further technical details). 

Let us remark that in our model we call FE a mean of Kullback-Leibler divergence. It is justified as there is a close link between FE and this divergence, specifically the distributions that optimize one or the other are the same when considering the true FE or the Kullback-Leibler divergence in the action phase (see also remark \ref{appendix-rem-fe-dkl} for more details on the link between the two quantities). We present our model using the Kullback-Leibler divergence as we find it more directly interpretable than FE itself.

\subsection{General settings}\label{general_settings}

We consider three different types of entities: objects, agents, and subjects. Let us note $E$ the collection of objects $O$ and agents $A$. Objects can be seen as a subset of the surrounding space $X_o\subseteq \mathbb{R}^3$. Agents are not objects, but as they are embodied, they also occupy a certain portion of space, like an object. We will use the term subject, in order to distinguish an agent under consideration from other agents. 

Two main processing phases are emphasized as parts of active inference. During the first phase, the agent assesses its surrounding contextually, e.g. whether preferred given entities appear large or small in its FoC, and thus close or far. Then, during the second phase, the agent decides which action to perform, e.g. how it could approach preferred entities while avoiding disliked ones based on perspectival projections of states, which anticipate the consequences of action. Actions then influence sensory inputs, which closes the loop of perception and action.  

Formally there are 3 spaces, $S$ the space of sensory input, $\Gamma$ the space of different contexts for the world, and $N$ the space of actions, and two cost functions, $F_1:S\times \Gamma \to \R$, $F_2: N\times \Gamma \to \R$. This can be summarized as the following scheme, for a given input $s\in S$,

\begin{enumerate}
    \item  Inference: 
    \begin{equation}
        \hat{\gamma}=\argmin_{\gamma\in \Gamma}{} F_1(s,\gamma)
    \end{equation}
    
    \item  Action:
    
    \begin{equation}
    \hat{m}=\argmin_{m\in N}{} F_2(a, \hat{\gamma})
    \end{equation}

\end{enumerate}

This general scheme applies to Friston's FEP as well as to our model, even though our approach departs in part from the FEP, as in FEP,  $\Psi$ is a space of causes for sensory inputs and $\Gamma$ parametrizes probability distributions on $\Psi$. In FEP, for some sensory input $s$ the cost function for the inference is, 

\begin{equation}
F_1(s,\gamma)=\DKL(Q_{\gamma}\Vert P(.\vert s))= \sum_{x\in \Psi} Q_\gamma(x) \ln \frac{Q_\gamma(x)}{P(x\vert s)}
\end{equation}

The cost function for action, in the case of FEP, is,

\begin{equation}
 F_2(m,\hat{\gamma})=E_{Q_{\hat{\gamma}}}[-\ln P(s(m),.) ]
\end{equation}

In our model, a subject embeds internal models of preferences or affective expectations, regarding objects and agents, $p_s\in [0,1]^{A\times E}$, which define a preference tensor $(p_{abe}\in[0,1], a\in A\quad  b\in A\quad e\in E )$. The true preferences of a subject $a\in A$ for an entity $e\in E$ is given by $p_{aae}$ denoted simply as $q_{se}$, and its beliefs about the preferences of another agent $b\in A$ for an entity $e\in E$ is given by $p_{abe}$. The sensory inputs are the emotions expressed and the configuration of the entities in the physical space (entities are explicitly encoded as 3-dimensional structures in ambient space). The space of the preferences of the subject, $[0,1]^{A\times E}$, plays the role of $\Gamma$. 

\subsection{The Field of Consciousness}

\subsubsection{Projective geometry, psychophysics of appraisal and affective drives}\label{intro:foc-motive}

The PCM assumes that 3-dimensional projective geometry is essential for capturing key features of consciousness. Following defining properties of consciousness as reviewed above (see Section \ref{princ-mach-consc}; see also Figure \ref{fig:Fig_0}), it is thus important to assess whether that geometry could intrinsically induce effects on subjective appraisal. More specifically, we should assess how these effects could relate to the psychophysics of appraisal in space, and whether they could in turn explain how the subjective perspective may contribute to the motivation of action as a result.

A documented effect, pertaining to the subjective appraisal of danger, is that reported subjective fear depends on the distance between a participant and a dangerous object presented in front of the participant, e.g. a snake, in a manner that approximately follows an inverse distance law \cite{teghtsoonian1}. This effect was linked to Stevens' psychophysical distance, which is a power law, such that when its exponent is equal to $-1$, the function follows an inverse distance law \cite{teghtsoonian1}. Let us consider how this relation may relate to the FoC.
In what follows when considering a vector $v$ in $\R^n$, we will denote $v[k]$, for $k\leq n$, its $k$-th coordinate.\\

The 3-dimensional projective space, $P_3(\R)$ is the set of lines of $\R^4$, any bijective linear transformation from $\R^4$ to $\R^4$, i.e. any invertible $4\times 4$ matrix denoted as $M$, defines a projective transformation; the projective transformation can also be expressed as a partial map from $\R^3$ to $\R^3$ as follows, for $(\lambda_1,\lambda_2,\lambda_3)\in \R^3$ such that $M(\lambda_1,\lambda_2,\lambda_3,1)[4]\neq 0$,
\begin{equation}
\phi(\lambda_1,\lambda_2,\lambda_3)=  \left(\frac{M(\lambda_1,\lambda_2,\lambda_3,1)[1]}{M(\lambda_1,\lambda_2,\lambda_3,1)[4]},\frac{M(\lambda_1,\lambda_2,\lambda_3,1)[2]}{M(\lambda_1,\lambda_2,\lambda_3,1)[4]},\frac{M(\lambda_1,\lambda_2,\lambda_3,1)[3]}{M(\lambda_1,\lambda_2,\lambda_3,1)[4]}\right)
\end{equation}

We shall denote $\phi$ or $M$ the projective transformation depending on whether it is seen in $\R^3$ or $\R^4$. We do not want to go into the detail of what a projective frame is, so here when referring to a frame we mean a choice of a basis of \footnote{Not their image in $P_3(\R)$, which would need to specify one more point} $\R^4$ that is used to define coordinates of $P_3(\R)$.

The locations of objects and other agents are expressed in coordinates that are proper to the subject, which correspond to what is located in front or behind of it, on its left or right, below or above. Each subject must select a perspective as part of its representation of the external world in its internal world, which is encoded by a 3-dimensional projective transformation. Following \cite{rudrauf5}, this transformation relates to a subjective, internal method of organization of an internal world model, based on sensory data and prior beliefs, which defines the FoC. There is no canonical choice of projective transformations for moving from the external world to the internal world. However, the problem of the choice of a possible family of projective transformations can be constrained based on considerations about an agent embodiment, according to the following properties:

\begin{enumerate}

\item The subject is centered in $0$ after projective transformation, i.e. in its perspective it is at the center of its frame.

\item  The three axes $x,y,z$ are preserved, i.e. the axes of the Euclidean frame associated with the agent (up-down, left-right, and back-front) must be preserved after projective transformation.

\item  No points in ambient space appears, to the subject, to be truly at infinity; this constraint being satisfied when the subject only directly represents a half space (in perception or imagination).

\item Objects that are near the agent appear to have the same size as in the reference Euclidean frame. 
\end{enumerate}

These properties limits the choice to transformations that have the following matrix expression in homogeneous coordinates,

\begin{equation}
M=\begin{pmatrix}
1 & 0 & 0 & 0 \\
0 & 1& 0& 0 \\
0 & 0  & 1& 0 \\
0 & 0 & C & 1
\end{pmatrix}
\end{equation}

where $c\in \R_+$ is a positive real number (see appendix \ref{projective-results} for a proof of this result) that represents the inverse of a depth of field.\\

The entities of the `real world' as encoded in the agent's world model are sent to the subject's internal subjective workspace by expressing the coordinates of the points that constitute the entities in the basis attached to the subject and then by applying the above projective transformation. We call this transformation the subject's projective chart and denote it as $\psi$. This applies to both perception and imagination of the subject. The subject could thus consider locations in the Euclidean space that are not accessible through immediate sensory evidence, and perform imaginary simulations of the consequences of anticipated motion toward those locations based on prior beliefs, which would entail the computation of an imaginary projective chart.

In the projective chart of the subject, the perceived size of objects, $r_p$, varies (asymptotically) as, 

\begin{equation}
\frac{r_p}{r}\simeq \frac{1}{cz}
\end{equation}

where $z$ is the coordinate of the direction along which the subject is aiming and $r$ the length of the object in ambient space; an exact expression also holds if the object position is orthogonal to the direction of sight of the subject (see Appendix \ref{projective-results}). The size evolves as a Stevens' Law with exponent $-1$, as found in (\cite{teghtsoonian1}) in the context of the appraisal of threat. See Figure \ref{fig:Fig_1a} for an illustration of these principles.

This suggests that, at some level, the subjective experience of the apparent, subjective size of the objects in the FoC itself (and not an actual estimation of objective Euclidean distance) might serve as a criterion for affective appraisal at the level of consciousness. Although the documented effects concern the appraisal of danger (negative valence), we will assume in the following that the property relating projective apparent size to Euclidean distance, is used both in the appraisal of negatively valenced and positivily valenced entities. In other words, the closer the object I fear or like, the bigger the object feels, and thus the fear or joy I experience. This is an hypothesis that is motivated by the geometrical properties of the model, which offer a potential mechanism for basic affective appraisal in general. Future studies should assess this prediction empirically, but this is beyond the scope of this theoretical contribution. 

Thus, in our operationalization, the perceived value of objects directly depends on the perspective being taken. Such relationship between affective appraisal and projective size can then be used to induce a direct effect on the motivation of approach-avoidance. Indeed, when an agent projects itself, in an imaginary manner, closer to an object it likes, respectively farther from an object it dislikes, the agent will experience a greater satisfaction of its preferences, respectively a greater dissatisfaction, due to the effect of the projection on apparent size, according to the current rationale. As specified in section \ref{mat-and-meth}, in the first case, expected FE will be lower, and in the second case, it will be higher, which will drive the agent to act accordingly, so that it will approach the liked object and avoid the disliked one. In other words, in this framework, the FoC acts as a non-trivial ``lens" on information, which it transforms and funnels in a subjective frame, in a manner that entails a mechanism of projective drive, as part of the processing pipeline underpinning the overall cybernetics of the agent (see section \ref{cost-function}, which introduces the specific cost function implementing such drive). 

\begin{figure}
    \centering
    \includegraphics[width=0.8\textwidth]{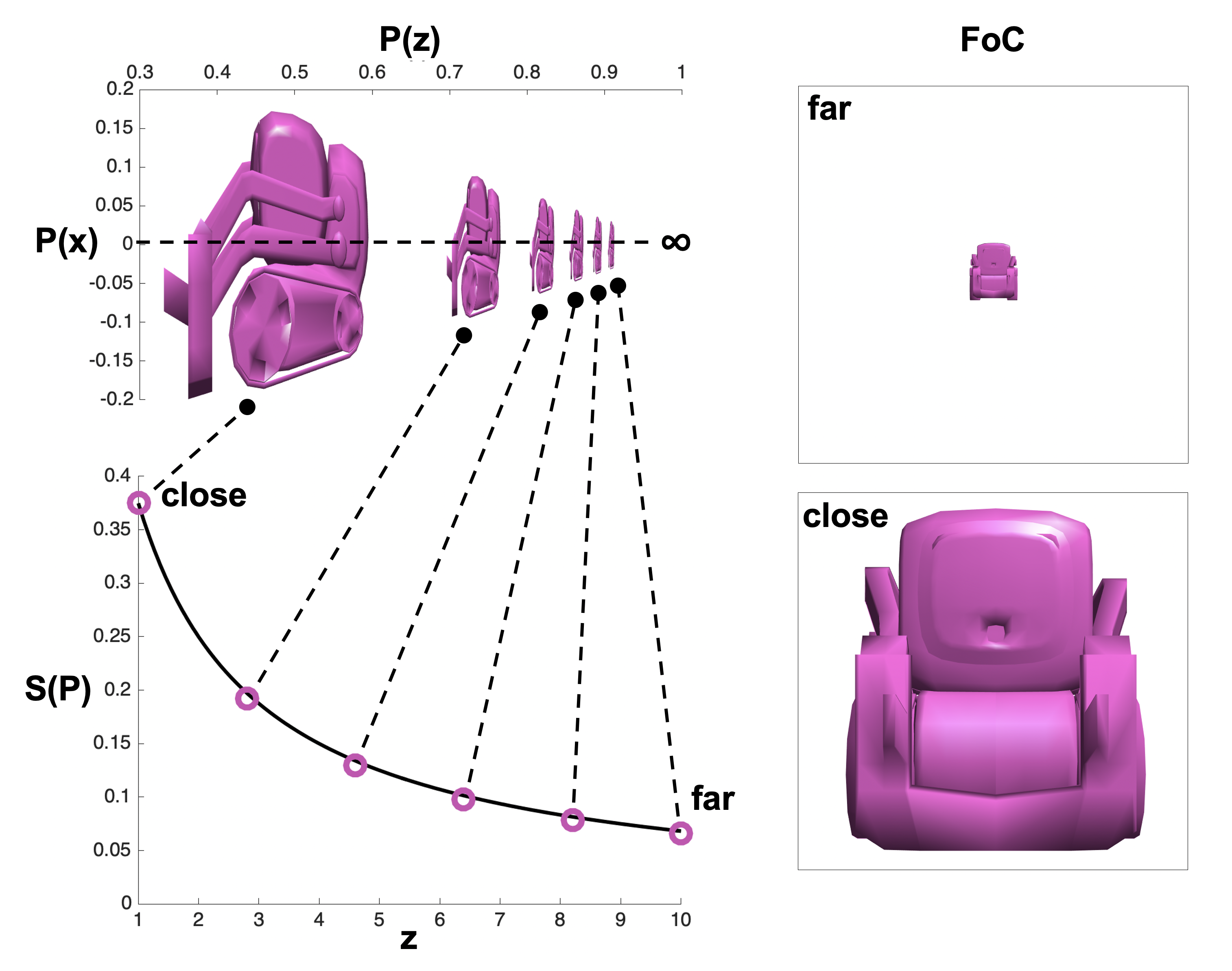}
    \caption{\textbf{The Field of Consciousness and psychophysics of appraisal}}
    \medskip
    \small
    \raggedright
    Effect of 3-dimensional projective transformations on the representation of an object (the mesh of a Cozmo robot) as a function of its distance from the observing agent. \textit{Left-tier}. The upper chart shows how the Euclidean coordinates $(x,z)$ are transformed by the 3-dimensional projection $P(.)$. The direction of projection is $z$. The chart shows the projective space as if one could see it from side to reveal the effect of transformation on a 2-dimensional viewing plane. The lower chart shows how the size of the transformed object (here the quadric root of its volume in the projective space) varies as a function of Euclidean distance $z$. The linear size decreases approximately as the inverse of the distance. \textit{Right-tier.} Views showing the first person perspective (FoC) on the same transformations corresponding to close \textit{versus} far distances. Note that the views, which here are reduced to an image plane, are in fact 3-dimensional, with an explicit parameter of depth $P(z)$.    
    \label{fig:Fig_1a}
\end{figure}

\subsubsection{Uncertainty with respect to sensory evidence and epistemic drives}\label{subsec:uncertainty}

In the previous section, we considered the relation of the FoC to affective drives, but epistemic drives related to uncertainty with respect to sensory evidence can be embedded in the process using similar principles to define the cost function described in section \ref{cost-function} implementing the drive mechanism.

The FoC frames information by accessing projectively an unconscious world model, which contains prior beliefs updated based on sensory evidence. For perception, the FoC has to be localized and oriented with respect to the agent's body in ambient space, so that its field of view encompasses regions of space that can be sampled based on sensory processing. For imaginary projections, this constraint is released, and the FoC can explore its world model beyond regions of space that can be sampled based on sensory processing. In the latter case, the FoC has to rely entirely on prior beliefs. 

Thus, uncertainty with respect to sensory evidence can be taken into account. In our approach, FE will tend to decrease for actual or anticipated states of the FoC associated with lower amounts of sensory uncertainty than the current one. Such states of the FoC will acquire as a result an intrinsic interest, competing with other factors determining the value of FE. We integrated such uncertainty in our model the following way. For a state of the FoC that samples regions of space compatible with sensory evidence acquisition, we defined uncertainty as a function ranging from $0$ to $1$, increasing from the center of the field of view towards the back of the agent, and from proximal locations towards more distant locations in space (we are more uncertain about what we cannot perceive well). We leveraged psychophysical functions derived from vision science for the effect of excentricity on uncertainty (building upon the decreasing resolution of the visual field from to fovea to the periphery)\cite{loschky2005certainty}. Of course, the sampling of space in a multimodal context is richer and more complex (think about source localization from auditory processing), but we restricted ourselves to simple vision-inspired models concerning the uncertainty at hand for this report. 

This lead us to consider the uncertainty $\sigma$ of an object or an agent. $\sigma_e$ is the uncertainty a subject has towards an entity $e$ relative to its field of view. The more remote and off-centered $e$ is, the higher the associated uncertainty $\sigma_e$. Figure \ref{fig:uncertainty} shows the behaviour of $\sigma$ in a 2-dimensional plane from above the subject. The metric is computed following \cite{loschky2005certainty}, as
\begin{equation}
    \sigma_e = 1 - \rho_e
\end{equation}
$\rho_e$ is the certainty, as 
\begin{equation}
    \rho_e = \frac{2\zeta}{\omega}
\end{equation}
where $\omega$ is the distance between the subject $s$ and $e$. The acuity, $\zeta$, is defined as,
\begin{equation}
    \zeta = e^{-\kappa\epsilon}
\end{equation}
with $\epsilon$ as the eccentricity metric and $\kappa$ a constant that represents the acuity of the subject's vision.\\

The uncertainty depends on $s$ through $\psi_s$ which we shall denote as $\sigma_{s, \psi_s}$.\\  

\begin{figure}
    \centering
    \includegraphics[width=0.6\textwidth]{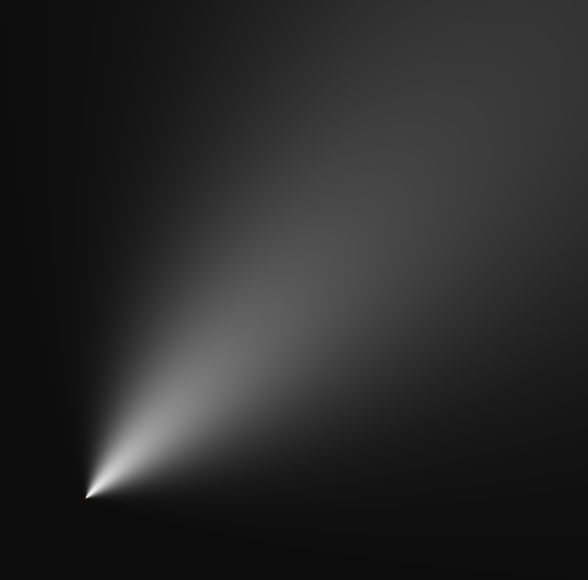}
    \caption{\textbf{Illustration of $\sigma$}}
    \medskip
    \small
    \raggedright
    Parameter $\sigma$ (uncertainty with respect to sensory evidence) as seen from above. The base of the white beam is the point of view of the subject. The whiter a point is, the lower its $\sigma$.
    \label{fig:uncertainty}
\end{figure}

\subsection{Perceived value and basic cost function}\label{cost-function}

We introduce the perceived value, $\mu_e$ of an entity by a subject, $s$, which plays a pivotal role in its drives. It is one innovation of our model, as it is the quantity that connects perspective taking and action. For an entity $e\neq s$, $\mu$ is a weighted mean of the intrinsic preference associated with the entity and a neutral preference associated with ambient space, weighted by the readjusted perceived volume $v_p$ of the entity in the subject's projective chart,

\begin{equation}
\mu= p\gamma\frac{v_p^{1/4}}{v_{tot}^{1/4}}+ q_n(1-\gamma\frac{v_p^{1/4}}{v_{tot}^{1/4}})
\end{equation}

with $\gamma\in [0,1]$ that captures an attentional focalisation that is centered with respect to the body of the agent and its orientation, and $v_{tot}$ the total volume of the bounded field of view of the subject in its projective chart. The adjustment of the volume with a power $\frac{1}{4}$ is introduced to be in agreement with the law of inverse distance found experimentally (see above), since the ratio between the perceived volume and the true volume, $v$, varies as (see appendix \ref{projective-results} for further details),

\begin{equation}
\frac{v_p}{v}\sim \frac{1}{(cz)^4}
\end{equation}

Importantly, this property of the apparent volume directly relates to the subjective perspective of the subject and only indirectly to objective properties of ambient space, such as Euclidean distance.

The perceived value is a quantitative estimation of the impact an entity has in the subjects perspective. When choosing how to move, the aim of the subject is to maximize positive outcomes, which correspond to high values of perceived value and low uncertainty about it. In order to achieve this aim, the subject needs to quantify the extent to which its state departs from the ideal configuration of high perceived value and low uncertainty, $(\mu_0,\sigma_0)$. 

We do so by embedding these parameters into a space of probability distributions over perceived values, i.e. probability distributions over $[0,1]$, and to quantify the ``distance" of $(\mu,\sigma)$ from $(\mu_0,\sigma_0)$ using the relative entropy also called Kullback-Leibler divergence. $\mu$ and $\sigma$ parameterise a probability law on $[0,1]$ centered on $\mu$ and of dispersion $\sigma$, which we shall note as $Q(\lambda|\mu, \sigma)$ with $\lambda\in [0,1]$. We chose to consider that $Q(.|\mu,\sigma)$ is a truncated Gaussian of mean $\mu$ and variance $\sigma$. The law attributed to the ideal configuration will simply be denoted as $P$, and for simplicity, we take it to be the same for every agent. It is important to emphasize that each actual or imagined move changes the subject's projective chart, i.e. FoC, which directly influences the perceived value and uncertainty on it, $\mu$ and $\sigma$.

For each move $m\in N$, $\DKL(Q\left(.|\mu_{s,\psi_s(m),q_s}, \sigma_{s,\psi_s(m)}\right)\Vert P)$ is the building block of the cost function for the choice of moves (see Figure \ref{fig:Fig_1b}, and appendix \ref{def-perceived-value-uncertainty} for further details).  

As we shall see below, when there are several entities, and the subject simulate and take into account others' perspectives, a mean of the different Kullback-Leibler divergence can be taken. 

\begin{figure}
    \centering
    \includegraphics[width=1\textwidth]{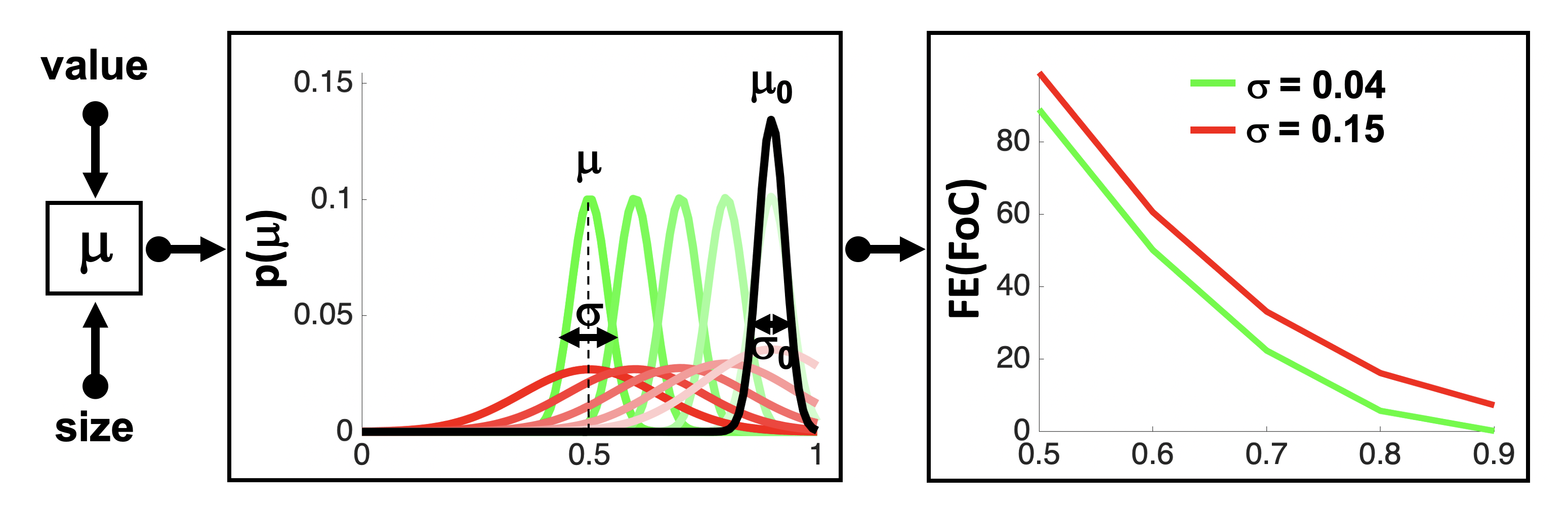}
    \caption{\textbf{Perceived value and cost function}}
    \medskip
    \small
    \raggedright
    Simplified presentation of the overall appraisal processing pipeline (see section \ref{mat-and-meth} or Appendix \ref{model-tech-details} for details). The projective size parameter and prior on the intrinsic value of an object are integrated into parameter $\mu$. \textit{Left chart.} The value of $\mu$ is used to define the maximum of a probability density $p(\mu)$, representing the actual or expected state of the agent, with dispersion proportional to the parameter $\sigma$, which quantifies uncertainty with respect to sensory evidence. Several examples are shown with red and green traces. Likewise, the preferred or ideal state is also defined as a probability density (black trace) with maximum at $\mu_0$ and dispersion $\sigma_0$. \textit{Right chart.} FE(FoC), the FE associated with a given FoC, is computed by comparing the agent's state density to the preferred prior probability density, using Kullback-Leibler Divergence. The closer the current state to the preferred one, the lower the FE. Importantly, when the preferred prior density imposes low uncertainty (small $\sigma_0$), for the same value of $\mu$, state densities with larger uncertainty (red traces) tend to be associated with higher FE than those with lower uncertainty (green traces). In other words, minimizing free energy also favors choosing actions minimizing uncertainty, which entails a form of intrinsic curiosity in the agent.    
    \label{fig:Fig_1b}
\end{figure}

\subsection{Algorithm: social inference and action selection steps}\label{intro:foc-tom}\label{model:perceive-value-cost}\label{model:action-optimization}

In the PCM, an agent has an internal model of its own beliefs and preferences, i.e. expectations. This model is used to quantify FE as a function of perspective taking and action based on the FoC. One contribution of this report is to demonstrate how the PCM principles explained above can be naturally generalised to social-affective perspective taking as a method for ToM in embodied agents. Different components of ToM have to be considered and related to the model: 
\begin{enumerate}
    \item the ability to infer preferences of others about entities from their spatial and emotional behavior (see section \ref{model:inference-step}); 
    \item the ability to predict other agents' behaviours through simulation of their FoC;
    \item factors of influence of other agents on the subject and its behaviors related to ToM (see section \ref{model:action-optimization}). 
\end{enumerate}

Internal models of beliefs and preferences attributed to other agents, including of other agents' beliefs about the beliefs and preferences of other agents, as well as the degree to which they are expected to perform ToM and to influence the agent, can be integrated in the agent's belief system. The agent can then use the same mechanisms it uses for itself recursively, in order to make inferences about others' states and predict their actions, by attributing perspectives, i.e. states of the FoC, to others, and using the corresponding internal models of (true or false) beliefs attributed to them to simulate how the others would appraise their environment based on their FoC. These principles relate to simulation theories of empathy, social perspective taking, affective learning, and emotion regulation, as reviewed above. In other words, a PCM agent will assume that other agents are equipped with the same general algorithm (or mind) as themselves, but with potentially different prior beliefs and preferences, and of course, different perspectives. The beliefs and preferences attributed to others can be updated via social feedback, by combining contextual, affective and spatial sensory evidence, e.g. another agent's emotion expression and orientation in relation to an object.

\subsubsection{Inference and update of preference step}\label{model:inference-step}

In the inference step of the active inference process, the subject infers an agent's internal preference about a given entity through the integration of the emotions (valence) that agent expresses and the spatial relationships between the subject, the agent and the entity. For this report, emotion expression and the inference of preferences from expressed emotions was implemented in simple and deterministic manner, in order to reduce complexity and computational cost. Emotion expression is a simple function of perceived value (see appendix \ref{emo-express} for technical details). Thus, as a first approximation, we do not assume that the expression of emotions by the subject is itself part of its optimisation process (a subject will always express an emotion that is consistent with its state of appraisal). The inference of preferences in another agent is a function of the other agent's expressed emotional valence, previous preferences attributed to the agent, and of the relative spatial uncertainty in terms of access to sensory evidence existing between the perspective of the subject on the agent and that of the agent toward a definite entity. See Figure \ref{fig:Fig_2a} for an illustration of the process (see appendix \ref{pcm:inverse-inference}, \ref{action-decision} for technical details). 
To account for the influence other agents have on the update of preferences of a subject, we consider a matrix $(I^p_{sab}\in [0,1], a,b\in A)$: it takes into account the influence that the subject believes another agent $b$ has on agent $a$ (which could be the subject itself) during the inference of preferences. The true influence that agent $a$ has on $s$ is given by the vectors $(I^p_{ssa},a\in A)$. The collection of all influence matrices over the subject $s\in A$ is a 3 dimensional tensor, which, in this implementation, is not updated during active inference, but fixed once and for all (see also appendix \ref{tom-factors}).

This tensor stem from the following psychological considerations. A subject's own preferences can be modified according to the preferences it attributes to others. For instance, a person can start liking something initially indifferent or disliked because another person, e.g. an admired or respected figure, expresses such liking. In other words, a person can embrace, or be inspired by the preferences of somebody else. Since a subject's preferences influence its drive, this influence of others' preferences on its own preferences will also impact its action. We will call this tensor the tensor of influence on preferences.

\begin{figure}
    \centering
    \includegraphics[width=0.8\textwidth]{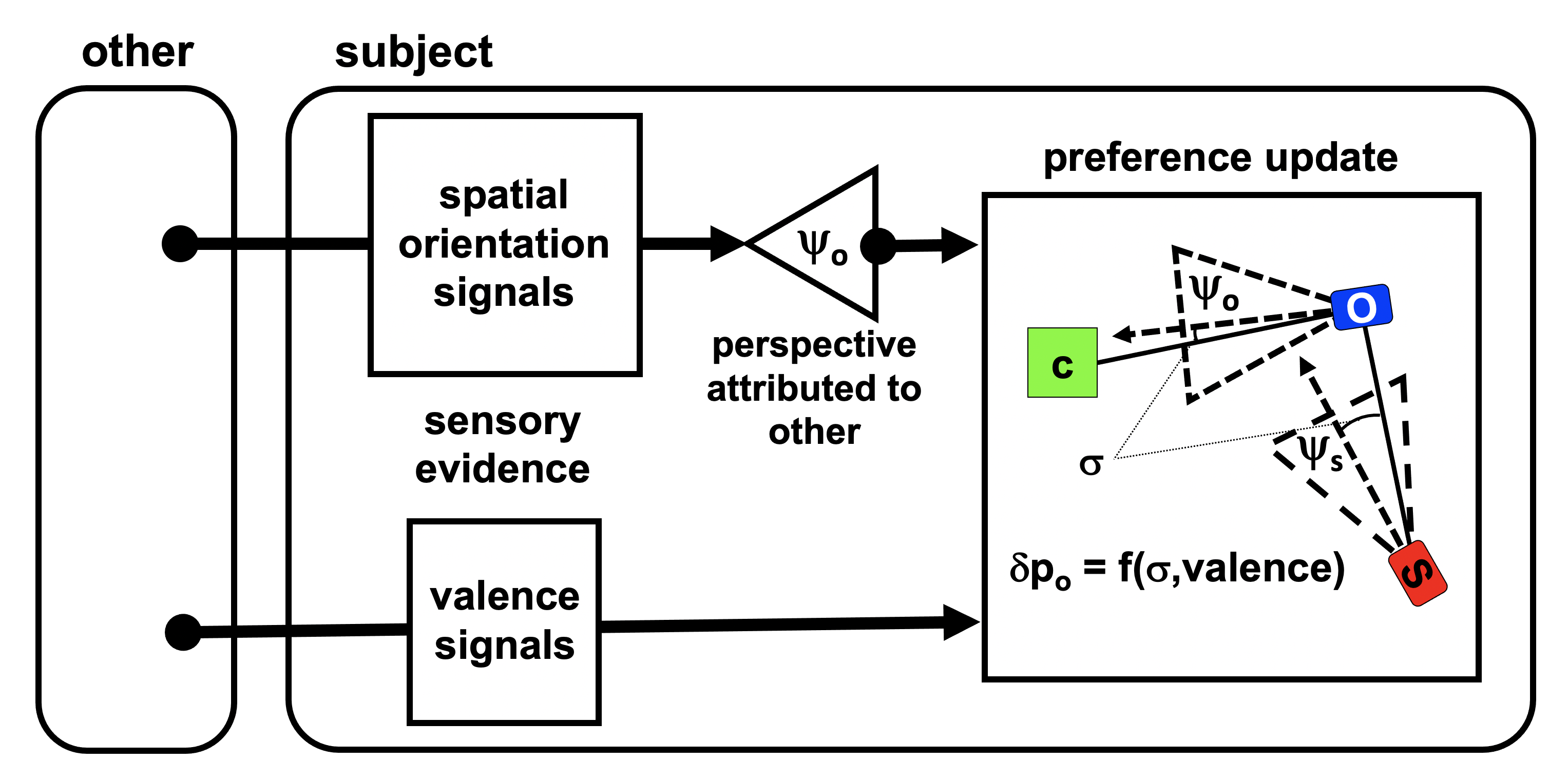}
    \caption{\textbf{Illustration of the preference update process}}
    \medskip
    \small
    \raggedright
     Prior beliefs about the preferences of another agent are updated by combining spatial orientation and affective valence signals as sensory evidence. Parameters of the subject's (S) current projective transformation $\psi_s$ are compared to the inferred parameters of the projective transformation of the other (O) $\psi_o$ to estimate an uncertainty $\sigma$ combining the uncertainty $\sigma_sa$ of the subject with respect to the other agent, and the uncertainty $\sigma_ae$ regarding the spatial relationships between the other agent and its direction of aiming towards an entity $e$ (here object $c$), which weighs in the update.        
    \label{fig:Fig_2a}
\end{figure}

Of note, this mechanism of update of preferences attributed to others implements a basic form of error correction regarding predictions of others' behaviours. For instance, if a subject believed that another agent considered an entity as neutral, but then observes the other agent expressing a positive emotion while being oriented toward the entity, it will progressively attribute to the other agent a positive preference for the entity.  

\subsubsection{Action selection steps}
\label{model:action-optimization}

The action selection step of the active inference process includes two components: the move of the subject and the expression of its emotion. The space of actions is the space of moves a subject can make times the possible emotions it can express, emotions taking values in $[0,1]$ (see appendix \ref{emo-express}). For both components, the organisation of the internal world plays a crucial role for action selection. 

As we shall see, the total FE associated with multiple perspectives taken through the FoC can then be quantified by integrating not only the FE associated with the actual or anticipated actions of an agent, as an isolated system, but also the FE that the agent expects to be associated with others' FoCs, as a function of actions and world contingencies. In other words, the action of the subject may be driven by how it expects others to act and react, including by taking into account what the subject believes about what other agents believe about it and each other, and so on (up to the limit of its own computational power). 

An agent thus can take the opinion of other agents into account in the planning of its own actions. It expects the other agents to infer and express emotions in the same way as it does. In order to simplify presentation and simulations, here we assume that the moves of the subject only take into account other agents that can at most perform ToM of order 1, that is subjects can at most perform ToM of order 2. Thus, in this simplified model, according to the subject, the other agents do not develop a strategy for future actions based on ToM, i.e. another agent can simulate the state of another agent at a given time, but do not use this simulated state to predict its future behaviour. 

Similarly to the tensor of influence on preferences, which was introduced above in the section \ref{model:inference-step} on inference and update of preferences, we introduce an influence matrix on actions (or moves), $(I^m_{sab}\in [0,1], a,b\in A)$, for a subject $s$, which takes into account the influence that the subject believes another agent $b$ has on agent $a$ (which could be the subject itself), during the choice of its moves, irrespective of the agent's $a$ own preferences. This influence tensor is not updated during active inference but fixed at the beginning of the simulations. The psychological motivation for considering such tensor is that a subject may not change its own preferences based on different preferences expressed by others, but nevertheless take them into account in the way it programs its actions. For instance, a person may choose not to act according to personal preferences because this may displease another person, e.g. a figure of authority. This tensor, which we will call the tensor of influences on actions, is independent from the tensor of influences on preferences. 

Let us denote configurations of entities in the Euclidean space as $(X_e\subseteq \R^3,e\in E)$.
An agent $a$ with first order ToM chooses its moves by minimizing the following cost function:


\begin{equation}
C(m,p_{a..}, I^m_{aa.},X)=\underset{b\in A}{\sum}\underset{\substack{e\in E\\e\neq b}}{\sum}  \omega_{b,e}\DKL(Q(.|\mu_{a,\psi_b(m),p_{a,b,.}}(X_{e,m}), \sigma_{b,\psi_b(m)}(X_{e,m})\Vert P),
\end{equation}

which is defined for moves $m$ in a set of moves $N$ and where $X_e(m)$ is the configuration of the entity once the move of the agent $a$ is accomplished; the only projective chart that changes when an agent moves is its own, in other words only $\psi_a(m)$ really depends on $m$, and similarly the only configuration that depends on the move of the agent is the agents own configuration $X_a(m)$. Here

\begin{equation}
\omega_{b,e}=I^m_{a,a,b}\alpha_{b,e}
\end{equation}

and $\alpha_{b,e}=\frac{1}{\vert E \vert-1}$.

Let us breakdown what the cost $C$ is: up to a multiplicative constant, it is a weighted average over all agents of how pleasant and certain a situation appears in the perspective of the agents, i.e. a mean over all agents $b\in A$ of,
\begin{equation}
\sum_{e\in E} \alpha_e \DKL(Q(.|\mu_{b,\psi_b,p_{a,b,.}}(X_e), \sigma_{b,\psi_b}(X_e))\Vert P) 
\end{equation}

The subject does not have access to the true $p_{a,..}$ nor does it have access to $I_{aa.}^p, I_{aa.}^m$, instead, as a first approximation, it estimates it by $\tilde{p}$, $\tilde{I}^p$, $\tilde{I}^m$ as follows: 
$\tilde{p}_{a..}= p_{s..}$, $\tilde{I}_{aa.}^p= I^p_{sa.}$, $\tilde{I}^m_{aa.}=I^m_{sa.} $.

We consider subjects $s\in A$ that can plan $n$ cycles of perception and action ahead in order to decide what move to make in the first cycle. To do so, a subject imagine the way other agents would react to its moves at each time step, based on preference and influence tensors, $\tilde{p}_{a..}, \tilde{J}_{aa.}^p, \tilde{I}^m_{aa.}$, as described above, which are attributed at step $0$. From this prediction, it weights each sequence of its moves $m=(m^{k}_s, k\in [1,n])$ accroding to the following energy,

\begin{equation}\label{PCM:total-free-energy}\tag{FE}
\FE(m)=\underset{k\in [1,n]}{\sum}a_{k} C_s(m^k,p^k_s,X^k)
\end{equation}

The best first move is the first move of the best sequence of moves $m^*$ (see Appendix \ref{action-decision}).

The number of steps the subject can predict, $n$, is what we call, for short, the depth of processing, and denote as $dp$. (Of note, this concept does not directly relate to the notion of "depth of processing" from the classical level of processing model in psychology and its derivatives \cite{craik1972levels}, which is the subject of a broad specialized literature). 

Figure \ref{fig:Fig_2b} summarizes the general algorithmic principles of the model (see appendix \ref{section:algorithm-pc} for technical details with pseudo-code), and illustrates their expected effect in our robotic context (see also appendix \ref{appendix:cozmos}).

\begin{figure}
    \centering
    \includegraphics[width=1\textwidth]{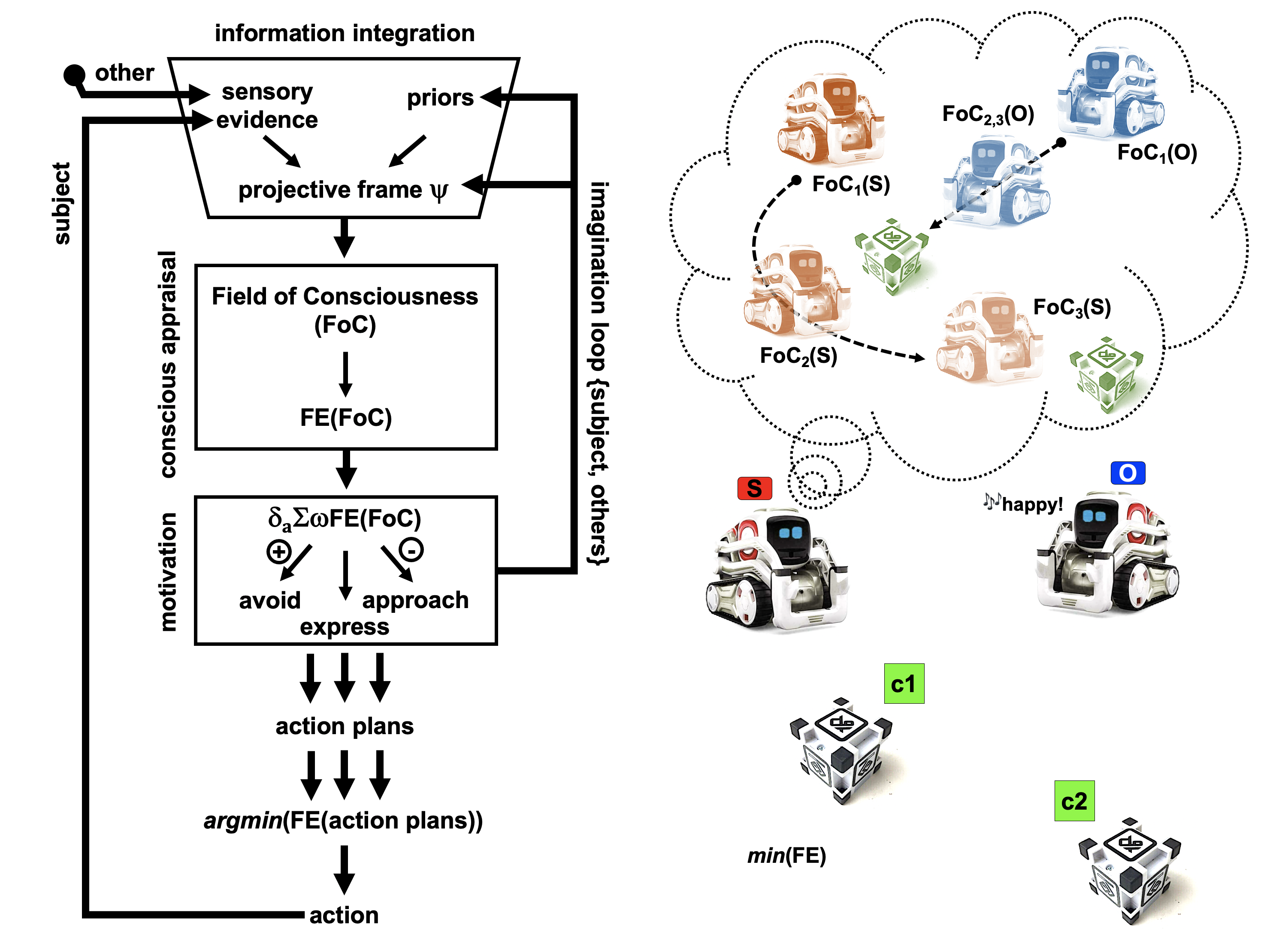}
    \caption{\textbf{Illustrations of overall algorithm}}
    \medskip
    \small
    \raggedright
    \textbf{Left-tier.} Information is integrated by combining prior beliefs and sensory evidence, and framed through a projective transformation $\psi$, for access and appraisal in the FoC. FoCs can be attributed to the subject or other agents. A value of free energy FE(FoC) is associated with the current FoC and averaged with previous values as part of a cumulative process of optimisation aimed at selecting alternate action plans through an internal imagination loop. The contribution of the current FE(FoC) to the average is weighted by the importance $I^m$ attributed by the agent to itself or others, depending on which agent the simulated FoC represents. In other words, agents use, in a variable manner, projective social-affective perspective taking to integrate their updated beliefs about others’ preferences and anticipated behaviours. The current change $\delta_a$ of the average FE with respect to anticipated actions and action consequences (in time and space) motivates behavioural tendencies of approach (if $<0$) or avoidance (if $>0$), and related emotion expression. Imagined action plans are compared based on their overall FE in order to select the next actual action using the plan with minimal FE. Action outcomes from self and others feedback into the process. \textbf{Right-tier.} Illustration of the principles of the algorithm in the robotic context. Robot S (the subject) likes cube $c2$, and robot O (the other agent) cube $c1$. Robot O expresses a positive emotion towards $c1$. Robot S updates its beliefs about the preference of robot O towards $c1$, and starts taking perspective to predict the behaviour of robot O and plan its own actions, by computing FoCs and associated FE(FoC) within its imagination loop. It starts with $FoC_1(S)$ and $FoC_1(O)$, representing respectively robot S and O current perspectives. It iterates to select $FoC_2(S)$ and $FoC_2(O)$, and then $FoC_3(S)$ and $FoC_3(O)$ (note that on the figure, $FoC_2(O)$ and $FoC_3(O)$ are conflated into $FoC_{2,3}(O)$ due to space limitations). Altogether, these series of expected FoCs are associated with the smallest overall FE $min(FE)$, and define the best action plan for robot S, in which robot S goes around $c1$ to avoid running into robot O, and approach $c2$.       
    \label{fig:Fig_2b}
\end{figure}

\subsection{Summary of parameters}

Table \ref{table:1} summarizes the parameters and functions integrated in the model.

\begin{table}
\centering
\begin{tabular}{|l|c|}
\hline 
$p_s$ & Preference tensor of a subject $s$\\
\hline 
$q_s$ & Preference vector of a subject $s$\\
\hline
$q_n$ & Neutral preference\\
\hline
$e_+$,$e_{-}$, $v$& Positive, negative emotion and valence\\
\hline
$I_s^p, I_s^m$ & Preference and action Influence tensor for the subject $s$\\
\hline 
$\psi$ & Projective chart with respect to the Euclidean frame\\
\hline
$\frac{1}{c}$ & Depth of Field of the Projective Chart of a subject.\\
\hline 
$\gamma$  & Attentional focalisation\\
\hline 
  $\mu$ & Perceived preference\\
\hline
$\zeta$ & Acuity\\
 \hline
  $\sigma$ &  Uncertainty on perceived preference\\
    \hline
    $dp$ &  Depth of processing\\
\hline 
 \end{tabular}
\caption{\textbf{Summary of parameters and functions}}
\label{table:1}
\end{table}

\pagebreak

\section{Simulations of adaptive and maladaptive behaviours}\label{simusection}

\bigskip

Our approach followed a scientific rationale, from which we wanted to assess whether theoretical principles, independently motivated, could be the basis of a general generative model that would adequately operate, when implemented in artificial agents, in multiple relevant contexts and in a robust manner. Our aim with the following simulations was to demonstrate how the manipulation of relevant model parameters, based on classical, established psychological rationale and hypotheses, could generate the emergence of a variety of complex, well-document, and relevant adaptive and maladaptive behaviours, predicted by those rationale and hypotheses. It was not to try to predict specific datasets based on optimizations of the model and analyse their predictive power.

We manipulated the depth of imaginary processing, prior beliefs and preferences of the subjects about others and objects, its beliefs about others’ beliefs and preferences, and the implication of social-affective perspective taking, via the preference, preference influence and action influence tensors (see Secions \ref{method:pref-tensor}, \ref{tom-factors}). 

We defined the simulations within different situational contexts (similar to experimental or clinical assessment setups), aimed at demonstrating target model properties, namely: how deeper imaginary perspective taking may help agents to overcome unpleasant obstacles towards a better end state; how the degree of social-affective perspective taking may lead agents to enter joint attention behaviours, and broaden their interest; how negative prior beliefs attributed to others about the subject may lead to social avoidance; and how the understanding of false beliefs in other agents may be exploited by a subject towards a better end state. 

The results focus on the behaviours and states of the subject more than on those of the other agent, both to limit the length of this report, and because the subject is the agent under evaluation, as in an experimental or clinical setting. Results are presented by both analyzing the virtual simulations corresponding to the mind of the agents, and by showing snapshots of replay using the physical robots for more concrete demonstration (see also videos online). We wanted to study the consequences of the parametric manipulations on FE, valence, approach-avoidance and joint attention behaviours. We characterized the overt behaviour of the agents (approach-avoidance, and joint attention) using a set of metrics, which approximated metrics that could be used to characterize the behaviour of real human participants in the context of experimental protocols, as a function of time (see Appendix \ref{section:behav-met-analysis}). 

\subsection{General setting of simulations}

All simulations took place within the same general setting. There was always two agents. The environment comprised a floor plane and two objects (cubes). One agent, the subject, served as an agent of interest for evaluation. We refer to the two agents as the "subject" and the "other agent". The parameters and functions in Table \ref{table:2} were fixed for all simulations and agents. \\

\begin{table}
\centering
\begin{tabular}{|l|c|}
\hline
$q_n$ & 0.5\\
\hline
$e^0_+$,$e^0_{-}$, $v^0$& Initialized at $(0,0,0)$\\
\hline
$\frac{1}{c}$ & 10\\
\hline 
$\gamma$  & The function is the same in every setting\\
\hline 
 $\zeta$ & with $\kappa=20$\\
 \hline
\end{tabular}
\caption{\textbf{General parameters and fixed functions for the simulations}}
\label{table:2}
\end{table}

Let us remark that the influence tensors are used here in a limited manner, as many values in these tensors are set to zero, in order to simplify simulations and their interpretation in the situational contexts we consider.\\

Let us introduce two more parameters:\\

\begin{enumerate}
    \item $\update$ is a parameter that can be True or False and that represents whether or not the subject predicts others' behaviour.
    \item $\scb$ is a parameter that can be True or False and that represents whether or not the subject features self-centered beliefs (see Section \ref{method:sit3}).
\end{enumerate}

By convention, the subject is indexed as agent $1$, the other agent as agent $2$, and the two objects (cubes) as objects $1$ and $2$; in the parameter table they are enumerated (in tuples) in the following order: agent 1, agent 2, object 1, object 2. Furthermore, let us recall the convention that two vectors $u,v\in \R^n$ are equivalent if there is a multiplicative constant $\lambda\in \R^{*}$ such that $u=\lambda v$.

The simulations were performed over $n = 70$ iterations of actions (each including embedded iterations according to the depth of imaginary processing parameter). For each situation, we repeated the simulations $N = 100$ times, with different initial positions and orientations of the agents and cubes across trials (under some constraints related to the ``experimental" design of the situations), in order to characterize the robustness of the behavioural outcomes. All conditions within a given situation and for a given trial shared the same initial positions and orientations, in order to render them comparable on a trial-to-trial basis.    

\subsection{Situation 1: imagining a better place to overcome an obstacle}

\subsubsection{Rationale and design}

The aim of this situation was to assess how projective imagination can contribute to regulate emotional states and the motivation of approach-avoidance behaviour, in order to cope with a situation of approach-avoidance conflict. 

The agents did not perform social-affective perspective taking. The subject was placed on the floor plane in the presence of two cubes, one associated with positive prior preferences (liked-cube 1), and one associated with negative prior preferences (disliked-cube 2). Initial conditions were defined so that the subject started close and facing the disliked cube, while the liked-cube was located farther away behind the disliked cube, which acted as an obstacle, generating a situation of approach-avoidance conflict. The other agent was neutral for the subject. The other agent disliked the two cubes and the subject was neutral for it.

We contrasted three sets of parameterisations of the model:

\begin{itemize}
    \item \textit{Agents with low projective imagination}. The subject had a low depth of imaginary processing ($dp = 1$), i.e. the number of iterations of imaginary projections before action was small. We hypothesized that agents with limited imagination would demonstrate reflex-like behaviours of avoidance in order to minimise their free energy by turning their back on the disliked-cube.
    \item \textit{Agents with intermediary projective imagination}. The subject had an intermediary depth of imaginary processing ($dp = 2$). We hypothesized that agents with intermediary imagination would rarely overcome the obstacle and most of the time would avoid it by running away.
    \item \textit{Agents with higher depth of projective imagination}. The subject had a higher depth of imaginary processing ($dp = 3$). We hypothesized that the subject, which was capable of deeper exploration of possible actions than in the previous case, would appear more resilient, and get around the obstacle, i.e. the disliked-cube, to approach the liked one, and thus further minimise its free energy.
\end{itemize}

Relevant simulation parameters are shown in table \ref{table:3}. 

\begin{table}
\centering

\begin{tabular}{|l|c|}
\hline 
$p_1$ & For agent $1$, $p_{11.}=[0.5,0.5,0.75,0.1]$, $p_{12.}$ is irrelevant. \\
\hline
$p_2$ & For agent $2$, $p_{21.}$ is irrelevant, $p_{22.}=[0.5,0.5,0.1,0.1]$\\
\hline
$I_s^p$ & $I_{saa}^p=0$ for all $a\in A$ such that $a\neq s$ and any $s\in A$ \\
\hline 
$I_s^m$ & $I_{saa}^m=0$ for all $a\in A$ such that $a\neq s$ and any $s\in A$ \\
\hline
$dp$ &  Depth of processing $1$, $2$ or $3$\\
\hline
$\update$ & False \\
\hline 
$\scb$ & False \\
\hline
\end{tabular}
\caption{\textbf{Specific parameters for situation $1$}}
\label{table:3}
\end{table}

\subsubsection{Results}

Figure \ref{fig:Fig_3} shows the results of the simulations. Within the end state, as the depth of processing ($dp$) increased from $1 to 3$, free energy ($FE$) decreased ($t(298) = -23.19; p = 1.06\cdot10^{-68}$), approach towards the preferred cube ($c1)$ increased, both in terms of motion ($t(298) = -12.87; p = 1.89\cdot10^{-30}$) and look at ($t(298) = 26.04; p = 1.89\cdot10^{-79}$), and valence increased ($t(298) = 22.87; p = 1.48\cdot10^{-67}$). The joint attention parameter was not relevant for this simulation, but is indicated in the figure. Note that for the depth of processing of $3$, free energy showed a transitory increase and valence decreased at the beginning of the simulation, expressing the fact that the subject had to overcome a period of disliking while getting around the cube ($c2$) it did not like, motivated by the promise of a better outcome further down the road (see end state).     

\begin{figure}
    \centering
    \includegraphics[width=1\textwidth]{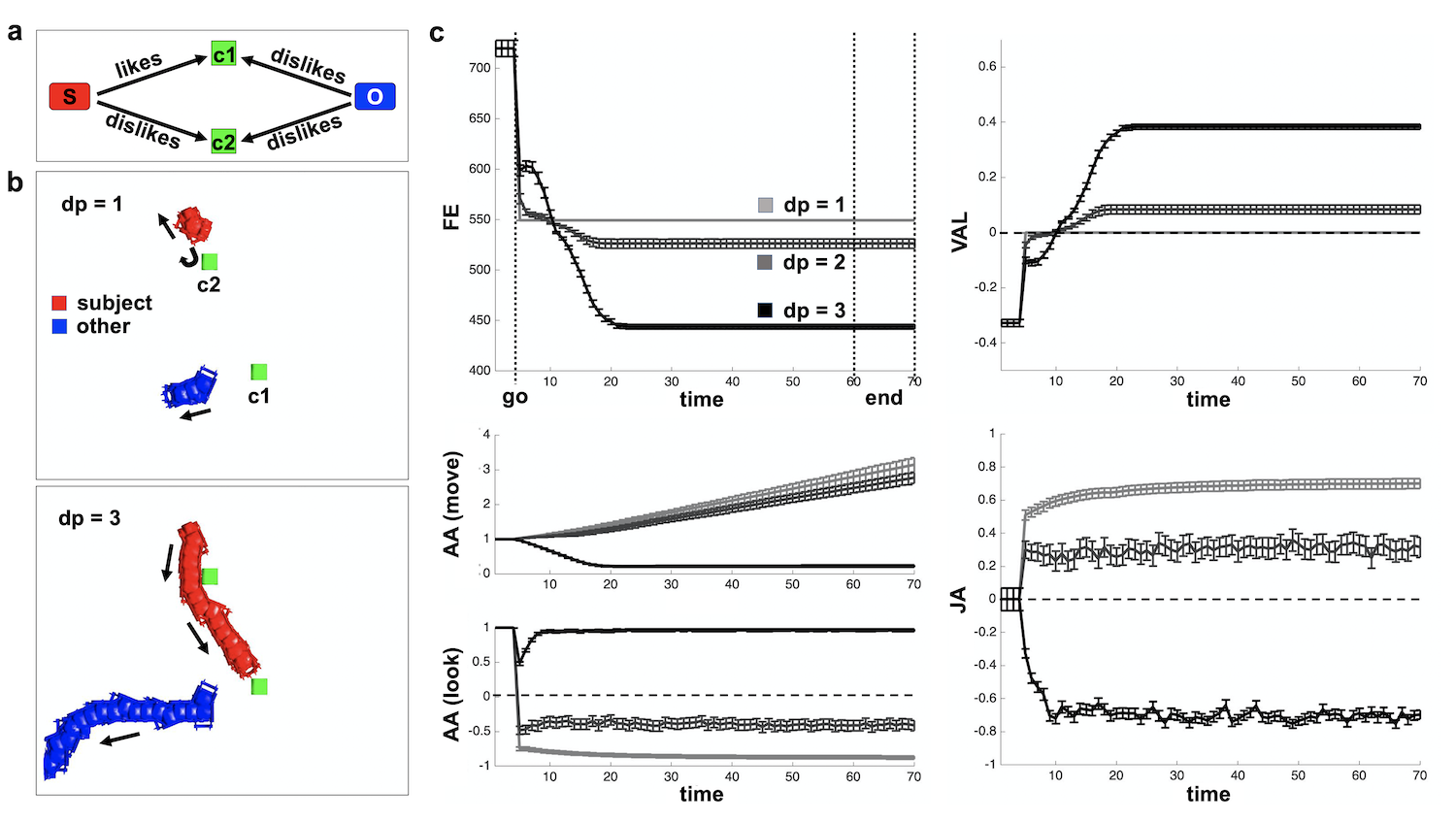}
    \caption{\textbf{Results of simulations for situation 1}}
    \medskip
    \small
    \raggedright
    \textbf{a.} Summary of initial configuration of preferences. S indicates the subject, and O the other agent. Arrows indicate whether an agent likes or dislikes a cube. The absence of arrows between items (cubes or agents) implies a neutral preference. \textbf{b.} Actions of the agents across iterations. Agents are represented across iterations to indicate their paths. Arrows indicate main directions of motion. \textit{Upper-tier.} Depth of processing $(dp) = 1$. \textit{Lower-tier.} Depth of processing $(dp) = 3$. \textbf{c.} Charts of dependent variables characterizing the subject as a function of time (iterations), for the three conditions of depth of processing. Average across trials. Error-bars are standard errors. \textit{Upper-left quadrant.} Free energy (FE). Agents were forced to stay in place for $3$ iterations, before they could 'go' (dashed vertical line). Window between iterations $60$ and $70$ (dashed vertical lines) indicates 'end' state period for statistics. \textit{Upper-right quadrant.} Valence (VAL). \textit{Lower-left quadrant.} Approach-avoidance metrics (AA), type 'move' ($0$ indicates that the subject is closest to cube 1), and type 'look' ($1$ indicates the subject is looking straight at cube 1). \textit{Lower-right quadrant.} Joint attention metrics (JA) (if the two agents are simultaneously looking straight at cube 1 $JA = 1$).       
    \label{fig:Fig_3}
\end{figure}

Figure \ref{fig:Fig_3r} shows snapshots of robot behaviours for a representative trial of situation $1$. 

\begin{figure}
    \centering
    \includegraphics[width=1\textwidth]{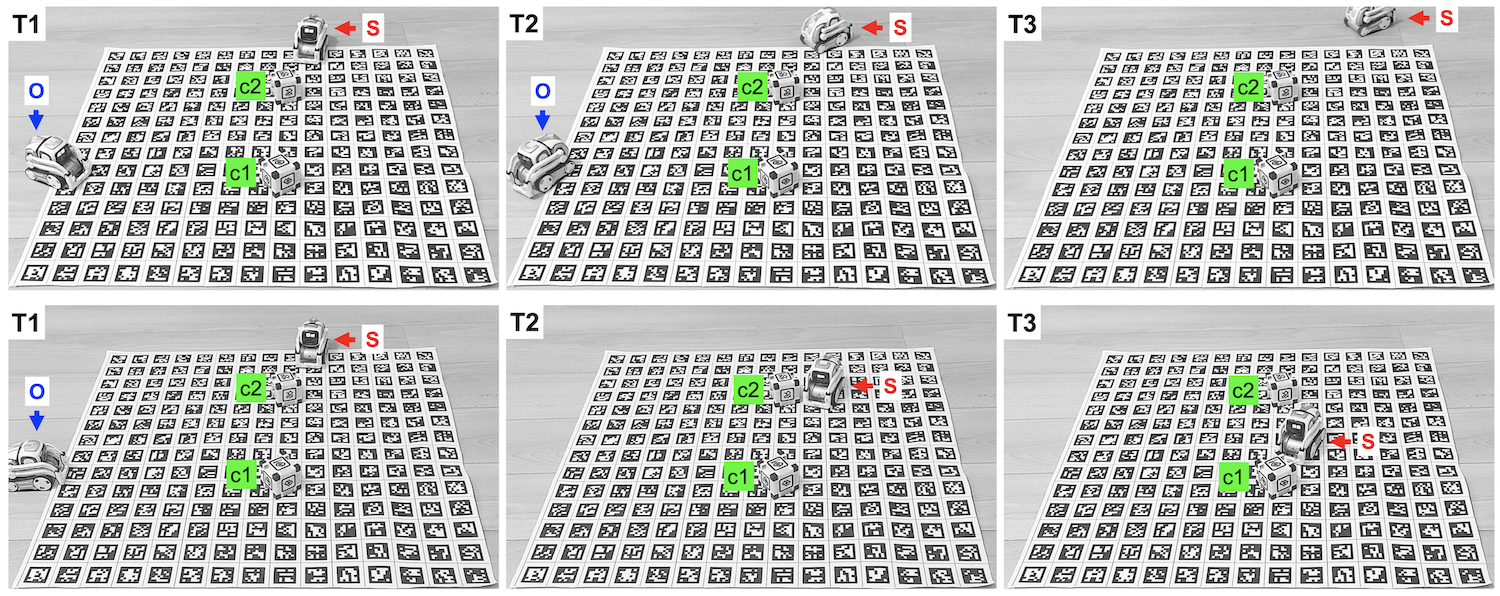}
    \caption{\textbf{Illustration of situation 1 with robots}}
    \medskip
    \small
    \raggedright
    S indicates the subject, and O the other agent. Snapshots are shown for three time points: T1, T2, T3 (note that the time points with the same label, e.g. T1, across conditions, do not necessarily correspond exactly to the same time instant in the simulations, and were chosen for their illustrative value). 
    \textbf{Upper-row.} Condition $dp=2$. Robot S cannot project itself far enough and prefer to fly away from the disliked cube $c2$ ($T2$, $T3$). \textbf{Lower-row.}. Condition $dp=3$. Robot S can project itself far enough, goes around the disliked cube $c2$ ($T2$) to approach and enjoy the liked cube $c1$ ($T3$).      
    \label{fig:Fig_3r}
\end{figure}

\subsection{Situation 2: social-affective perspective taking and joint attention behaviours}

\subsubsection{Rationale and design}

The aim of this situation was to assess the role of social-affective perspective taking, while being influenced by the preferences expressed by others, in the emergence of adaptive joint attention behaviours versus maladaptive restricted interests. 

Initial conditions were as follow. The subject had positive prior preferences about cube 2. The other agent had positive prior preferences about cube 1. 

We contrasted two sets of parameterisations of the model: 

\begin{enumerate}
    \item \textit{Agents without social-affective perspective taking}. None of the agents performed social-affective perspective taking. This condition aimed at simulating impairments typically associated with ASD. We hypothesized that the subject would not develop joint attention behaviours with the other agent, and manifest restricted interest by remaining focused on its initial cube of interest.
    \item \textit{Agents with social-affective perspective taking}. The subject was highly receptive to the other agent interest (non-zero values in the preference influence tensor). We hypothesized a shift of interest in the subject, leading the subject to approach the other agent and its preferred cube, resulting in the emergence of joint attention behaviours.
\end{enumerate}

Relevant simulation parameters are shown in table \ref{table:4}. \\

\begin{table}
\centering
\begin{tabular}{|l|c|}
\hline 
$p_1$ & For subject $0$, $p_{11.}=[0.5,0.5,0.3,0.9]$, $p_{12.}=[0.5,0.5,0.9,0.5]$. \\
\hline
$p_2$ & For agent $2$, $p_{21.}$ is irrelevant, $p_{22.}=[0.5,0.5,0.9,0.5]$\\
\hline 
$I^m$ & $I_{saa}=0$ for all $a\in A$ such that $a\neq s$ and any $s\in A$ \\
\hline
1) $I^p$ & condition 'no tom': $I_{saa}^p=0$ for all $a\in A$ such that $a\neq s$ and any $s\in A$.\\
\hline
2) $I^p$ & condition 'tom': $I_{11.}^p\sim [1,0.075] $, $ I_{12.}^p\sim [0,1]$, $I_{2..}$ is irrelevant.\\
\hline
$dp$ &  Depth of processing $3$\\
\hline
$\update$ & conditions: 'no tom' = False, 'tom' = True\\
\hline 
$\scb$ & False \\
\hline
\end{tabular}
\caption{\textbf{Specific parameters for situation $2$}}
\label{table:4}
\end{table}

\subsubsection{Results}

Figure \ref{fig:Fig_4} shows the results of the simulations. There was a significant difference at the end state, between free energy ($FE$) with \textit{versus} without ToM ($t(198) = 3.88; p = 0.0001$). Though a relatively small difference across means (see Figure \ref{fig:Fig_4}), this suggests that the use of ToM implies a certain cost on the optimisation of preferences and behaviours. Approach by the subject towards the cube preferred by the other agent ($c2)$ increased with ToM, both in terms of motion ($t(198) = -56.36; p = 6.10\cdot10^{-124}$) and look at ($t(198) = 64.04; p = 2.67\cdot10^{-134}$). There was a significant though small difference in valence, which was larger without than with ToM ($t(198) = -2.30; p = 0.02$). Joint attention increased to a maximum when performing ToM ($t(198) = 63.94; p = 3.62\cdot10^{-134}$). Note that, in the middle of the simulations, when the subject performed ToM, $FE$ showed a transitory increase, and look at, valence and joint attention a transitory decrease, expressing the fact that the subject had to overcome a period of disliking while experiencing a change of preferences and the corresponding need to change its action, resulting from the influence exerted by the other agent.   

\begin{figure}
    \centering
    \includegraphics[width=1\textwidth]{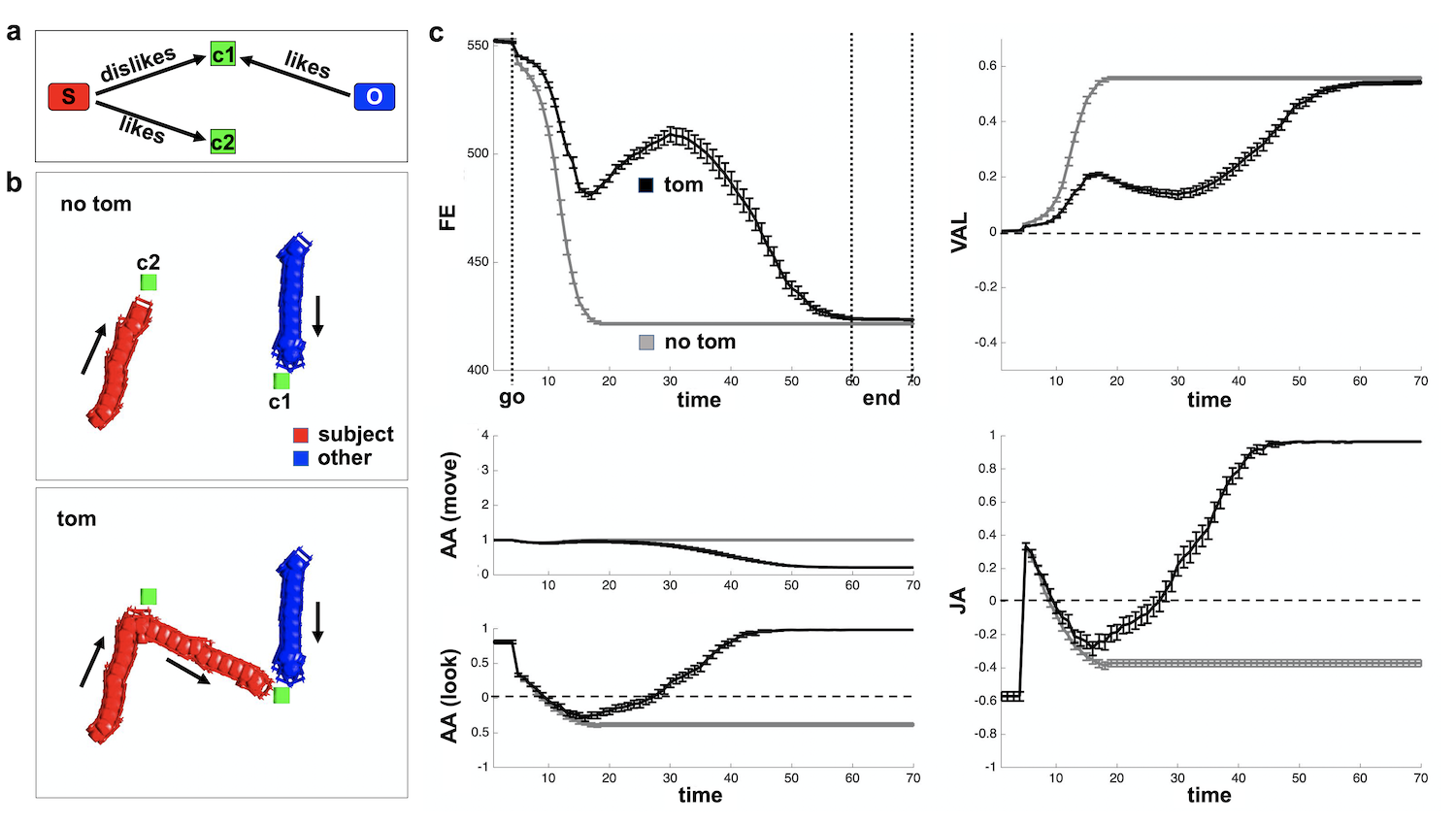}
    \caption{\textbf{Results of simulations for situation 2}}
    \medskip
    \small
    \raggedright
    \textbf{a.} Summary of initial configuration of preferences. S indicates the subject, and O the other agent. Arrows indicate whether an agent likes or dislikes a cube. The absence of arrows between items (cubes or agents) implies a neutral preference. \textbf{b.} Actions of the agents across iterations. Agents are represented across iterations to indicate their paths. Arrows indicate main directions of motion. \textit{Upper-tier.} Condition without ToM (no influence of the other agent on preferences). \textit{Lower-tier.} Condition with ToM (influence of the other agent on preferences). \textbf{c.} Charts of dependent variables characterizing the subject as a function of time (iterations), for the two conditions of ToM. Average across trials. Error-bars are standard errors. \textit{Upper-left quadrant.} Free energy (FE). \textit{Upper-right quadrant.} Valence (VAL). \textit{Lower-left quadrant.} Approach-avoidance metrics (AA), type 'move' ($0$ indicates that the subject is closest to cube 1), and type 'look' ($1$ indicates the subject is looking straight at cube 1). \textit{Lower-right quadrant.} Joint attention metrics (JA) (if the two agents are simultaneously looking straight at cube 1 $JA = 1$).        
    \label{fig:Fig_4}
\end{figure}

Figure \ref{fig:Fig_4r} shows snapshots of robot behaviours for a representative trial of situation $2$. 

\begin{figure}
    \centering
    \includegraphics[width=1\textwidth]{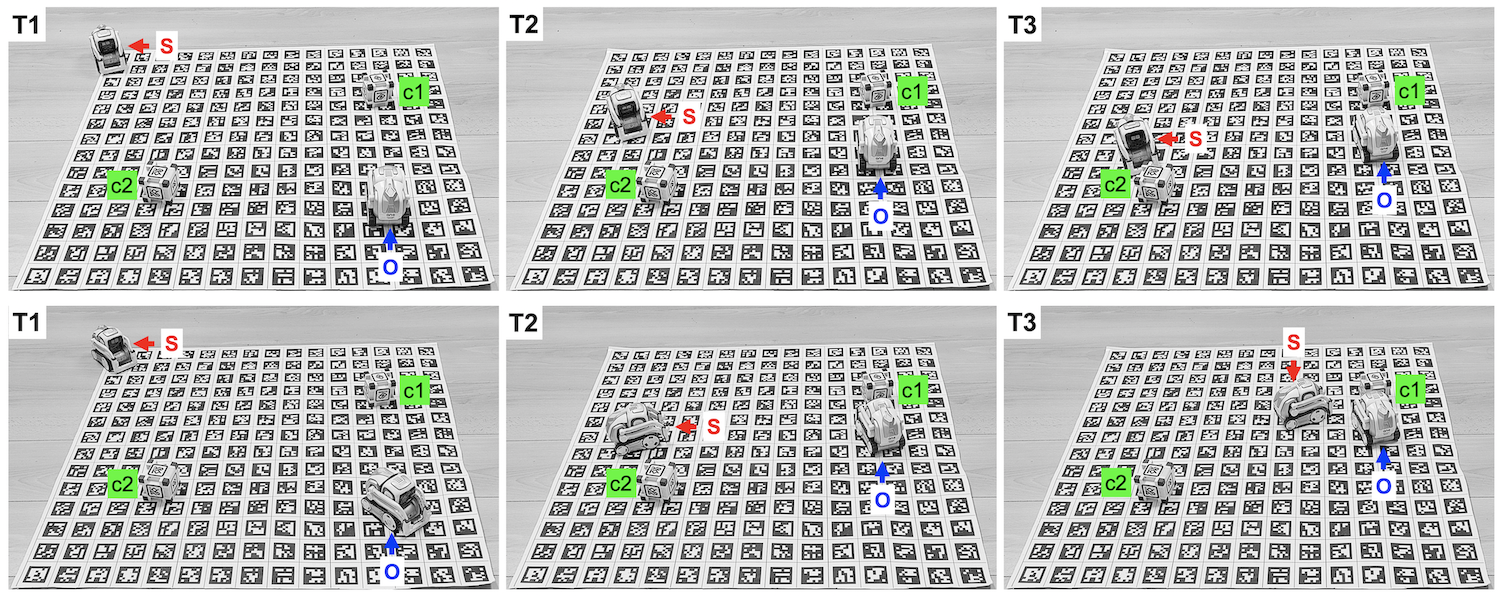}
    \caption{\textbf{Illustration of situation 2 with robots}}
    \medskip
    \small
    \raggedright
     S indicates the subject, and O the other agent. Snapshots are shown for three time points: T1, T2, T3 (note that the time points with the same label, e.g. T1, across conditions, do not necessarily correspond exactly to the same time instant in the simulations, and were chosen for their illustrative value). 
    \textbf{Upper-row.} Condition 'no tom'. Robot S did not perform ToM. It approached its preferred cube $c2$ ($T2$), and remained only interested in that cube ($T3$). \textbf{Lower-row.} Condition 'tom'. Robot S performed ToM. After approaching $c2$, it oriented towards the cube preferred by the other robot O $c1$ ($T2$), approached it, and showed joint attention behaviours with robot O towards $c1$ ($T3$).            
    \label{fig:Fig_4r}
\end{figure}

\subsection{Situation 3: attribution to others of negative prior about the self and maladaptive avoidance behaviours}\label{method:sit3}

\subsubsection{Rationale and design}

The aim of this situation was to assess the role of the attribution to others of negative prior preferences about the self in the emergence of maladaptive avoidance behaviours. We designed the situation as a toy model of SAD.

Initial conditions were as follow. The subject had positive prior preferences about cube 1. The other agent had positive prior preferences about cube 2. The other agent was neutral about the subject. Both agents performed social-affective perspective taking. The subject was influenced by the other agent at the level of the action influence tensor. 

We contrasted three sets of parameterisations of the model: 

\begin{enumerate}
    \item \textit{No attribution of negative priors about itself by the subject to the other agent}. The subject did not attribute negative prior preferences about itself to the other agent, i.e., was not defiant of the other agent (condition: 'no def'). We hypothesized that the subject would not avoid the other agent whenever it would get closer to the subject.  
    \item \textit{Attribution of negative prior about itself by the subject to the other agent without self-centered beliefs}. The subject attributed negative prior preferences about itself to the other agent, i.e., was defiant of the other agent beliefs about itself but did not believe that the other agent was constantly thinking about (imagine) the subject (condition: 'def'). We hypothesized that the subject would avoid the other agent whenever it would actually be in its field of view and get closer to the subject, but would be able to sneak towards its preferred cube through less direct paths than for the first parameterisation of the model.   
    \item \textit{Attribution of negative priors about itself by the subject to the other agent with self-centered beliefs}. The subject attributed negative prior preferences about itself to the other agent, and did believed that the other agent was constantly thinking about (imagine) the subject (condition: sad). This was implemented by making the subject, when performing social-affective perspective taking, always imagine that the other agent was aiming at the subject, i.e. in its imagination the subject attributed to the other agent a FoC that was always oriented towards it. This condition aimed at simulating processes and behaviours typically associated with SAD. We hypothesized that the subject would avoid the other agent whenever it would be around. 
   
\end{enumerate}

Relevant simulation parameters are shown in table \ref{table:5}. \\

\begin{table}
\centering
\begin{tabular}{|l|c|}
\hline 
$p_1$ & For agent $1$, $p_{11.}=[0.5,0.5,0.9,0.5]$, $p_{12.}=[0.15,0.6,0.5,0.9]$. \\
\hline
$p_2$ & For agent $2$, $p_{21.}=[0.5,0.5,0.9,0.5]$ is irrelevant, $p_{22.}=[0.5,0.6,0.5,0.9]$\\
\hline
1) $I^m$& 'no def' $I_{saa}=0$ for all $a\in A$ such that $a\neq s$ and any $s\in A$\\
\hline
2),3)$I^m$& 'def' and 'sad' $I_{11.}^m\sim[1,12.5]$, $I_{12.}$ irrelevant\\
\hline
$I^p$ & $I_{saa}=0$ for all $a\in A$ such that $a\neq s$ and any $s\in A$ \\
\hline
$dp$ &  Depth of processing $3$\\
\hline 
 $\update$ & True\\
\hline 
$\scb$ & In conditions 'no def' and 'def' is False, in condition 'sad' is True.\\
\hline
 \end{tabular}
\caption{\textbf{Specific parameters for situation $3$}}
\label{table:5}
\end{table}

\subsubsection{Results}

Figure \ref{fig:Fig_5} shows the results of the simulations. As the level of social anxiety increased from condition 'no def', to 'def' and 'sad' (see Section \ref{method:sit3} for interpretation), free energy ($FE$) increased ($t(298) = 7.39$ ; $p = 1.52\cdot  10^{-12}$), approach towards the preferred cube ($c1$) decreased, both in terms of motion ($t(298) = 23.69$; $p = 1.73\cdot 10^{-70}$) and look at ($t(298) = -24.96$; $p = 5.25\cdot 10^{-75}$), and valence decreased ($t(298) = -23.52$; $p = 7.21\cdot 10^{-70}$). The joint attention parameter was not relevant for this simulation, but is indicated in the figure. Note that for the condition 'def' (defiant without self-centered beliefs), free energy showed a transitory increase at the beginning, expressing the fact that the subject had to overcome a period of disliking, and choose a more indirect path to approach its preferred cube ($c1$) while avoiding the other agent, which the subject expected to approach its own preferred cube ($c2$), implying the other agent would get closer to the subject. 
\begin{figure}
    \centering
    \includegraphics[width=1\textwidth]{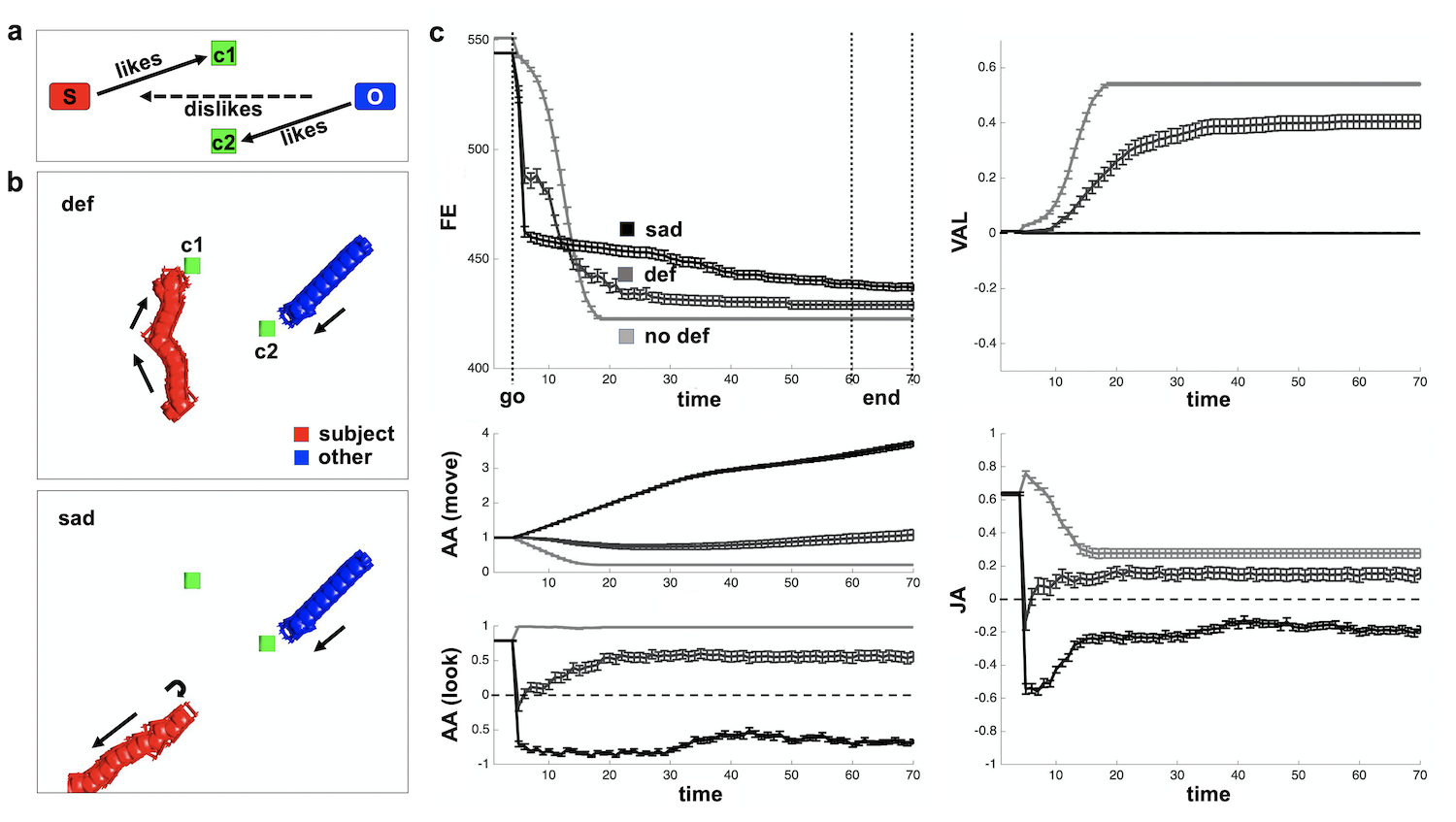}
    \caption{\textbf{Results of simulations for situation 3}}
    \medskip
    \small
    \raggedright
    \textbf{a.} Summary of initial configuration of preferences. S indicates the subject, and O the other agent. Arrows indicate whether an agent likes or dislikes a cube or an agent. The dashed arrow pointing from O to S indicates beliefs held by S about O. The absence of arrows between items (cubes or agents) implies a neutral preference. \textbf{b.} Actions of the agents across iterations. Agents are represented across iterations to indicate their paths. Arrows indicate main directions of motion. \textit{Upper-tier.} Condition 'def' in which S is defiant of O (see Section \ref{method:sit3} for interpretation). \textit{Lower-tier.} Condition 'sad' in which S is 'socially anxious' concerning O. \textbf{c.} Charts of dependent variables characterizing the subject as a function of time (iterations), for the three conditions 'no def', 'def', and 'sad'. Average across trials. Error-bars are standard errors. \textit{Upper-left quadrant.} Free energy (FE). Agents were forced to stay in place for $3$ iterations, before they could 'go' (dashed vertical line). Window between iterations $60$ and $70$ (dashed vertical lines) indicates 'end' state period for statistics. \textit{Upper-right quadrant.} Valence (VAL). \textit{Lower-left quadrant.} Approach-avoidance metrics (AA), type 'move' ($0$ indicates that the subject is closest to cube $c1$), and type 'look' ($1$ indicates the subject is looking straight at $c1$). \textit{Lower-right quadrant.} Joint attention metrics (JA) (if the two agents are simultaneously looking straight at cube $c1$ $JA = 1$).        
    \label{fig:Fig_5}
\end{figure}

Figure \ref{fig:Fig_5r} shows snapshots of robots behaviours for a representative trial of situation $3$. 

\begin{figure}
    \centering
    \includegraphics[width=1\textwidth]{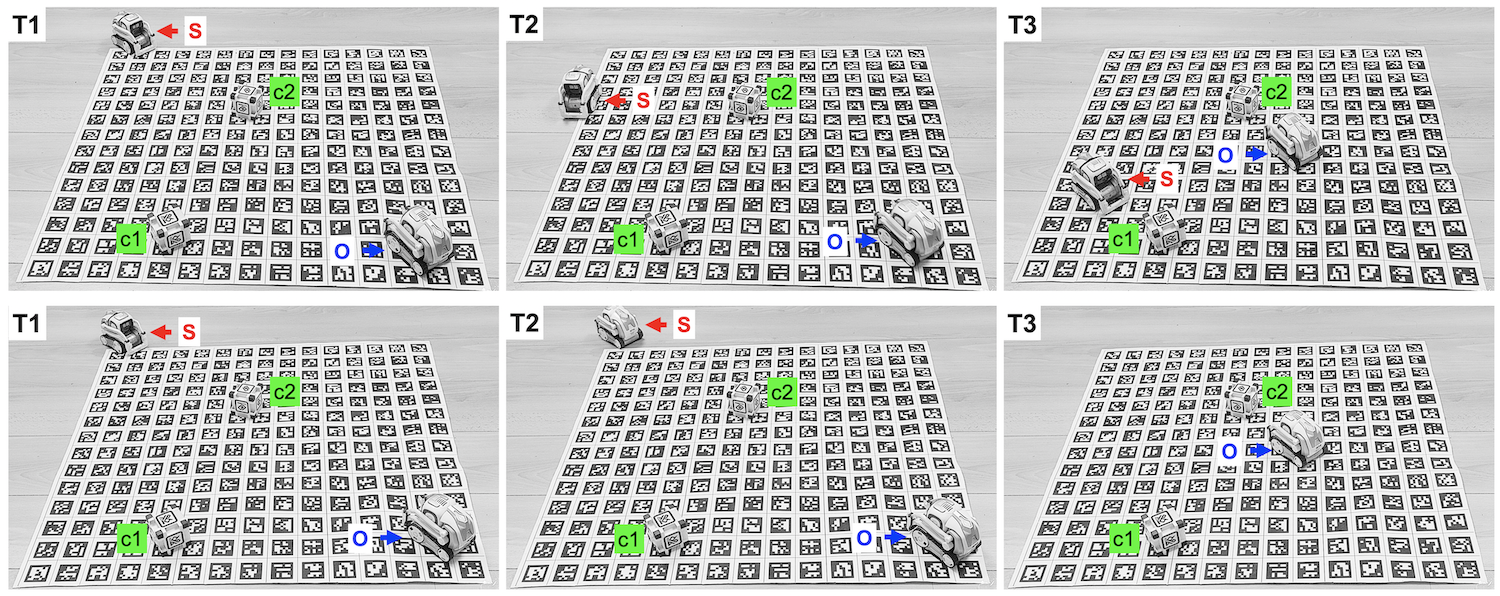}
    \caption{\textbf{Illustration of situation 3 with robots}}
    \medskip
    \small
    \raggedright
    S indicates the subject, and O the other agent. Snapshots are shown for three time points: T1, T2, T3 (note that the time points with the same label, e.g. T1, across conditions, do not necessarily correspond exactly to the same time instant in the simulations, and were chosen for their illustrative value). 
    \textbf{Upper-row.} Condition 'def'. Robot S approached its preferred cube $c1$ via a detour to avoid being seen too directly by robot O ($T2$-$T3$). \textbf{Lower-row.} Condition 'sad'. Robot S could not approach $c1$, but only fly away ($T2$-$T3$), as it assumed that robot O was constantly paying attention to it even if robot O was not overtly oriented towards it (self-centered beliefs, see Section \ref{method:sit3} for interpretation).     
    \label{fig:Fig_5r}
\end{figure}

\subsection{Situation 4: taking advantage of false beliefs in others}

\subsubsection{Rationale and design}

The aim of this situation was to assess the ability of an agent to naturally (automatically) take advantage of its implicit understanding of another agent's false beliefs based on its ability to infer the expected behaviour of the other through social-affective perspective taking.

We modeled the situation based on Sally and Anne's paradigm (see Section \ref{intro}), which we adapted to a non-verbal context focusing on approach-avoidance behaviours. The two cubes were made indistinguishable by the agents except based on their locations. Cube 2 was associated with positive prior preferences for both agents. Cube 1 was neutral. Both agents believed that the other agent also had positive preferences for cube 2. The subject performed social-affective perspective taking, and was influenced by the other at the level of the action influence tensor. 

We created a competitive situation between the agents, in order to generate an initial conflict of approach and avoidance for the subject, so that the subject would avoid approaching the preferred cube. We did so by introducing prior beliefs in the subject that the other agent held negative prior preferences about it, and by having the subject dislikes the other agent too. In fact, the subject was neutral to the other agent. Thus, even though approaching cube 2 would minimise FE for the subject in isolation, the prediction that the other agent would approach cube 2 would prevent the subject to further approach the cube, by anticipation of the negative reaction of the other agent and of the subject itself if the subject and the other agent would find themselves close to each other.

The simulations were divided in two phases. At some point (transition between phase $1$ and $2$, at iteration $30$), the locations of Cube $1$ and $2$ were switched, and initial conditions were reset (the subject and the other agent were repositioned at their initial location). 

\begin{itemize}
    \item \textit{The other agent has true belief about the location of cube 2}. In this condition, both agents could see at all times what was happening, i.e., they could both witness the switching of the cubes. Thus both the subject and the other agent would maintain true beliefs about the location of Cube 2, and believed that the other had true beliefs about it as a result. We hypothesized that the subject would avoid approaching Cube 2 in both phase $1$ and phase $2$.   
    \item \textit{The other agent has false belief about the location of cube 2}. In this condition, just before the switching between cubes 1 and 2, the other agent was turned around so it could not see the cubes for a short amount of time, while the subject could still observe the cubes and the fact that the other agent was not able to see the switching. As a result, the other agent would still believe that cube 2 was at its initial location, whereas the subject would update its own beliefs about the location of cube 2 but not the beliefs attributed to the other agent about it. In other words, for the subject, the other agent would hold false beliefs about the location of cube 2. We hypothesized that the following behaviours would ensue in the second phase. The subject would approach cube 2, expecting the other agent to (mistakenly) approach cube 1 (believing it is cube 2), and anticipating that it would not find itself close to the other agent as a result.     
\end{itemize}

Relevant simulation parameters are shown in table \ref{table:6}. \\

\begin{table}
\centering
\begin{tabular}{|l|c|}
\hline 
$p_1$ & For agent $1$, $p_{11.}=[0.5,0.1,0.5,0.7]$, $p_{12.}=[0.1,0.5,0.5,0.7]$. \\
\hline
$p_2$ & For agent $2$, $p_{21.}=[0.5,0.5,0.5,0.7]$, $p_{22.}=[0.5,0.5,0.5,0.7]$\\
\hline
$I^p$ & $I_{saa}=0$ for all $a\in A$ such that $a\neq s$ and any $s\in A$ \\
\hline 
$I^m$& $I_{11.}^m\sim [1,6]$ and $I_{12.},I_{2..} $ are irrelevant\\
\hline
$dp$ &  Depth of processing $3$\\
\hline
$\update$& True\\
\hline 
$\scb$ & False\\
\hline
\end{tabular}
\caption{\textbf{Specific parameters for situation $4$}}
\label{table:6}
\end{table}

\subsubsection{Results}

Figure \ref{fig:Fig_6} shows the results of the simulations. When the other agent had true beliefs (condition 'tb') about the location of the preferred cube ($c2$) as compared to when it had false beliefs (condition 'tb'), at the end state, the subject had a higher free energy ($FE$) ($t(198) = 5.39; p = 2.00\cdot10^{-07}$), a lower approach (or higher avoidance) towards the preferred cube ($c2)$, both in terms of motion ($t(198) = 53.77; p = 4.34\cdot10^{-120}$) and look at ($t(198) = -59.10; p = 9.86\cdot10^{-128}$), and lower valence ($t(298) = -65.08; p = 1.28\cdot10^{-135}$). The joint attention parameter was not relevant for this simulation, but is indicated in the figure. In the condition 'tb' the other agent would always approach the preferred cube ($c2$), in both $phase 1$ and $phase 2$ of the simulation (see Figure \ref{fig:Fig_6}), forcing the subject to avoid approaching the cube. In the condition 'fb', the subject could take advantage of the false beliefs of the other agent, and approach the preferred cube during $phase 2$. 

\begin{figure}
    \centering
    \includegraphics[width=1\textwidth]{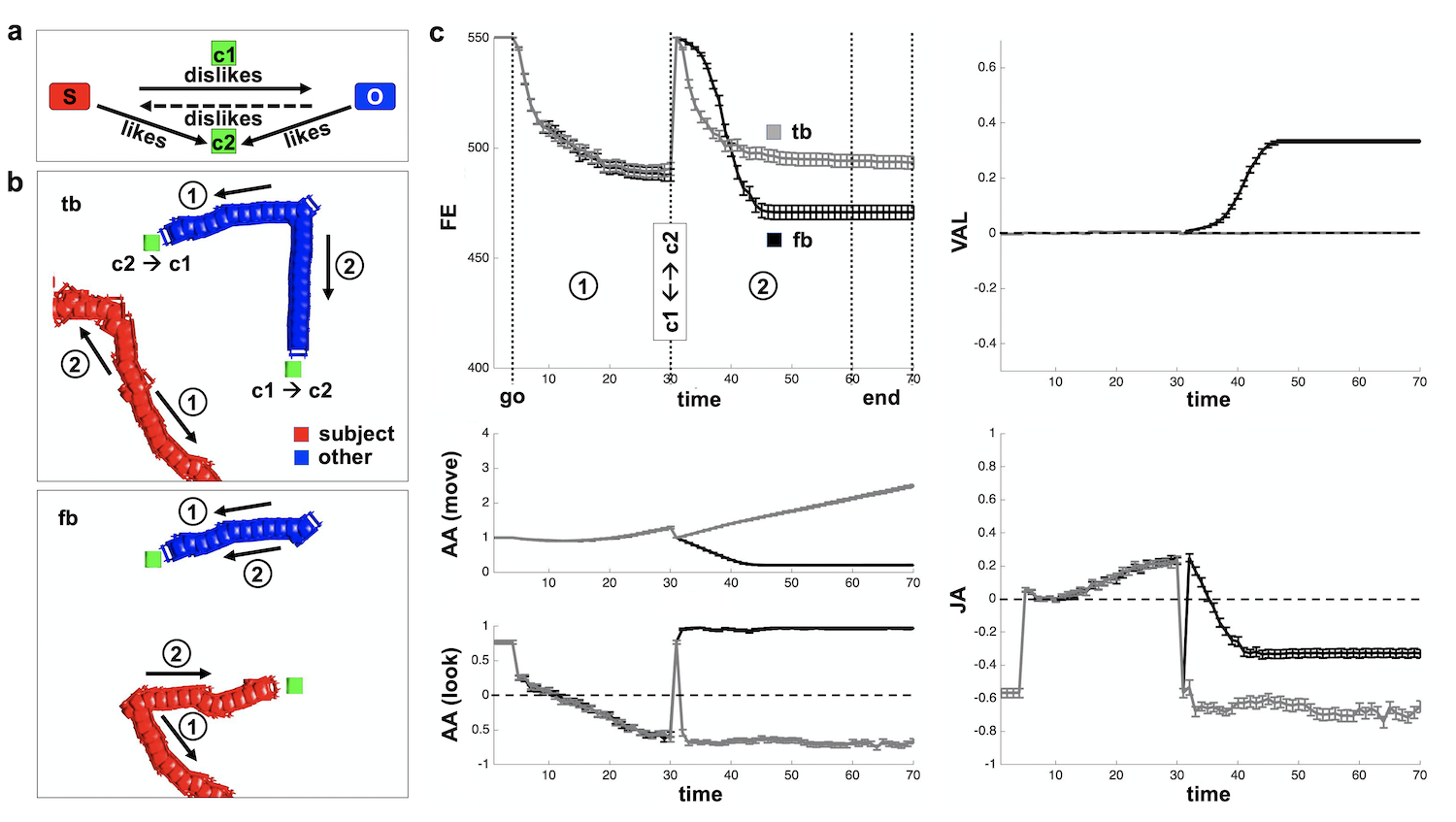}
    \caption{\textbf{Results of simulations for situation 4}}
    \medskip
    \small
    \raggedright
    \textbf{a.} Summary of initial configuration of preferences. S indicates the subject, and O the other agent. Arrows indicate whether an agent likes or dislikes a cube or an agent. The dashed arrow pointing from O to S indicates beliefs held by S about O. The absence of arrows between items (cubes or agents) implies a neutral preference. \textbf{b.} Actions of the agents across iterations. Agents are represented across iterations to indicate their paths. Arrows indicate main directions of motion. $Phase 1$ before the switching of the locations of the cubes, and $phase 2$ after the switching, are indicated with numbers in circles. \textit{Upper-tier.} Condition 'tb' in which O has true beliefs about the location of cube 2. \textit{Lower-tier.} Condition 'fb' in which O has false beliefs about the location of cube 2. \textbf{c.} Charts of dependent variables characterizing the subject as a function of time (iterations), for the two conditions 'tb' and 'fb'. Average across trials. Error-bars are standard errors. \textit{Upper-left quadrant.} Free energy (FE). Agents were forced to stay in place for $3$ iterations, before they could 'go' (dashed vertical line). Window between iterations $60$ and $70$ (dashed vertical lines) indicates 'end' state period for statistics. \textit{Upper-right quadrant.} Valence (VAL). \textit{Lower-left quadrant.} Approach-avoidance metrics (AA), type 'move' ($0$ indicates that the subject is closest to $c2$), and type 'look' ($1$ indicates the subject is looking straight at $c2$). \textit{Lower-right quadrant.} Joint attention metrics (JA) (if the two agents are simultaneously looking straight at $c2$ $JA = 1$).
    \label{fig:Fig_6}
\end{figure}

Figure \ref{fig:Fig_6r} shows snapshots of robots behaviours for a representative trial of situation $4$. 

\begin{figure}
    \centering
    \includegraphics[width=1\textwidth]{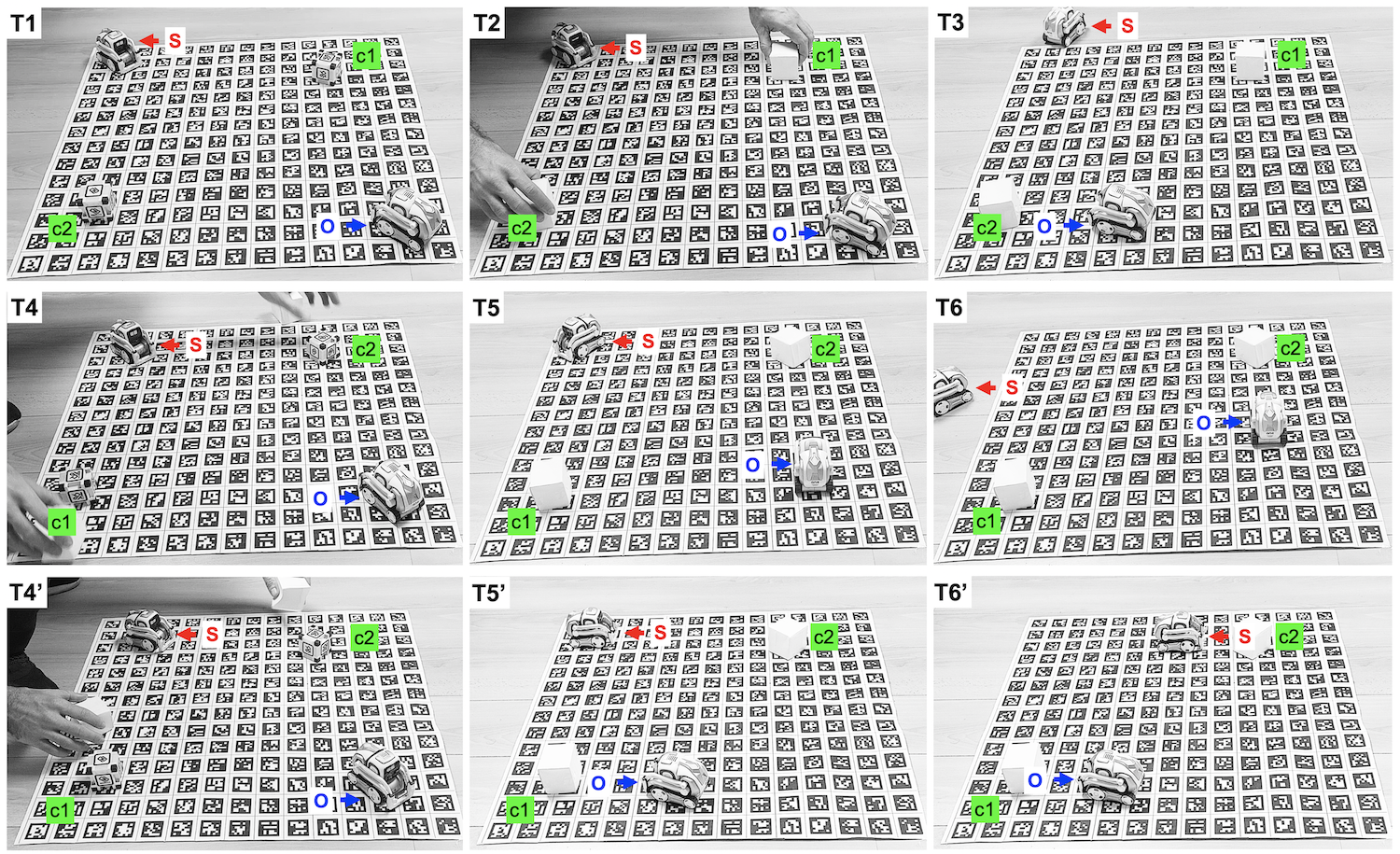}
    \caption{\textbf{Illustration of situation 4 with robots}}
    \medskip
    \small
    \raggedright
    S indicates the subject, and O the other agent. Snapshots are shown for three time points: T1, T2, T3 (note that the time points with the same label, e.g. T1, across conditions, do not necessarily correspond exactly to the same time instant in the simulations, and were chosen for their illustrative value). 
    \textbf{Upper-row.} Condition 'tb' $phase 1$. Robot S could not approach its preferred cube $c2$ and had to fly away ($T3$) as it expected robot O to approach it. \textbf{Middle-row.} Condition 'tb' $phase 2$. After switching the locations of $c1$ and $c2$, robot S still could not approach its preferred cube $c2$ and had to fly away ($T6$) as it expected robot O to approach it. \textbf{Lower-row.} Condition 'fb' $phase 2$. After switching the locations of $c1$ and $c2$, robot S could approach its preferred cube $c2$ ($T6'$) as it expected robot O to approach the wrong cube because of its false beliefs about the location of $c2$.     
    \label{fig:Fig_6r}
\end{figure}

\bigskip

\section{Discussion}

\bigskip

\subsection{Scope and aim of modeling}\label{disclaimer}

First, we discuss the scope and aim of our modeling approach and implementation.

Advances in artificial Intelligence (AI) have lead to the development of systems, based for instance on deep-learning neworks \cite{lecun2015deep}, which sometimes surpass human performance, e.g. in games such as Go \cite{silver}. However these architecture do not entail that these systems operate in a manner that imitates human cognition, let alone consciousness. Instead, current artificial systems most often implement computations and mechanisms of inference, which better correspond to unconscious processing in the brain \cite{dehaene2017consciousness}. 

Progress in modeling human cognition, including consciousness, and its possible role in biological cybernetics, requires to develop systems that primarily aim at mimicking how the human mind operate internally, above and beyond questions of performance in specific tasks and contexts. It does not necessarily entail to tackle the modeling of general intelligence, which is highly variable throughout development and across individuals (of note, current AI models still struggle with achieving many aspects of human infant intelligence \cite{hassabis}). As reviewed in section \ref{princ-mach-consc}, this is particularly true for consciousness, which has a limited capacity of information integration, and is subject to illusions \cite{rudrauf5} and to a variety of errors \cite{dehaene2017consciousness}. Furthermore, one could conceive of conscious systems with very limited intelligence.

Nevertheless, these limitations of consciousness are themselves of scientific interest. Indeed, in our perspective, it is important not only to understand how consciousness may contribute to adaptive behaviors but also to the emergence of maladaptive behaviors. Of note, this would not be typically expected from more general AI approaches, which in pursuing performance, are generally not meant to learn maladaptive functions. However, explicit and interpretable theories and models are warranted to study the mechanisms that could contribute to explain how consciousness may play a role in adaptive and maladaptive development and behavior \cite{rudrauf3}. 

Furthermore, we emphasize that our aim here was not to claim or demonstrate the superiority of the specific method of optimization we used to implement our agents. Our goal was not to implement a model competing with state-of-the-art Artificial Intelligence (AI) algorithms (from Reinforcement Learning and Deep learning, to Markov Decision Processing (MDP)), for instance in learning and performing specific tasks. We are not either claiming that the behaviours that proceed from our algorithm could not be generated through a multiplicity of different algorithmic approaches and methods, from classical to stochastic inference and control approaches. For instance, adaptive control theory and probabilistic learning have been applied to classical game theory frameworks \cite{Freire2019}, demonstrating high predictive power based on learning specific policies. More generally, predictive and machine learning frameworks have been applied to tackle specific practical problems, relevant to adaptive behaviors and survival \cite{hassabis,bach}. For instance, Yoshida et al. \cite{Yoshida2008game} used Reinforcement Learning (RL) and recursive predictions in a two-dimensional competitive digital board game between two agents, which tested capacities of ToM to learn strategies maximizing expected reward in the game. They used their approach for the assessment of ToM in Autism Spectrum Disorder by comparing simulated and empirical performance in the game \cite{yoshida2010cooperation}. However, these approach do not aim at developing a model that would necessarily be faithful to human psychology, and certainly not to consciousness, in its internal modes of representation and operation.

Our goal was to study through an integrative model how the subjective perspective of consciousness may contribute to the mechanisms and functions delineated by consciousness theories, as reviewed in section \ref{princ-mach-consc}, and in doing so might play a causal role in the generation of canonical adaptive and maladaptive behaviours. Our approach did not aim at developing a model of a specific phenomenon or problem that offers an optimal account for that phenomenon or problem, but to develop a more general model, based on a scientific rationale about consciousness, capable of accounting for a variety of phenomena, in a heuristic and constructive manner. We sought on this basis to delineate relevant mechanisms for the control of embodied agents, and wished to model agents that would generate relevant behaviours in a variety of contexts, based on such general mechanisms. As a result, our approach was not framed within specific game theoretic contexts. We chose to focus on modeling agents as subjective systems, capable of modeling and predicting other agents based on their own internal mode of operation following simulation theory \cite{lamm}. But in our approach, an agent is not trying to explicitly cooperate or compete with other agents. It can perform ToM, but do so without a goal of policy optimization as would be found in a game theory setting. Although beyond the scope of this contribution, it might be interesting to study how our approach could apply to different contexts relevant to game theoretic settings, e.g. in a collaborative search context, which would be a somewhat natural setting for our model.

At this point, our approach to the planning of future actions by an agent is naive, as it explores the tree of possible future actions while assuming that other agents react in a deterministic manner to its actions. Therefore, our current method is far from state of the art algorithms such as Monte-Carlo Tree Search \cite{Silver2016} or depth-limited search algorithms for perfect information games. We intend in future developments to consider methods for imperfect information games \cite{Emery-Montemerlo2004}\cite{Brown2018}, e.g. when considering collaborative agents, as it is indeed more realistic to consider that agents only have only access to partial information about the preferences, and more generally beliefs of other agents.

Most importantly, what we wanted to emphasize in this contribution was how the projective modeling of the subjective perspective of consciousness could play a meaningful role within an adaptive or maladatptive cybernetic process in embodied systems, in a manner that could capture known psychological mechanisms, and relate to perspective taking, appraisal, imagination and motivation. Mimicking how consciousness process information and use information to control behaviors may however prove relevant for social and empathetic artificial agents and robotics, and complement other approaches in AI interested in emulating human cognition.

\subsection{Interest of the model regarding the role of consciousness in the generation of adaptive and maladaptive behaviours}

Pursuing this modeling aim, we showed how the principles of the Projective Consciousness Model (PCM) could be applied to derive and implement a model of the subjective perspective of consciousness, operating as a global workspace
\cite{dehaene2017consciousness, merker2021integrated}, which could play a causal role in appraisal, affective and epistemic drives, and social-affective perspective taking. We demonstrated through simulations how that subjective perspective could play a causal role in a process of active inference (see Sections \ref{mat-and-meth}, \ref{intro:foc-motive}, \ref{intro:foc-tom}; and see Figures \ref{fig:Fig_1a},  \ref{fig:Fig_1b}, and \ref{fig:Fig_2b}). The approach was based on the concept of Field of Consciousness (FoC), which puts forth the integral role of 3-dimensional projective geometry in conscious access and processing, in a manner that is consistent with and makes sense of key aspects of the phenomenology of subjective experience.

The projective model could account for, operationalize, and exploit the viewpoint-dependent, conspicuously perspectival phenomenology of the subjective experience of space. In doing so it could also account for documented relationships between the intensity of experienced emotions and the distance of stimuli \cite{teghtsoonian1}, based on subjective properties of the projective manifestation of objects that do not simply conform, in and of themeselves, to the objective Euclidean properties of those object in ambient space (see \ref{bridgegap}). This yielded principles for relating basic mechanisms of appraisal, the imagination, social-affective perspective taking, and the affective and epistemic motivation of action, within a unified framework.  

The application of the concept of FoC to counter-factual perspective taking lead to an operational model of the imagination and its role in the motivation and programming of action. The possibility of attributing different states of FoC and expectations (beliefs, preferences) to other agents enabled agents to perform and leverage Theory of Mind (ToM) \cite{hadwin1}, in a manner that was consistent with simulation theory \cite{lamm}, and theories emphasizing the importance of the combination of spatial and affective information in such process \cite{baroncohen3,ciompi1991affects,berthoz2010spatial}. As a result, we obtained a generative model of basic approach-avoidance and joint attention behaviours. The psychologically inspired manipulation of the model's parameters yielded differentiated adaptive and maladaptive behaviours in simulations.

\subsubsection{Simulations of adaptive behaviours}

 The simulations demonstrated the potential adaptive value of projective imagination in fostering resilience when facing an obstacle ($Situation 1$), the role of social-affective perspective taking in adaptive approach and joint attention behaviours ($Situation 2$), and its strategic effectiveness in the context of competition between agents, when agents could leverage false beliefs attributed to others ($Situation 4$). ToM appeared in our simulations to come with a certain cost, related to the complexity and uncertainty of social-affective perspective taking and its consequences on the need to adapt behaviours accordingly in a flexible manner (see $Situation 2$). Although this is beyond the scope of this contribution, a more systematic and broad analysis of the impact of the model parameters on covert (internal dynamics) and overt (observable behaviours) metrics could yield interesting insights and predictions about the psychological processes encompassed by the model, and lead to new experiments to test its specific predictions.   
 
\subsubsection{Simulations of maladaptive behaviours}

Simulated disruption or biases within the projective process generated behaviours mimicking characteristic features of Autism Spectrum Disorders (ASD) and Social Anxiety Disorders (SAD). 
Our simulations showed that dysfunctional mechanisms of social-affective perspective taking implied, on the one hand, a reduction of the processing complexity and uncertainty entailed by flexible choices of action, and, on the other hand, impaired joint attention behaviours and generated more restricted interests ($Situation 2$). These outcomes are consistent with clinical theories linking this type of psychological mechanisms to symptomatic behaviours in ASD \cite{baroncohen4, lombardo1, chevallier1}. Future research using the PCM framework could further help disentangle the role of different factors and mechanisms in the emergence of differentiated symptomatic behaviours, by analyzing detailed effects of parameter manipulations related to specific, relevant mechanisms. For instance, such approach might potentially contribute to better understand the diversity of clinical presentations in the spectrum. Likewise, it might help assess competing, though not necessarily mutually exclusive theories, from theories emphasizing a primary deficit in perspective taking \cite{baroncohen4} to theories emphasizing deficits in social motivation \cite{moscovitch1}, based on parametric manipulations and the analysis of differentiated behavioural predictions.    

Our simulations also showed how a combination of pro-social motivation (interest in others) and prior beliefs about the negative opinion another agent may have about the subject, could lead to maladaptive social avoidance, in a manner that was consistent with current psychological hypotheses about SAD \cite{moscovitch1}. More targeted simulations, e.g. studying the role of parameters impacting social-affective learning and the ability to update beliefs, but also projective parameters affecting perspective taking and attentional focalisation, might contribute to better understand the role of specific mechanisms and their interactions, as aggravating versus mitigating factors in SAD. 

More generally, research based on the PCM framework might contribute to advance our understanding of psychopathology, and the role of the subjective perspective in it, through a systematic exploration of the impact of parameters on symptomatic behaviours (see \cite{rudrauf3}).

\subsection{Validation rationale and status of the simulations}\label{validation-rationale}

Our main goal was to offer a proof-of-concept showing the relevance and interest of the PCM principles as a basis for human-inspired, interpretable methods of active inference (and more generally cybernetics), based on the integration of a subjective perspective in models of embodied agents. 

We chose to demonstrate that we could reproduce a rather broad spectrum of differentiated, canonical behaviours, characterized by their adaptability, and well-described in developmental and clinical psychology. Somewhat in the spirit of building a non-verbal Turing test \cite{saygin2000turing}, we considered behaviours that would be criteria of evaluation in the context of classical experimental and clinical setups, i.e. approach-avoidance and joint attention behaviours, used to distinguish between groups and individuals, based on clearly discriminant outcome variables. 

The challenge was to show that a pipeline of processing implementing the core projective principles of the PCM could provide a robust generative model of such canonical behaviours. Importantly, these principles were initially derived in an independent manner from the consideration of these canonical behaviours, based on aspects of the phenomenology of consciousness (perspectival structure, perspective taking and imagination). Our fundamental modeling assumptions about consciousness as a projective mechanism allowed us, in a manner that was consistent with global workspace theory \cite{dehaene2017consciousness} and Alexander's axioms \cite{aleksander}, to naturally connect and unify multiple processes that are often studied separately, such as perception, attention, imagination, emotion, motivation and action, and to derive a specific, even though simplified implementation framework. From there, we were in a position to test whether mechanisms central to our model, could play a causal role in the generation of canonical behaviours, further motivated by the fact that mechanisms of that type had long been hypothesized, but mostly in a qualitative manner, to play a critical role in such behaviours. 

As explained above in section \ref{disclaimer}, because human behaviours are context dependent, approaches only attempting to learn those behaviours with variable functions and highest possible precision may not be always the most adequate when the modeling aim is to capture aspects of human psychology, in particular related to consciousness. Instead, we wanted to use a parsimonious theoretical framework based on psychological principles, as the basis of an integrative model, to show that a small number of functions and dimensions may be sufficient to generate a variety of highly relevant, intrinsically variable yet distinct behaviours.  

The ability to generate and discriminate well-known behaviours that are clearly distinguishable experimentally, based on a novel theoretical framework accounting for these behaviours through unified principles and integrating a model of subjective perspective, is not trivial. If the presiding principles could not generate these canonical behaviours in a robust manner, then the theory would have been automatically falsified. Generally speaking, we would like to emphasize that the evaluation of the validity of a model with such modeling aims, should not only be based on its predictive power and accuracy about specific data, but also based on the processes according to which it can predict such data. Such approach, from theory to predictions, is standard science, and has a well-established heuristic value. 

The implemented model components constitute a rather parsimonious if not simplistic implementation of the PCM principles, based on simple functions, with many specifications that were minimalist and somewhat arbitrary, as for emotion processing for instance. However, the generated behaviours were robust. We went through several versions of the detailed implementation before the one we presented here, but in general, the different details on cost functions and emotion processing had a limited impact on the results qualitatively. This is because the general projective framework, the mechanisms of appraisal and motivation through projective geometry, had a stronger impact on outcomes than the specifics of the different subcomponents of the implementation, and constrained the results robustly (specific cost functions were subordinate to the problem, and might perhaps be different in different species or individuals). Moreover, theoretical choices in our approach lead to explanations and robust predictions of behaviours in a parsimonious manner, but also entailed their interpretability. Every part of the model is indeed interpretable. 

Finally, because the model generates behaviours, as a function of internal parameters, which could be empirically observed and quantified, future work, which is beyond the scope of this contribution, should pursue the design of experiments that could be combined with the model, to test and apply its predictions. Model parameters with best predictive power could be retrospectively interpreted from a psychological standpoint, e.g. in the context of psychometric evaluations and training, and to test a variety of hypotheses about relationships between psychological mechanisms and behaviours, for instance using Virtual reality (VR) (see preliminary versions of the approach in our article on the Moon Illusion \cite{rudrauf5}).

\subsection{Current limitations and perspectives} 

Synthetic psychology is an iterative and constructive process, with the overarching goal of modelling psychology in the most encompassing and valid way. However, our goal in this report was only to integrate the principles of the PCM into a proof-of-concept, cybernetic pipeline, for the control of artificial agents. Multiple components had to be implemented in order to render the pipeline operational. But we did not aim at developing each component of the pipeline in a manner that the would systematically compete with state-of-the art implementations of each relevant component separately, e.g. computer vision, machine learning or emotion processing, including appraisal. Instead, we were interested in approaching the problem in a more holistic manner, in order to obtain a first functional solution emphasizing the role and relevance of projective processing, while limiting complexity and computational burden. A non-exhaustive list of current limitations, challenges and perspective, regarding the development of the PCM, can be highlighted (see also sections \ref{disclaimer} and \ref{validation-rationale}).

The version of the model formulated and implemented in this report focuses on a subset of mechanisms. We have not for instance integrated the perceptual inference steps of calibration of projective charts that we modeled in the context of explaining the Moon Illusion and other perceptual illusions \cite{rudrauf5}. In other words, the agents, as we modeled them, would not be subject to perceptual illusions such as the Moon Illusion. Likewise, the generative models of appraisal, emotion and emotion expression we used were limited to liking (intrinsic pleasantness versus unpleasantness), and more generally were too direct and simplistic. Further developing such models was beyond the scope of this contribution, which focused on the role of projective processes. We are working on integrating more sophisticated models of appraisal and emotion processing (see \cite{rudrauf3}), and advanced generative models of emotion expression in virtual humans for VR experiments \cite{tisserand1}. One of our aims is to incorporate emotion expressions as elements of the repertory of actions that is optimised through the projective imagination loop. The goal is to enable agents to control their emotion expressions in a strategic manner as part of the process of active inference. 

Our model yields a general model of agents that could potentially be applied to multiple embodied and implementation contexts, from virtual humans to robotic platforms. Here, we used a specific toy robotic context, based on two Cozmo (Anki) robots and cube objects that come with the platform. The conscious representation of the environment was simulated by instantiating 3-dimensional models of the robots' body and of their cubes as tessellations at a 1:1 scale (see \ref{fig:Fig_1a} for an example of Cozmo robot's high-resolution mesh). The actual physical robots were then controlled based on the outputs of these virtual simulations. However, because our algorithm was computationally intensive, we chose to down-sample the robot meshes to simple bounding boxes. More generally, because of the heavy computational cost of the current implementation, we could not run the simulations in real-time, based on our current access to computational resources, in order to enable direct interactive interfacing with the robotic platform, although such interfacing was implemented (see Appendix \ref{appendix:cozmos}). We ran the simulations offline, and used the physical robots for replay and illustration purposes. Real-time simulations are also essential for enabling direct interactions between humans and artificial agents (robotics or virtual humans). 

Furthermore, being too exhaustive, the optimisation scheme was not realistic from a psychological standpoint, as it explored all possible paths from a given initial condition, at a given depth of processing. We are working on optimisation strategies taking inspiration from psychological rationale, in order to jointly limit the amount of computation and ameliorate the validity of the model. These include heuristics based on contextual time pressure or time consciousness, and the integration of cost functions relating more explicitly to violations of expectations from predictions. The goal is to only involve a deep imagination loop and revision of priors when a significant divergence from prediction is experienced.More generally, we will explore the possibility of using more state-of-the-art methods of optimization, as briefly discussed in section \ref{disclaimer}.

Nevertheless, we believe that the model's implementation presented in this report is a promising basis for future refinements, and a starting point for instance to enable us to compare more specific models of emotion and decision making.

\subsection{Relevance of our approach to the science of consciousness}

Of course, we are not at all claiming that artificial agents (robots, virtual humans) governed by an implementation of the PCM are \textit{ipso facto} conscious. The question of the specific physical or computational conditions generating conscious experience as such remains open \cite{rudrauf2}, and might be beyond the reach of science \cite{merker2021integrated}. We are more willing to claim that the general principles of the model capture essential properties of consciousness, namely its subjective perspective and how it operates as a global workspace, and that, as a result, the behaviours of artificial agents governed by it could be said to be generated according to principles that govern consciousness. In our view, the scientific endeavour at stake is more about modelling and predicting \textit{how} observable phenomena are governed and relate to each other, than about explaining \textit{why} fundamental phenomena ultimately exist. 
\bigskip

\subsection{Conclusion}

The Projective Consciousness Model (PCM) accounts for key aspects of the phenomenology of consciousness, namely its subjective perspective, and in doing so allowed us to understand and model central aspects of the putative causal role of consciousness as a global workspace in information processing and the generation of behaviour. The approach, which was implemented in this report as a first proof-of-concept, yielded interpretable behaviours based on explicit psychological mechanisms and parameters related to perception, appraisal, imaginary perspective taking and drives, in non-social and social contexts. Much work remains to be accomplished in order to develop and validate the approach, but we believe it is a meaningful first step and useful foundation for future research.  

\bigskip

\section*{Acknowledgments}

We thank Profs. Daniel Bennequin, Kenneth Williford, Karl Friston, Martin Debbané, and Andrea Samson, for their help with previous related projects that contributed to the emergence of this article. We thank master students Vladimira Ivanova, Pauline Jeantet and Isabelle Rambosson, for their work under David Rudrauf's supervision on related topics, which helped support and focus this article. We thank anonymous reviewers for very useful suggestions of revisions.
The project was funded with initial seed funding form the Computer Science University Center at the University of Geneva, and with a grant to David Rudrauf from the Swiss National Science Foundation (205121-188753, Division II: Mathematics, Natural sciences and Engineering).      

\section*{Authors' contributions}
David Rudrauf: conceived and developed the principles of the model and algorithm, implemented a first version of the code, performed the analyses of the simulation results, and co-wrote the article.  
Grégoire Sergeant-Perthuis: developed the mathematical formulation of the model and algorithm, and co-wrote the article.  
Olivier Belli: developed and implemented the code used for the simulations and the interfaces with Cozmo robots, ran the simulations and robotic replay, and co-wrote the article.
Yvain Tisserand: helped with the conceptual and technical rationale, and co-wrote the article.
Giovanna Di Marzo Serugendo: contributed to initiate the robotic project, to co-design an early version of the robotic system, and co-wrote the article.

\begin{appendices}

\section{Projective Geometrical components of the model}
\label{projective-results}

\subsection{Field of Consciousness}

The PCM assumes that 3-dimensional projective geometry plays a central role in consciousness; for an overview of projective geometry in the context of the PCM, see \cite{rudrauf5}. Let us recall some definitions. \\

\begin{defn}[Euclidean frame, chart, change of coordinates]
Let $E$ be a 3-dimensional Euclidean space, any 4 points $(O,A,B,C)$ such that no 2-dimensional Euclidean space contains all of these points is a frame. Equivalently, there is a one to one correspondence between the frame and the data of a reference point $O\in E$ and a basis of the associated vector space $\overset{\rightarrow}{E}$. The chart associated to the frame is the unique affine transformation $T$ from $E$ to $\R^3$ that sends,

\begin{align}
\text{O to }& 0\\
\text{A to } &(1,0,0)\\
\text{B to } &(0,1,0)\\
\text{C to } &(0,0,1)
\end{align}

The change of coordinate between two frames $(O,A,B,C)$, $(O_1,A_1,B_1,C_1)$, or change of charts, is the affine transformation $L$,

\begin{equation}
L= T_1\circ T^{-1}
\end{equation}

\end{defn}
\vspace{0.5cm}

\begin{defn}[Projective space and projective transformations]\label{spatial-stat}
The 3-dimensional projective space, $P_3(\R)$ is the set of lines of $\R^4$, more precisely it is the quotient space of $\R^4$ for the following equivalence relation: for any $v,u\in  \R^4\setminus \{0\}$,

\begin{equation}
u\sim v \iff \exists \lambda\neq 0,  u=\lambda v
\end{equation}

Any bijective linear transformation from $\R^4$ to $\R^4$, i.e. any invertible matrix $4\times 4$, defines a projective transformation.
\end{defn}
\vspace{0.5cm}

\begin{rem}
Let $M: \R^4\to \R^4$ be a projective transformation of $P_3(\R)$, and let,

\begin{equation}
I=\{(\lambda_1,\lambda_2,\lambda_3)\in R^3: \quad M(\lambda_1,\lambda_2,\lambda_3,1)[4]=0\}
\end{equation}

Then the projective transformation can be represented as an application from $\phi:\R^3\setminus I \to \R^3$ defined as, for any $(\lambda_1,\lambda_2,\lambda_3)\in \R^3$

\begin{equation}
\phi(\lambda_1,\lambda_2,\lambda_3)=  \left(\frac{M(\lambda_1,\lambda_2,\lambda_3,1)[1]}{M(\lambda_1,\lambda_2,\lambda_3,1)[4]},\frac{M(\lambda_1,\lambda_2,\lambda_3,1)[2]}{M(\lambda_1,\lambda_2,\lambda_3,1)[4]},\frac{M(\lambda_1,\lambda_2,\lambda_3,1)[3]}{M(\lambda_1,\lambda_2,\lambda_3,1)[4]}\right)
\end{equation}

\bigskip

We shall refer to one or the other representation simply as $M$ or $\phi$.
\end{rem}

Let us recall that a subject is just an agent that we singled out. Any subject is given its own 3-dimensional Euclidean frame, i.e. 4 points of the 3-dimensional Euclidean space, from which it will take a perspective: a center $0$ and three vectors $(u_1,u_2,u_3)$ that correspond respectively to the upward direction (up-vector), to its right, and to the direction in front of it. The positions of objects and other agents are expressed in these coordinates.\\

\begin{ex}
In the context of computer graphics this frame would be related to the frame of the camera. Let $(x,y,z)$ be the center of the camera and let $(u_1,u_2,u_3)$ be the orthonormal basis that corresponds to what is on top (up vector), on its right, and in front. To be more precise if the vector plane of the frame is spanned by two orthogonal vectors $(u_1,u_2)$ of norm $1$ then,

\begin{equation}
u_3=u_1\wedge u_2
\end{equation}

\end{ex}

The change from the external world to the internal world of the subject is done by a 3-dimensional projective transformation. In order to choose a 3-dimensional projective transformation that represents the point of view of the subject and its embodied first person perspective, we need to specify some properties of this projective transformation.\\

Let $M$ be a projective transformation, we shall denote a vector of $\R^4$ as $(x,y,z,k)$; if we want the $x$-axis, the $y$-axis and the $z$-axis to be preserved by $\phi$, $M$ is constrained to be,

\begin{equation}
M=\begin{pmatrix}
a & 0 & 0 & a_1 \\
0 & b& 0& b_1 \\
0 & 0  & c& c_1 \\
\alpha & \beta & \gamma & \rho
\end{pmatrix}
\end{equation}

We assume that the $z$ axis is the axis along which the subject is aiming (looking) at its surrounding space, and that only positive values of $z$ are considered by the subject at a given time instant, i.e. the subject only represents what is in front of it (or could be in front of it in imagination). This is a restriction that is motivated by considerations about constraints imposed by normal adaptive sensorimotor contingencies, which prevents a singular half-plane to appears in the FoC. Such restriction might be released in a variety of atypical, pathological or abnormal states, which are in and of themselves of great interest but beyond the scope of this report (see \cite{rudrauf4} for preliminary considerations). $M$ is thus constrained to be such that,

\begin{equation}
\alpha,\beta=0
\end{equation}

and $\gamma$ and $\rho$ have the same sign (and $\rho>0$).\\

If furthermore the origin is sent to the origin by the projective transformation then,

\begin{equation}
a_1=b_1=c_1=0
\end{equation}

Let $\phi$ be defined as, 

\begin{equation}
\phi(x,y,z)= \left(\frac{M(x,y,z,1)[1]}{\gamma z+\rho}, \frac{M(x,y,z,1)[2]}{\gamma z+\rho},\frac{M(x,y,z,1)[3]}{\gamma z+\rho}\right)
\end{equation}

\begin{equation}
\phi(x,y,z)= \left(\frac{ax}{\gamma z+\rho}, \frac{by}{\gamma z+\rho}, \frac{cz}{\gamma z+\rho}\right)
\end{equation}

Finally, we assume that projective distorsions only appear away from the origin, which is consistent with the intuition of a 3-dimensional space in perspective. In other words, we want the space close to the subject to be represented faithfully. This translates by requiring that the differential of $\phi$ in $0$ (when $\rho \neq 0$) is the identity matrix i.e. for any $v\in \R^3$

\begin{equation}
d\phi[0; v]= v
\end{equation}

where $d\phi[0;v]$ is the differential of $\phi$ at point $0$ applied to the vector $v\in \R^3$.
And then, 

\begin{equation}
M=\begin{pmatrix}
\frac{1}{\rho} & 0 & 0 & 0 \\
0 & \frac{1}{\rho}& 0& 0 \\
0 & 0  & \frac{1}{\rho}& 0 \\
0 & 0 & \gamma & \rho
\end{pmatrix}
\end{equation}

However as for any $\lambda\neq 0$, $\lambda M$ and $M$ induce the same projective transformation we only have to consider the following form,

\begin{equation}
M=\begin{pmatrix}
1 & 0 & 0 & 0 \\
0 & 1& 0& 0 \\
0 & 0  & 1& 0 \\
0 & 0 & c & 1
\end{pmatrix}
\end{equation}

$\frac{1}{c}$ can be assimilated to a depth of field and is positive. \\

\subsection{Steven's law in the context of appraisal} 

We sum up the discussion of the previous paragraph in the following proposition.\\

\begin{prop}[Characterizing the subject's projective charts]\label{subject-projective-frame}
The projective transformations are such that,

\begin{enumerate}

\item the subject is centered in $0$ after projective transformation, i.e. in its perspective it is at the center of its frame.

\item  the three axis $x,y,z$ are preserved, i.e. the axis of the Euclidean frame associated with the agent (up-down, left-right, and back-front) must be preserved after projective transformation.

\item  no points in the ambient space appears, to the subject, to be truly at infinity; this constraint is realisable because the subject can only directly represent what is in front of it.

\item objects that are near the agent appear to have the same size as in the reference Euclidean frame. 
\end{enumerate}

Thus the retained projective transformations are the transformations, $M$, such that

\begin{equation}
M=\begin{pmatrix}
1 & 0 & 0 & 0 \\
0 & 1& 0& 0 \\
0 & 0  & 1& 0 \\
0 & 0 & \gamma & 1
\end{pmatrix}
\end{equation}

with $\gamma\in \R_+$ a positive real number.
\end{prop}
\vspace{0.5cm}

We shall call the frames of the previous Proposition (\ref{subject-projective-frame}) as the subject's projective chart. Let us justify this denomination. Let $E$ be the ambient Euclidean 3-dimensional space. In the Euclidean frame of the subject any point of $E$ has a unique coordinate $(\lambda_1,\lambda_2,\lambda_3)$, then when applying the previous projective transformations, $\phi$, one gets the representation of the half space in front of the subject in its FoC. If one has a reference Euclidean space $(O,e_1,e_2,e_3)$, common to all agents, in which the coordinates of each objects can be expressed, in order to express the projective transformation proper to each agent, $\psi_a$, one needs to change coordinates, from $(O,e_1,e_2, e_3)$ to $(O_a, u_{1,a}, u_{2,a}, u_{3,a})$, and then apply the projective transformation $\phi$. In other words, if one notes $L_a$ the Euclidean change of charts from $(O,e_1,e_2,e_3)$ to the Euclidean frame $(O_a, u_{1,a}, u_{2,a}, u_{3,a})$, the projective transformation proper to $a$ is,

\begin{equation}
\psi_a=\phi \circ L_a
\end{equation}

The changes of coordinates $L_a$ we consider are isometries, in other words, they preserves distances. In what follows we will characterise how lengths and volumes of objects changes when they are moved away in the direction of aiming (sight), the z-axis, of the subject. As the distances are preserved by $L_a$, it is not different from considering that the reference frame is the Euclidean frame of the agent.  In this setting, Stevens' Law with exponent $-1$, as found in (\cite{teghtsoonian1}) in the context of appraisal, appears to be a consequence of the representation of the external data in a projective chart that satisfies natural constraints (Proposition \ref{subject-projective-frame}).\\

\begin{thm}[Stevens' Law]\label{Stevens-law}
For $M$ a projective chart of a subject, the ratio between the perceived size of an (infinitesimal) object, $r_p$ and its real size, $r$ varies asymptotically as,

\begin{equation}
\frac{r_p}{r}\sim \frac{1}{cz}
\end{equation}

where $z$ is the coordinate of the direction along which the subject is aiming, i.e. the third coordinate in the rigid bases associated with the subject $(u_1,u_2,u_3)$. If the object is aligned with the orthogonal plane of the direction of sight then,

\begin{equation}
\frac{r_p}{r}= \frac{1}{cz+1}
\end{equation}

Furthermore the ratio of the volumes varies as,

\begin{equation}
\frac{v_p}{v}\sim \frac{1}{(cz)^4}
\end{equation}

To be more precise, for any $(\lambda_1,\lambda_2,\lambda_3)\in \R_{+}^3\setminus 0\times 0\times \R_{+}$, there is a constant $C$ such that,

\begin{equation}
\Vert d\phi[x,y,z;\lambda]\Vert \sim_{z\to +\infty} \frac{C}{z}
\end{equation}

and,

\begin{equation}
\det d\phi[x,y,z]\sim_{z\to +\infty} \frac{1}{(cz)^4}
\end{equation}
\end{thm}

\begin{proof}
 Let $c\geq 0$ be a positive real number and $\phi$ be the associated projective chart. For any $(x,y,z)\in \R^3$ and $\lambda=(\lambda_1,\lambda_2,\lambda_3)$,
 
 \begin{equation}
 \Vert d\phi[x,y,z; \lambda]\Vert^2= \frac{\lambda_1^2}{(cz+1)^2} + \frac{\lambda_2^2}{(cz+1)^2}+ \frac{\lambda_3^2(1+2c)}{(cz+1)^4}- \frac{2c\lambda_1\lambda_3}{(cz+1)^3}-\frac{2c\lambda_2\lambda_3}{(cz+1)^3}
 \end{equation}
 
 Therefore when $\lambda_1\neq 0$ or $\lambda_2\neq 0$, in other words, when $\lambda$ is not in the direction $u_3$ (the object is not aligned with the sight direction), then,
 
\begin{equation}
\Vert d\phi[x,y,z;\lambda]\Vert \sim_{z\to +\infty} \frac{\sqrt{\lambda_1^2+\lambda_2^2}}{cz}
\end{equation}

Furthermore,

\begin{equation}
d\phi[x,y,z;\lambda]= \frac{1}{(cz+1)^4}
\end{equation}

Therefore,

\begin{equation}
\det d\phi[x,y,z]\sim_{z\to +\infty} \frac{1}{(cz)^4}
\end{equation}

\end{proof}

\section{Model technical details}\label{model-tech-details}

As in Friston's Free Energy principle, the PCM-mediated process of active inference we focus on includes two phases: a ``perception phase", we shall refer to as an inverse inference problem, and an ``action phase", we shall refer to as a direct inference problem. The approach somehow differs from Friston's Free Energy principle, as the system is not only trying to enforce homeostasis but also to increase its utility (in the sense of a rewarding outcome) based on internal prior expectations. This intrinsic distinction changes the perspective on the action phase. \\

We explain how a given agent represents the entities of the external world in its internal 3-dimensional space, and what cost function is minimised for action selection.

Let us define the spaces that come into play in our approach to active inference.
An entity $e\in E$ is an agent $a\in A$ or an object $o\in O$. Agents can perform active inference, objects cannot. Each agent takes as sensory inputs: 1) the valence, i.e. the difference of the positive and negative emotions expressed by other agents, represented by a number in $[-1,1]$; 2) the space they occupy and that each object occupies in the external 3-d space, which is a subset of $\R^3$ (the collection of subsets of $\R^3$ is denoted as $\mathcal{P}(\R^3)$). In other words, for a given agent $s\in A$, what plays the role of the space of sensory inputs for active inference as described in Section \ref{general_settings} is,

\begin{equation}
S=[-1,1]^A\times \mathcal{P}(\R^3)^A\times  \mathcal{P}(\R^3)^O
\end{equation}

To simplify the implementation, the contexts are not inferred trough Bayesian inference (see Subsection \ref{pcm:inverse-inference}) but in a deterministic manner. 
Here the only truly internal parameter is a preference matrix $p_s$ (see below), and the space of contexts is,

\begin{equation}
\Gamma=[0,1]^{A\times E}
\end{equation}

The space of action is a finite set of basic actions that an agent can perform: moving along a given direction, orienting, and emitting signals expressing emotions. The selection of a move is given by the optimisation of a cost function that relates to a free energy (Equation \ref{appendix:PCM:total-free-energy}), however it is important to emphasize that it differs from the way the optimisation problem is handled in the action phase according to Friston's Free Energy principle \textit{stricto sensu}. 
We will use the term subject, in order to distinguish an agent under consideration from other agents.\\

The chart that corresponds to the subject's point of view of the real world is its projective chart (see Appendix \ref{projective-results}) and is denoted depending on whether it is seen in $\R^3$ or $\R^4$ as respectively $\phi$ or $M$. Following \cite{rudrauf5}, its is a subjective, internal chart that the subject uses to organise its internal model of the world, based on sensory data and prior beliefs, also called the Field of Consciousness (see Section \ref{FoCprojective}). Here that frame can be conceived as the actual or simulated 3-dimensional FoC of the subject.

\subsection{Preferences and emotion}\label{method:pref-tensor}


Emotions expressed by each agent are sensory inputs for inference of preferences. A subject embeds internal models of preferences or affective expectations, regarding objects and agents, $p_s\in [0,1]^{A\times E}$. We call the collection of the preference matrices, the preference tensor and it is denoted as $p$.




 The preferences of a subject $a\in A$ for an entity $e\in E$ is given by $p_{aae}$, and its beliefs about the preferences of another agent $b\in A$ for an entity $e\in E$ is given by $p_{abe}$. In the following simulations, we will assume that each agent tends to have a good image of itself (although it may believe that others do not), so that set $p_{abb}=1$. Note that our approach incorporates not only the preferences attributed by a subject to other agents, but also the preferences it attributes to other agents about other agents.\\

The preferences of a subject $s\in A$ is given by its preference vector, $q_{s,e}= p_{s,s,e}$ for any $e\in E$; the collection of the true preferences of the agents is a submatrix of $p$ defined as  $q_{a,e}=p_{a,a,e}$ for $a\in A$ and $e\in E$. 

Let us take an example to illustrate the previous point. 

\begin{ex}
Let Sally (S) and Anne (A) be the two agents we consider and let the marble (M) be the only object in the scene. In Sally's subjective perspective, her preferences are as follows, she really likes the marble at a level of $0.9$, and is indifferent with respect to Anne at a level of $0.5$ (here considered as a neutral level). She also believes that Anne is indifferent with respect to her at the level $0.5$, and that Anne slightly care about the marble at the level $0.6$. Sally's preferences are therefore,
\begin{center}
\begin{tabular}{l ||l | c | r | }
 & S &A& M    \\
S  &    1& 0.9 & 0.5 \\
A &    0.5 & 1 & 0.6 \\
\end{tabular}
\end{center}

In this example the preferences of Sally herself (preference, $q$) are given by,

\begin{center}
\begin{tabular}{l ||l | c | r | }
 & S& A & M    \\
S &    1& 0.9 & 0.5
\end{tabular}
\end{center}
\end{ex}

\vspace{0.5cm}

\subsection{ToM-related factors: other agents' influence on the subject}\label{tom-factors}

In the model, ToM relies on the subject taking perspectives attributed to other agents to appraise information from their standpoint based on beliefs and preferences attributed to them. This mechanism is explained in subsection \ref{def-perceived-value-uncertainty}. Here we describes two factors controlling the influence of other agents on the subject that contributes to the process. These factors translate in the model into two tensors: the preference influence tensor and the action influence tensor. The use of tensors (\textit{versus} simple matrices) is necessary to encode the information that a subject needs to perform $level-6$ (or second order) ToM, i.e. for the subject to take into account the beliefs it attributes to another agent about another agent (which can be the subject itself).
The preference influence tensor is denoted as $I^p$, and the action influence tensor as $I^m$. Both are a collection $(I_{sab}\in [0,1],(s,a,b)\in A^3)$ such that for any $s,a\in A$ $\underset{b}{\sum}I_{sab}=1$. When agent, $a$, can only do first order ToM the matrice $I_{a..}$ are replaced by vectors $(J_{ab},b\in A)$ and $\sum_b J_{ab}=1$.





\subsection{Attentional focalisation}\label{FoCprojective}

In order to further model the modulation of the salience of information through spatial attention in the FoC, we defined an attentioal focalisation function, which can be thought of as an ``attentional beam" of sorts \cite{baars, baars1}. Attentional focalisation within the FoC can be conceived as a weight on its projective coordinates, highlighting so-to-speak a subset of the FoC. Here we assume a simple scheme, in which information is weighted down as a function of its eccentricity from the center of the 3-dimensional projection (the subject pay more attention to what appears in front of it). We use a Gaussian function to estimate a weight that is equal to $1$ at the center of the projection in the $(x,y)$ plane, and decreases with eccentricity, with a rate that depends on a parameter $\sigma$. More precisely, if an object is centered on $(x,y,z)\in \R_+\times \R^+\times [0,\frac{1}{c}[$ in the projective chart of the subject with parameter $c$, then the attention focalisation weight is,

\begin{equation}
\gamma(x,y)=e^{-\frac{x^2+y^2}{\sigma^2}}
\end{equation}

\subsection{Inverse inference on preferences}\label{pcm:inverse-inference}

We shall now detail how the internal parameters, which for us corresponds to the different contexts ($\Gamma$), are computed. \\

The preference tensor is inferred by a subject from the expressed emotions of other agents as follows.
A subject, $s$, will use the expressed emotions (here the valence) of another agent $a\in A$ to update the preferences it attributes to the other agent in its preference tensor $(p^{t+1}_{s a e}, a\in A\setminus \{s\}, e\in E)$ at the next time step. Updated preferences $(p^{t+1}_{s a e}, a\in A\setminus \{s\}, e\in E)$ are,

\begin{equation}
p^{t+1}_{sae} = p^t_{sae}+g(v^t_{a},p^t_{sae})c_{sae}
\end{equation}

where $c_{sae}$ is the combined certainty of $s$ towards $a$ towards $e$:

\begin{equation}
c_{sae}=(1-\sigma_{sa})(1-\sigma_{ae})
\end{equation}

Where $1-\sigma_{ab}$ is the certainty and $\sigma_{ab}$ is the uncertainty. We present more details about uncertainty in subsection \ref{subsec:uncertainty}.\\

$g$ is a function that computes the gain to apply to update the current preference $q\in [0,1]$ as,
\begin{equation}
g(v,q) = -q_n + (q_n - q)\alpha + \frac{1}{1+e^{-\lambda f(v)}}
\end{equation}

where $\alpha$ and $\lambda$ are parameters (i.e. $\alpha = 1$ and $\lambda = 25$) and $f$ is used to modulate the sensory gain of the perceived valence. $f(x) = x$ can be used to avoid modulating the valence, but using $f(x) = x^3$ is a sound choice as it will enhance the update strength of the valence for extremities ($-1$ and $1$) and reduce it around neutral expressions ($0$) (seeing another agent express a positive emotion with small intensity will only have a very small weight on the update of preferences, whereas a strong intensity will have a quite strong update weight). $q_n$ is a reference preference, defining a neutral level, which we assume to be the same for everybody, above which entities appear as positive and below which they appear as negative. \\

The subject updates its own preferences ($q_{s,e}$) based on the influence others have on it (see subsection \ref{tom-factors}), 
\begin{equation}
    q^{t+1}_{s,e} = I^p_{sss}q^t_{s,e}+\underset{\substack{a \in A\\a \neq s}}{\sum}I_{ssa}^p p^{t}_{s,a,e}
\end{equation}

\subsection{Action of the Field of Consciousness on appraisal}\label{def-perceived-value-uncertainty}

The way a subject selects its actions with respect to prior preferences depends on its appraisal of the external world. Note that this also entails a variable uncertainty regarding sensory evidence that translates into uncertainty in its appraisal. In the first paragraph of Subsection \ref{action-decision}, we will explain how this uncertainty is integrated, in a somewhat natural manner, in the decision making process. But first, let us explain how the appraisal works. \\

\subsubsection{Perceived value}\label{spatial-statistic} 

 We explained how the perspectival properties of the FoC induce an effect of shrinking on the apparent size of entities that are far from the subject point of view, and can be hypothesized to play a direct role in the affective appraisal of those entities, and as a result act as a motivational factor. Let us recall that from Theorem \ref{Stevens-law}, the ratio of the apparent volume of an entity in the FoC, $v_p$, and its real volume, $v$, varies as the inverse of the quartic power of depth from the subject, i.e.

\begin{equation}
\frac{v_p}{v}\sim \frac{1}{(cz)^4}
\end{equation}

To recover documented effects of inverse distance on appraisal from intrinsic properties of projective geometry, we therefore quantify the quadric root of the estimated volume of the entity as a measure of apparent linear size. For each entity considered by an agent, we weight the preference for the entity with that quantity. To simplify modelling, we consider the space that is not occupied by entities as homogeneous with respect to preference. It is assigned a neutral preference that we shall note $q_n$, which is assumed to have the same value for all agents. We then weight the preference $p=q_{s,e}$ for a given entity $e\in E\setminus {s}$ by the volume it occupies in the FoC; this new quantity $\mu_s(e)$ is the perceived value of the entity,

\begin{equation}
\mu= p\gamma\frac{v_p^{1/4}}{v_{tot}^{1/4}}+ q_n(1-\gamma\frac{v_p^{1/4}}{v_{tot}^{1/4}})
\end{equation}

$q_n$ is the neutral preference, $\gamma\in [0,1]$ is the attentional focalisation factor that takes into account `eccentricity', in other words,
\begin{equation}
\gamma= \gamma(x,y,z)
\end{equation}

with $(x,y,z)$ the point at which the entity is centered and $v_{tot}$ a reference volume, e.g. the total volume encompassed by the FoC (taking into account the clipping of the plane $(x,y)$). \\

For a given subject $s$, $\mu_s(e)$ is a function of the configuration of $e$ in the external space, i.e. $X_e$, and depends on the frame, $\psi_s$, of $s$ and on its preference vector $q_s$; this can be noted as $\mu_{s,\psi_s,q_s}(X_e)$. \\


\subsection{Action 1 : Emotion expression}\label{emo-express}

In this implementation of the model, we use a rather simplistic model of emotion expression, which is a direct proxy of the preference of the subject. Thus we do not assume that the expression of emotions by the subject is itself part of its optimisation process (a subject will always express an emotion that is consistent with its state of appraisal). A subject uses the perceived values, $(\mu_s(e),e\in E)$, as a base to express emotions. The emotions are expressed along two dimensions: positive and negative emotions. The expressed valence is simply the difference between the two quantities. The relative perceived preferences $\mu'_s$ are defined as,
\begin{equation}
\mu'_s(e)=\frac{\mu_s(e)-q_{n}}{1-q_{n}}
\end{equation}
and used to express emotions. Here again the reference neutral preference $q_n$ appears; a refinement would be to consider that each subject has its own neutral level of preference, $q_{s,n}$ (an agent may consider that a certain level of well-being is overall neutral with respect to its own expectations, when another agent might consider that it is already a level of well-being that is substantially above baseline). The intensity of positive emotion expression is given by the following expression,
\begin{equation}\label{positive- emotion}
      e_s^{+}   =\frac{\sum_{\nu\in M_s^{+}} (\nu^n )^\frac{1}{n}}{|M_s^{+}|^{1/n}}
\end{equation}
where $M_s^{+}$ are the values of $\mu_s$ above $q_n$, 
\begin{equation}
M_s^{+} = \{ \mu'_s(e) \geq 0 \vert e\in E\}
\end{equation}
and $n$ is a constant (the choice of which is explained below).\\

Likewise, the intensity of the negative emotion expression is,

\begin{equation}\label{negative-emotion}
e_s^{-}   =\frac{\sum_{\nu\in M_s^{-}} (\nu^n )^\frac{1}{n}}{\vert M_s^{-}\vert^{1/n}}
\end{equation}
where $M_s^{'-}$ are the values of $\mu_s$ below $q_n$,

\begin{equation}
M_s^{-} = \{ \mu^{'}_s(e)\leq 0 \vert e\in E\}
\end{equation}

The expressed valence is the difference between $e_s^{+}$ and $e_s^{-}$ as,

\begin{equation}
v_s = e_s^{+} - e_s^{-}
\end{equation}

Let us now explain how one can set the value of $n$ and more generally why we chose the previous formula for positive and negative emotion expression (Equation \ref{positive- emotion}, \ref{negative-emotion}). Let us consider a situation in which the subject faces an object that it really likes but surrounded by neutral objects (for which $\mu(e)=0.5$, i.e. the subject is indifferent to these objects). It seems reasonable to assume that, in this situation, the subject should be fairly happy. One option would be to consider $n=1$ in Equation \ref{positive- emotion}, which would mean that the objects which the subject is indifferent to are as important as the one it really enjoys in its appreciation of the situation, which does not seem plausible. Another option could be to just take the maximum of $\mu(e)$. However, if the scene is crowded with objects that the subject is rather indifferent to, this would entail, which we may not want to assume, that the level of positive emotion would be exactly the same as when the liked object is on its own. The first option corresponds to considering the $l^1$ norm on $\R^{M_s^{+}}$, and the second option to considering the $l^\infty$; fortunately there are many norms between these two norms, which are called the $l^p$ norms, here $n$ in Equation \ref{positive- emotion} corresponds to these $p$.\\

We choose $n$ based on the following scheme. Let us consider that there is a fraction $\alpha\in [0,1]$ of the objects to which an agent is indifferent and that the agent really likes all the other objects, i.e. that for these objects the values, $\mu_s^{'}$, are respectively $0$ and $1$. In such case, we still expect the positive emotion to be high for example at a reference level $d\in [0,1]$; we can write

\begin{equation}
 e_s^{+}  = \alpha^{\frac{1}{n}}+0=d
\end{equation}

then, 

\begin{equation}
n =\frac{\ln \alpha}{\ln d}
\end{equation}

We chose $\alpha= 1/2$ and $d=0.9$ which gives us $n=7$.

\subsection{Action 2: Choice of best move} \label{action-decision}

The method of optimisation for the choice of the action of moving based on free energy (Equation (\ref{appendix:PCM:total-free-energy})) entails a certain complexity, in particular because of the ToM component of the model, which requires a graded, stepwise presentation to explain its different components, from the simplest cases to the more complex ones.\\

\subsubsection{$\mu$ and $\sigma$ parameterise probability distributions} 

The aim of the subject is to maximize positive outcomes, which correspond to high values of perceived value and low uncertainty about it (see Section \ref{subsec:uncertainty} for the definition of uncertainty). In order to achieve this aim, the subject needs to quantify the extent to which its state departs from the ideal configuration of high perceived value and low uncertainty, $(\mu_0,\sigma_0)$. A natural way to do so is to embed these parameters in the probability simplex $\mathbb{P}([0,1])$ of probability laws over the perceived value, and to quantify the ``distance" of $(\mu,\sigma)$ from $(\mu_0,\sigma_0)$ by using the relative entropy corresponding to the Kullback-Leibler divergence.\\

For a given subject, $s$, and a given entity, $e$, different from the subject, $\mu_s(e)$ and $\sigma_s(e)$ parameterise a probability law on $[0,1]$ centered on $\mu_s(e)$ and of dispersion $\sigma_s(e)$, which we shall note as $Q(\lambda|\mu_s(e), \sigma_s(e))$ with $\lambda\in [0,1]$.\\

One can for instance consider $Q(.|\mu_s(e),\sigma_s(e))$ to be a truncated Gaussian of mean $\mu_e$ and variance $\sigma_e$. The prior is a law of high certainty centered on a high weighted preference, i.e on $\lambda$ close to $1$; we shall note it as $P$. Here for simplicity sake, it is assumed to be the same for every agent.\\

We shall now explain how a subject selects its next action. The quantity we want the subject to minimise in the end (Equation (\ref{appendix:PCM:total-free-energy})) integrates multiple terms of free energy.\\

\subsubsection{Simplest case: subject without Theory of Mind}

In the simplest case, let us assume that the subject cannot take into consideration the perspective of other agents on the situation nor predict their reactions and therefore that it cannot establish a strategy by anticipating the others' actions.\\

Let us also start by assuming that the subject is confronted with only one other entity of configuration $X_e$ in the ambient space $E$. The subject can choose different actions. Here the repertory of actions is reduced to moving in a given direction, orienting in a given direction, or not moving. Let us note the set of possible moves as $M$. The subject will attempt to minimise the Kullback–Leibler divergence between $Q(.|\mu_s(X_e), \sigma_s(X_e))$, denoted simply $Q(.|\mu_s(e), \sigma_s(e))$, and a given probability law $P$ that represents an ideal state in which the perceived preference of the agent is high and the uncertainty low, i.e.,

\begin{equation}\label{first-equation-action}
m^*=\underset{m\in M}{\min}\DKL(Q\left(.|\mu_{s,\psi_s(m),q_s}(e), \sigma_{s,\psi_s(m)}(e)\right)\Vert P)
\end{equation}

\begin{rem}\label{appendix-rem-fe-dkl}
We will make the abuse of calling free energy the Kullback-Leibler divergences or sums of these divergences; the reason why we do this is to insist on the link between the free energy and this divergence that we will explain now. The free energy of a probability density $Q\in \mathbb{P}(\Omega)$ with respect to a random variable $H:\Omega\to \R$, with $\Omega$ a finite space is,

\begin{equation}
F(Q)= \E_Q[H] - S(Q)
\end{equation}

Let, $Z=\sum_{\omega\in \Omega} e^{-H(\omega)}$ and, 

\begin{equation}
P= e^{-H}/Z
\end{equation}

On can rewrite the free energy as follows,
\begin{equation}
F(Q)=  \DKL(Q\Vert P)-\ln Z
\end{equation}

Minimizing the free energy with respect to $Q$ is the same than minimizing the Kullback-Leibler divergence. In the literature of active inference one most often encounters the notion of free energy this is why we kept this denomination.
\end{rem}

It is important to emphasize that each move changes the subject's projective chart, i.e. FoC, which directly influences the perceived value and uncertainty on it, $\mu$ and $\sigma$ (see Section \ref{def-perceived-value-uncertainty}). This dependency is indicated by the index $\psi(m)$ in Equation (\ref{first-equation-action}). What is implicit in this expression is that the map from moves to projective transformations, $\psi_s(m)$, depends on the Euclidean frame $\pi_s$ of the subject.\\

When there are more than one other entity than the subject, the minimisation is performed on a weighted sum, with weight $(\omega_s(e),e\neq s)$, of the Kullback-Leibler divergences. Let us denote $X$ the vector of configuration of the entities,

\begin{equation}\label{cost_without_strategy}
L_s(m_s,q_s,X)=\underset{\substack{e\in E\\e\neq s}}{\sum} \omega_s(e) \DKL(Q(.|\mu_{s,\psi_s(m),q_s}(e), \sigma_{s,\psi_s(m)}(e))\Vert P)
\end{equation}

\subsubsection{Intermediary case: Theory of Mind without prediction of the actions of others}

Let us consider a case in which a subject can perform social-affective appraisal of others perspectives, i.e. a basic component of ToM according to the model, but without predicting the actions of other agents based on it. In other words, the subject may understand and be influenced to some extent by the affective states it attributes to others but will not be able to develop competitive strategies based on the computation of others' expected moves. 

Note that the extent to which the subject may tend to let itself be influenced by others could be a contextual state or a (personality) trait. In the following simulations, we consider such influence as a trait. Also note that the subject can also entertain beliefs about how other agents let themselves be influenced by other agents. The parameter that accounts for how much a subject takes the opinion of other agents into account in its own actions is given by the action influence tensor.\\

Thus if the subject takes into consideration the opinion of the other agents in making its move but still cannot predict their moves,

\begin{equation}\label{cost-function}
C_s(m_s,p_s,X)=\underset{a\in A}{\sum}\underset{\substack{e\in E\\e\neq a}}{\sum}  \omega_{s,a,e}\DKL(Q(.|\mu_{a,\psi_a(m),p_{s,a,.}}(e), \sigma_{a,\psi_a(m)}(e))\Vert P)
\end{equation}

where for any agent $a\in A$ and entity $e\in E$,

\begin{equation}
\omega_{s,a,e}=J^m_{s,a}\alpha_{a,e}
\end{equation}

and $\alpha_{a,e}=\frac{1}{\vert E \vert-1}$; the best move if the one that minimises $C$. This cost function is chosen when an agent can only do first order theory of mind. For each subject, minimising the energies gives a unique possible move:




\begin{rem}
Let us remark that in Equation \ref{cost-function}, a given subject uses its own beliefs on what the preferences of the other agents are, i.e. it uses $p$ and not $q$.
\end{rem}


\subsubsection{Remark on update rules}

In order for a subject to anticipate possible future outcomes, it needs to consider how preferences may be updated across several time steps but also to keep track of its own actions and those of others (moves and emotions). It assumes that the other agents do not predict others' actions. In the previous section we implicitly assumed that the spatial configuration of all the agents and objects were given, i.e. that at time $0$ $(X^0_e,e\in E)$ is given. As we want to describe what happens during several time step we will be very precise on what changes with time and what doesn't. \\

The configuration in the real space of each agent $(X_a^t, a\in A)$, the preference tensor $p^t$ and the move of each agent $(m^t_a,a\in A)$ vary with time; the influence tensor doesn't change with time. For simplicity it is assumed that the objects are kept fix, in other word agents cannot act on them. The evolution of the variables for an agent $s$ with first order ToM  are prescribed by the following rules,

\begin{align}
p^{t+1}_{s..}= &f(p^t_{s..}, X^{t+1},J^{p}_{s.})\label{strategy-update-1}\\
m^{t+1}_s= &\underset{m\in M_s}{\argmin} C_s(m, p^t_{s..}, X^t, J^{a}_{s.}) \label{strategy-update-3}\\
X^{t+1}_s=& f_1(X^t_s, m^t_s)\label{strategy-update-2}
\end{align}

\bigskip

\subsubsection{More sophisticated case: Theory of Mind with prediction of the actions of others}

When a subject $s\in A$ establishes a strategy, in other words that in our model it can do second order ToM, it iterates $n$ times the update rules in Equations (\ref{strategy-update-1}), (\ref{strategy-update-3}),and (\ref{strategy-update-2}), for all other agents but itself, taking as input the sequence of moves it makes, and it has an idea of what $(X_a^k, p_{a,..}^k, m^k_a)$ is for $k\in [n]$ and any agent $a$. However the subject does not have access to the true $p_{a,..}^0$ nor does it have access to $J_{a.}^p, J_{a.}^m$, instead it initiates its predictions with for each agent $a\neq s$ $\tilde{p}_a, \tilde{I}^p_{aa.}, \tilde{I}^m_{aa.}$ chosen as,

\begin{align}
\tilde{p}^0_{a..}= p^0_{s..}\\
\tilde{I}_{aa.}^p= I^p_{sa.}\\
\tilde{I}^m_{aa.}=I^m_{sa.} 
\end{align}

The number of steps the subject can predict, $n$, is what we call the depth of processing and we denote it as $dp$.  Let us denote the integer interval $[1,n]\subseteq \N$ as $[n]$. From this prediction, it will decide what sequence of moves $(m^{k*}_s, k\in [n])$ is the best, and the first move of this sequence will be its true move (see Equation (\ref{appendix:strategy-best-move})). Let us insist on the fact that the evolution of all the variables $X_a^k, p^k_a, m^k_a$ for $a\neq s$ and $k\in [n]$ is prescribed by Equations (\ref{strategy-update-1}), (\ref{strategy-update-3}), and (\ref{strategy-update-2}), but the subject's moves are kept free, i.e. not constrained by Equation (\ref{strategy-update-3}), and will be the variables on which the optimisation will be carried out. These collection of moves, or paths $m_s=(m^k_s,k\in [n])$, are selected from a set of paths $\Path$. For any $m\in \Path$,

\begin{equation}\label{appendix:PCM:total-free-energy}\tag{FE}
FE(m)=\underset{k\in [n]}{\sum}a_{k} C_s(m^k,p^k_s,X^k)
\end{equation}

where $\sum_{k=1...n} a_k=1$ and $a_k$ are chosen here to be $a_k=\frac{1}{n}$.\\

The best next move is then the first step move of the best path of $n$-moves. Let,

\begin{equation}
m^*=\underset{m\in \Path}{\argmin}{} FE(m)
\end{equation}

then the best next move is,

\begin{equation}\label{appendix:strategy-best-move}
b^*=m^*[1]
\end{equation}

Different sets of paths ($\Path$) can be considered, for example one can consider at each step the $m\in \N^*$ best moves.\\

\section{Behavioural metrics and data analysis}\label{section:behav-met-analysis}

\subsection{Approach-Avoidance metrics} 
Approach-Avoidance metrics corresponded to an overt behavioural estimate of approach-avoidance behaviours from an agent $i$ towards an entity $j$. We used two metrics. The first one, a position or "motion" metrics, corresponded to the Euclidean distance of the agent from the entity, with higher values when the agent was far, and lower values when it was close to the entity. The second one, a look-at metrics, quantified approach-avoidance in terms of the overt attention or orientation of the agent with respect to the entity. It was defined as the dot product $a_ij=dot(u,w)$ between the unit vector $u$ pointing from the agent location towards the entity, and the unit vector $w$ representing the current, real or imaginary orientation of the agent. It is equal to $1$ if the two vectors point in the same direction (the agent look at the entity), $0$ if they are orthogonal, and $-1$ if they point in opposite directions. 

\subsection{Joint Attention metrics}
The joint Attention metrics corresponded to an overt behavioural estimate of joint attention behaviours, between a pair of agents $(i,j)$ towards an entity $k$. It was defined as $ja_ijk=a_ik.a_jk$, the product of the look-at metrics of the two agents with respect to the entity. It is equal to $1$ when the two agents look simultaneously at the same entity, or in the opposite direction from that entity, tends towards $0$ when the orientation of at least one agent is orthogonal to the direction of the entity for this agent, and it is $-1$ when one agent look at the entity and the other look in the opposite direction.

\subsection{Data analysis}

These metrics were computed over time for each simulation, across conditions and across trials. This allowed us to compute mean and standard errors over the metrics as a function of time and condition. We were also interested in the end state of the agents (at the end of the simulations between iterations $60$ and $70$), as an index of the efficacy of the agent's global optimisation process, in order to compare the performance of agents across conditions.   

Statistical tests (t-tests based on the General Linear Model) were performed within the end state period, after averaging results within the corresponding time window for each condition, in order to compare conditions in terms of final outcomes (even though relevant differences between conditions turned out to be quite strong, and it was obvious that they would be significant).

\section{Algorithm description with pseudo-code}
\label{appendix:algorithm}

\subsection{Settings}

The algorithm was implemented in C\# with the future aim of interfacing the code with the game engine Unity3D. The simulations used virtual representations of the robots and cube objects coming with the robots, as triangular tessellations of their envelop, which were simplified to a parallelepipedic bounding box. 

\subsection{Algorithm overview}\label{section:algorithm-pc}

\subsubsection{Algorithm base}\label{section:pcm-alg-base}

The algorithm focused on agents and their states. Its goal was to allow an agent at a time $t$ in a state $s_t$ to evaluate which action $a \in A(s_t)$ (where $A(s_t)$ are all the actions available for the agent in the state $s_t$), would place the agent in the best possible state $s^{'*}$. The best possible state $s^{'*}$ was defined as the state $s' \in S'$ for which the associated free energy $fe(s')$ is minimum. Algorithm \ref{alg:alg_1} summarizes that process. Note: in case of multiple results with identical levels of free energy, the $argmin^*$ function (algorithm~\ref{alg:alg_1}, line~\ref{cline:argmin}) picked a solution randomly.\\

\begin{algorithm}[H]
\label{alg:alg_1}
\caption{Core}
  \SetKwProg{NextState}{NextState}{}{}
  \SetKwInOut{Input}{inputs}
  \SetKwInOut{Output}{output}
\NextState{$(s,A)$}{
\Input{Agent's current state $s$; available actions $A$ for agent in state $s$}
\Output{Agent's next best state $s^{'*}$}
 $S' \gets A(s)$\\
 $FE(S') \gets ComputeFreeEnergy(S')$\\
 $best \gets argmin^*(FE(S'))$\label{cline:argmin}\\
 $s^{'*} \gets S'(best)$\\
 \KwRet{$s^{'*}$}
 }
\end{algorithm}

\subsubsection{States}

A state was a representation of the world as perceived by an agent. Basically, a state consisted of beliefs about:
\begin{itemize}
    \item[a.] entities (i.e. agents and objects) positions in space;
    \item[b.] preferences associated to those agents and objects;
    \item[c.] other agents preferences and beliefs;
    \item[d.] influence of agents on each other;
\end{itemize}

\subsubsection{Actions}

An action was defined as a function that takes a state $s$ as an input and produces a new state $s'$. For example, an action such as "move forward" created a $s'$ where the position of the agent was ahead of its previous position in $s$.

\subsubsection{Free energy computation process}

The main part of this algorithm was to compute the free energy associated with a given state $s$. Algorithm \ref{alg:alg_2} presents an overview of that process.\\

\begin{algorithm}[H]
\label{alg:alg_2}
\caption{Free energy computation process}
  \SetKwProg{ComputeFreeEnergy}{ComputeFreeEnergy}{}{}
  \SetKwInOut{Input}{inputs}
  \SetKwInOut{Output}{output}
\ComputeFreeEnergy{$(s)$}{
\Input{Agent state $s$}
\Output{Free energy $fe$ associated with $s$}
 $s' \gets ComputeUncertainty(s)$\label{cline:uncertainty}\\
 $s^{''} \gets UpdatePreferences(s')$\label{cline:prefs}\\
 $s^{'''} \gets GetPerceivedValue(s^{''})$\label{cline:spatial}\\
 $fe \gets fe(s^{'''})$\label{cline:fe}\\
 \KwRet{$fe$}
 }
\end{algorithm}

Lines \ref{cline:uncertainty} to \ref{cline:fe} are formally described in section~\ref{mat-and-meth}
\begin{itemize}
    \item[a.] $ComputeUncertainty$ (line~\ref{cline:uncertainty}) refers to \ref{subsec:uncertainty};
    \item[b.] $UpdatePreferences$ (line~\ref{cline:prefs}) refers to \ref{pcm:inverse-inference};
    \item[c.] $GetPerceivedValue$ (line~\ref{cline:spatial}) refers to \ref{spatial-stat};
\end{itemize}
Free energy was computed following Equation \ref{f-e} using the uncertainty $\sigma$ ($\sigma \gets ComputeUncertainty$) and $mu$ ($\mu \gets GetPerceivedValue$).

\subsubsection{Simulation algorithm skeleton}

Simulations were done using algorithm \ref{alg:alg_1} as their core. Algorithm \ref{alg:alg_3} describes a general skeleton for a simulation.\\

\begin{algorithm}[H]
\label{alg:alg_3}
\caption{Simulation skeleton}
  \SetKwProg{Run}{Run}{}{}
  \SetKwInOut{Params}{parameters}
\Run{}{
\Params{Initial entities positions in the world $W$; Initial agents states $S$;}
$S' \gets S$\\
$W' \gets W$\\
\While{simulation is running}{
    $render(W')$\label{cline:render}\\
    $S'' \gets predict(S')$\label{cline:predict}\\
    $W'' \gets updateEntitiesPositions(W',S'') $\label{cline:update-world}\\
    $S'' \gets correctPredictions(S'',W'')$\label{cline:correct}\\
    $S' \gets S''$\\
    $W' \gets W''$\\
 }
 }
\end{algorithm}

A $render$ function (line~\ref{cline:render}) could be used for display purposes. It renders the current positions of entities, and the emotions and predictions of the agents about themselves and others.

The $predict$ function (line~\ref{cline:predict}) was the base of the algorithm. Basically, it applied the core algorithm (cf algorithm~\ref{alg:alg_1}) for each agent.
Algorithm~\ref{alg:alg_4} describes this function where $availableActions$ (line~\ref{cline:available-actions}) returns a list of available actions $A' \subset A$ for the state $s$.\\

\begin{algorithm}[H]
\label{alg:alg_4}
\caption{Prediction of next state}
  \SetKwProg{Predict}{Predict}{}{}
  \SetKwInOut{Input}{inputs}
  \SetKwInOut{Params}{parameters}
  \SetKwInOut{Output}{output}
\Predict{$(S)$}{
\Params{Set of actions $A$;}
\Input{Agents' states $S$;}
\Output{Best agents' states $S^{'*}$;}
$S^{'*} \gets \emptyset$\\
\ForEach{Agent state $s \in S$}{
    $A' \gets availableActions(s,A)$\label{cline:available-actions}\\
    $s^{'*} \gets NextState(s,A')$\\
    $S^{'*} \gets S^{'*} + s^{'*}$
 }
  \KwRet{$S^{'*}$}

 }
\end{algorithm}

$updateEntitiesPositions$ (algorithm~\ref{alg:alg_3}, line~\ref{cline:update-world}) used the predicted desired states $S^{'*}$ and tried to place the agents in such states. In short, it tried for each agent to execute the action $a$ ($a = a(s) \to s^{'*}$). The execution of $a$ in the external ambient world could succeed or fail. For example, if $a$ is "move forward X meters" and the path is clear, $a$ is then doable and the agent will move forward. However, if an obstacle that the agent did not anticipate blocks the path, $a$ will fail (or not complete). Thus, if an agent's belief about the world was close to the actual external ambient world, its actions would almost always succeed and the agent would end up in a state $s'$ that is almost equal to the predicted $s^{'*}$ ($a(s) \to s', s' \sim s^{'*}$). Otherwise, its predictions would always be bad and it would end up in states $s'$ that are far from the expected state $s^{'*}$ ($a(s) \to s', s' \nsim s^{'*}$).

$correctPredictions$ (algorithm~\ref{alg:alg_3}, line~\ref{cline:correct}) was implemented to correct agents' states. After the execution of the actions, agents were witnesses of the results of those actions. Their inner states were corrected according to what they could actually see (sensory feedback). For example, if an agent expected another agent to move and could see that it actually did not happen, it could then correct its inner state. However, an agent was unable to correct what it could not witness. For example, if something happened in its back, it would stick with its beliefs and prediction, even if they were wrong.

\subsubsection{Prediction of others and planning}

On top of that base algorithm (section~\ref{section:pcm-alg-base}), two layers were added to enable agents to predict others and to foresee or plan actions.

\bigskip

\textbf{Prediction of others}\\

As presented in algorithm~\ref{alg:alg_4}, the current state $s$ of the agent was used to determine the next best state $s^{'*}$. To be proactive, an agent could try to guess what the other agents would be doing before choosing an action. In doing so, the agent would be able (provided its assumptions were good enough) to end up in a state $s'$ that was much closer to $s^{'*}$. Algorithm~\ref{alg:alg_5} describes this process.\\

\begin{algorithm}[H]
\label{alg:alg_5}
\caption{Prediction of next state while predicting others}
  \SetKwProg{Predict}{PredictProactive}{}{}
  \SetKwInOut{Input}{inputs}
  \SetKwInOut{Params}{parameters}
  \SetKwInOut{Output}{output}
\Predict{$(s)$}{
\Params{Set of actions $A$;}
\Input{Agent state $s$;}
\Output{Best agent state $s^{'*}$;}
$s^{p} \gets s$\\
$beliefs(s^{p}) \gets \emptyset$\\
\ForEach{Agent state $s_i \in beliefs(s)$}{
    $A' \gets availableActions(s_i,A)$\\
    $s_i^{'*} \gets NextState(s_i,A')$\\
    $beliefs(s^{p}) \gets beliefs(s^{p}) + s_i^{'*}$\\
 }
 $A' \gets availableActions(s^{p},A)$\\
 $s^{'*} \gets NextState(s^{p},A')$\\
 \KwRet{$s^{'*}$}
 }
\end{algorithm}

\bigskip

\textbf{Planning}\\

Alongside trying to predict others, an agent could be given the ability to plan a chain of actions. This made, for instance, an agent more resilient when facing obstacles.

To illustrate the approach, let us assume a simple situation with an agent $a$ in a room. In this room, there is a fruit and a closed door to go outside. Outside the room there is a very large box of chocolates. If $a$ always choose the next best action without planning, it will always end up approaching and eating the fruit. That is because the free energy associated with eating the fruit is lower than the one associated with opening the door. If given the possibility to plan, it could imagine a chain of actions leading to the box of chocolate: open the door $\to$ eat the chocolate. This is presented in algorithm~\ref{alg:alg_6} and works using a tree as a data structure. In this tree, the root is the current state and other nodes are expected states modified by actions. Each node has an expected free energy associated with it. Once the tree is built, branches (from the root to a leaf) are evaluated and the best branch is picked. Note that there are many ways to pick the best branch(\ref{cline:best-branch}), for example picking the one with the lowest mean free energy (which we used for the simulations). The next best state is then defined as the first node of the best branch.\\

\begin{algorithm}[H]
\label{alg:alg_6}
\caption{Planning}
  \SetKwProg{ExploreActions}{ExploreActions}{}{}
  \SetKwInOut{Input}{inputs}
  \SetKwInOut{Output}{output}
\ExploreActions{$(s,A,d,n)$}{
\Input{Agent's current state $s$; actions $A$; depth of the tree $d$; parent node $n$}
\Output{Agent's next best state $s^{'*}$}
 $n' \gets node(s,ComputeFreeEnergy(s))$\\
 $children(n) \gets children(n) + n'$\\
 $A' \gets availableActions(s,A)$\\
 $S' \gets A'(s)$\\
 \uIf{$d > 1$}{
    \ForEach{$s' \in S'$}{
        $ExploreActions(s',A,d-1,n')$\\
    }
 }
 }
\SetKwProg{NextState}{NextState}{}{}
\NextState{$(s,A)$}{
\Input{Agent's current state $s$; actions $A$}
\Output{Agent's next best state $s^{'*}$}
 $root \gets node()$\\
 $ExploreActions(s,A,d,root)$\\
 $b \gets BestBranch(root)$\label{cline:best-branch}\\
 $s^{'*} \gets FirstNode(b)$\\
 \KwRet{$s^{'*}$}
 }

\end{algorithm}

\bigskip

\subsubsection{Note on performance}

Algorithm~\ref{alg:alg_6} is a polynomial time algorithm ($O(n^\alpha)$ where $n$ is the number of actions and $\alpha$ the tree depth). Some strategies could be used to mitigate that computational burden, for example:
\begin{itemize}
    \item[a.] Pruning the tree while building it: the deeper the three, the less branches are kept.
    \item[b.] "Blurring" while planning: instead of having a planned action corresponding exactly to the action executed in the external ambient world, the predictions could use compound actions or quicker actions. For example, movement actions could be faster while predicting ($speed(agent) < speed_{pred}(agent)$) and accelerate with tree depth. That allowed agents to cover more space by imagination with a lower depth.
\end{itemize}

It is important to note that this complexity is due to a certain immaturity of the algorithm and will be greatly reduced when a variety of heuristics will be integrated (see Discussion section). Also note that we are currently in the process of porting the algorithm to a small dedicated cluster of computers with hybrid CPU/GPU architecture. 

\bigskip

\section{(Anki) Cozmo Robot controller and spatial referential}
\label{appendix:cozmos}

In order to make the PCM more "tangible", we decided to bring it to the real world using Cozmo robots. Those robots from Anki are primarily toys connected to a smartphone. However, a complete SDK to control them is provided.

\subsection{Cozmo robots' features and SDK}

Cozmo robots are able to "see" things using a low-resolution camera, to talk (using text-to-speech) and make sounds using a speaker, and to display faces or various patterns on their front screen. They can move around using their wheels, and use their arms to interact with objects. A smartphone running the Cozmo application is required to pilot them, using the robot's WIFI network. The robots have no internal computational resources on board to run an external algorithm. 
The smartphone application can be switched to "SDK mode" which allows to run custom code. The SDK is available in Python and requires a computer connected to a smartphone to work.

\subsection{Multiple robots and common spatial reference system}
Cozmo robots are not designed to interact between each others. However, it is still possible to connect multiple smartphone connected to different Cozmo robots to a computer running external code, to control multiple robots. To be able to efficiently control the robots, we created a common coordinate system. The robots already have a internal coordinate system that they update based on the movement of their wheels. The common coordinate system was built using Aruco markers (see \cite{aruco1}, \cite{aruco2}), which the robots could process through their camera to control and update their internal positions. Once a robot could locate itself, it knew how to transform its own coordinate system into the common coordinate system, and then was able to move around, even when it would not actually see any marker.

\subsection{Interfacing Cozmo robots and the PCM}

To interface the PCM with Cozmo robots, we added a server to the python code controlling the robots. The server sent spatial information and sensory inputs from robots to the PCM. Those data could then be used to inform the internal models of the agents about the current state of the world. Knowing the states, the algorithm could run and choose the next action for each agent. Those actions could then be forwarded via the server to the robots. In short, we could map functions described in the algorithm (see Appendix \ref{appendix:algorithm}) to the action of the robots.

\subsection{Note on offline and online modes}

The mechanisms of interfacing and control presented in the previous paragraph were implemented as an online mode for real time interactions, which does work. However, due to the complexity of the algorithm, real time performance could not be achieved at this point when using higher depth of processing ($dp>2$), which was required by the simulations. 

To run simulations more effectively and create smooth videos, we decided to switch to offline mode. We ran the simulations without using the robots and replayed them later. Basically, pre-computed positional and emotional data were fed back to the server and ran on the robots. The robots were thus used for illustration purposes only.

\end{appendices}

\newpage

\bibliographystyle{ieeetr}
\bibliography{biblio}
\end{document}